\documentclass[a4paper, 11pt]{article}
\usepackage[dvipsnames]{xcolor}
\usepackage{amsmath}
\usepackage{amsthm}
\usepackage{amssymb}
\usepackage{mathcomp}
\usepackage{dsfont}
\usepackage{comment}
\usepackage{authblk}

\usepackage[active]{srcltx}
\usepackage{stackengine}
\usepackage[bookmarksnumbered=true]{hyperref}

\usepackage{geometry}
\newtheorem{definition}{Definition}[section]
\newtheorem{lemma}[definition]{Lemma}
\newtheorem{proposition}[definition]{Proposition}
\newtheorem{theorem}[definition]{Theorem}
\newtheorem{remark}[definition]{Remark}

\newtheorem{assumption}[definition]{Assumption}

\numberwithin{equation}{section}

\begin{document}
\title{An optimal lower bound for the low density Fermi gas}

\author{Emanuela L. Giacomelli}
\affil{University of Milan, Department of Mathematics, Via Cesare Saldini 50, 20133 Milan, Italy}

\maketitle

\abstract{\textcolor{black}{We consider a dilute Fermi gas in three dimensions interacting through a positive, radially symmetric, compactly supported, and integrable potential in the thermodynamic limit. We establish a second-order lower bound for the ground-state energy density, with an error term that is optimal in the sense that it matches the order of the next correction term conjectured by Huang and Yang in 1957 [Huang and Yang, Phys. Rev. 105, 767–775, 1957]. Although a first rigorous derivation of the Huang–Yang formula has recently been obtained in the combined papers [Giacomelli, Hainzl, Nam and Seiringer,arXiv:2409.17914][Giacomelli, Hainzl, Nam and Seiringer,arXiv:2505.22340], the present work takes a different approach, inspired by the construction of suitable trial states introduced in [Giacomelli, Hainzl, Nam and Seiringer,arXiv:2409.17914]. In particular, we provide a simple proof that is effective enough to yield the optimal error bound, by adapting the argument originally used to establish the upper bound on the energy density in [Giacomelli, Hainzl, Nam and Seiringer,arXiv:2409.17914].}}

\tableofcontents

\section{Introduction}
\textcolor{black}{We consider a system of $N$ interacting fermions with spin $\sigma = \{\uparrow, \downarrow\}$ confined in a box $\Lambda := [-L/2,L/2]^3$, with periodic boundary conditions. The Hamiltonian of the system acts on $\bigwedge^N L^2(\Lambda; \mathbb{C}^2)$ and is given by:
\begin{equation}\label{eq: def HN}
  H_N= -\sum_{i=1}^N\Delta_{x_i} + \sum_{i<j=1}^N V(x_i-x_j).
\end{equation}
We assume that the interaction potential $V$ satisfies the following assumptions.
 \begin{assumption}\label{asu: potential V} The interaction potential $V$ is of the form
\begin{equation}\label{eq: period V}
  V(x) = \sum_{n\in\mathbb{Z}^3}V_\infty (x+nL),
\end{equation}
where $V_\infty \in L^1(\mathbb{R}^3)$ is radial, non negative and compactly supported with $\mathrm{supp} V _\infty\subset \{x\in\mathbb{R}^3\, \vert\, |x| < R_0\}\subset \mathbb{R}^3$, for some $R_0>0$.
\end{assumption}
Since the  Hamiltonian is spin independent, the number of particles with each spin component is conserved. Therefore, we restrict our attention to wave functions of $N$ particles, where $N= N_\uparrow + N_\downarrow$, and $N_\sigma$ denotes the number of particles with spin $\sigma \in \{\uparrow, \downarrow\}$. In the following, we denote the set of these wave functions by $\mathfrak{h}(N_\uparrow, N_\downarrow)\subset \bigwedge^N L^2(\Lambda; \mathbb{C}^2)$. Our goal is to study  the ground state energy density in the dilute regime  $\rho_\sigma \rightarrow 0$, in the thermodynamic limit:
\begin{equation}
  e(\rho_\uparrow, \rho_\downarrow) = \lim_{\substack{L \rightarrow \infty \\ \frac{N_\sigma}{L^3} \rightarrow \rho_\sigma, \,\sigma\in\{\uparrow,\downarrow\}}} \frac{E_L(N_\uparrow, N_\downarrow)}{L^3},
\end{equation}
where $E_L(N_\uparrow, N_\downarrow)$ is the ground state energy, i.e., $E_{L}(N_\uparrow, N_\downarrow) = \inf_{\Psi\in\mathfrak{h}(N_\uparrow, N_\downarrow)}\langle \Psi, H_N \Psi\rangle/\langle \Psi, \Psi\rangle$. It is well known that the limit above exists and is independent of the boundary conditions (see \cite{Ro,R}). Throughout the paper, we use the notation $\rho = \rho_\uparrow + \rho_\downarrow$.}

\textcolor{black}{In recent decades, considerable effort has been devoted to rigorously establishing asymptotic formulas for the ground state energy and free energy of dilute quantum gases.  For bosonic systems, first-order asymptotics were proven in \cite{Dy,LY1,LY2,DMS,S08,Yin}, while the second-order correction, known as the Lee-Huang-Yang term, was rigorously derived in \cite{BCS, FGJMO,FS1,FS2,HHST, HHNST,YY}.}

\textcolor{black}{The situation for dilute Fermi gases is qualitatively different, as the Pauli exclusion principle gives rise to a fundamentally different energy asymptotics structure.}

\textcolor{black}{In this setting, the derivation of the Huang–Yang (HY) asymptotics \cite{HY}, conjectured in 1957, has long represented a central problem. In the  case $\rho_\uparrow = \rho_\downarrow = \rho/2 \rightarrow 0$, the HY formula predicts
\begin{equation}\label{eq: HY formula}
  e(\rho_\uparrow, \rho_\downarrow) = \frac{3}{5}(3\pi^2)^{\frac{2}{3}} \rho^{\frac{5}{3}} + 2\pi a \rho^2 + \frac{4}{35}(11 - 2\log 2) (9\pi)^{\frac{2}{3}} a^2 \rho^{\frac{7}{3}} + o(\rho^{7/3}),  \qquad \rho\rightarrow 0,
\end{equation}
where $a$ denotes the scattering length of the interaction potential. The first rigorous proof of the HY asymptotics was obtained in \cite{GHNS25}, where a matching lower bound was established to complement the upper bound derived in \cite{GHNS24}. More recently, \cite{CWZ25} proved the HY formula also in the low-temperature regime. Several earlier works contributed important intermediate results; see, for instance, \cite{LSS, FGHP, Gia1, Lau, GHNS24}.}

\textcolor{black}{We adapt a method inspired by \cite{GHNS24} to derive a lower bound on the energy asymptotics. The resulting bound is weaker than the full Huang–Yang (HY) asymptotics; nevertheless, our main goal is to show that, with a suitable adaptation of the approach of \cite{GHNS24}, one can obtain a lower bound whose error term is of the same order as the HY correction. Compared to \cite{GHNS24}, we adopt a less refined strategy, in the sense that we neglect part of the excitations around the Fermi surface, as their contribution is of order $\rho^{7/3}$. This simplification is not possible in \cite{GHNS24}, where a more delicate analysis is required.}

\textcolor{black}{Our analysis highlights that the derivation of the HY term requires a very refined treatment of excitations around the Fermi surface, which is precisely what makes the problem technically challenging. The main outcome is an energy asymptotic with an error term of the same order as the HY correction, valid for any compactly supported, positive, and integrable interaction potential. This slightly extends the class of potentials considered in \cite{GHNS24, GHNS25, CWZ25}. We also emphasize that the HY correction originates from correlations between particles of opposite spin. For single-component Fermi gases, the behavior of the ground-state energy is qualitatively different, as shown in \cite{LS1, LS3} (see also \cite{LS2, LS4} for related results at positive temperature).  A simplified version of \eqref{eq: HY formula} in the Gross–Pitaevskii regime was recently derived in \cite{CWZ24}.
The proof relies on conjugating the Hamiltonian with suitable unitary transformations, which allows one to capture the relevant quantum correlations among the particles. Such conjugations are a standard and powerful tool in the study of many-body quantum systems and have been used extensively in this context; see, for instance, \cite{FGHP, Gia1, GHNS24, CWZ25}. Here, however, this strategy leads to a considerably shorter and simpler proof than those in \cite{GHNS25, CWZ25}, although it does not recover the full HY correction.}

\textcolor{black}{The relevant correlations are described in terms of pairs of fermions, one with momentum inside the Fermi ball and one outside it. For the low-density Fermi gas, this method was first developed in \cite{FGHP, Gia1}. Often referred to as the ``bosonization method'', it was originally introduced in the 1950s by Sawada and Brueckner, Gukuda, and Brout \cite{Sa, SBFB} in the context of high-density Fermi gases. More recently, it has been rigorously applied to both the low-density regime \cite{FGHP, Gia1, Gia2, LS3, LS4, GHNS24, CWZ24, GHNS25, CWZ25} and the high-density setting \cite{BNPSS, BPSS, CHN1, CHN2, CHN3, FRS} in the study of the ground state energy of fermionic systems.} \\
\noindent\textbf{Organization of the paper.} \textcolor{black}{In Section~\ref{sec: main and strategy}, we present the main results and outline the strategy of the proof, introducing the unitary transformations used to analyze the energy asymptotics. Section~\ref{sec: useful bounds} collects several auxiliary estimates needed to control the error terms. In Section~\ref{sec: conj T1}, we conjugate the Hamiltonian with the first quasi-bosonic unitary transformation. Section~\ref{sec: T2 conjugation} contains the second conjugation, which leads to an optimal lower bound for the ground state energy density. Finally, in Section~\ref{sec: number excit}, we derive improved estimates on the number of excitations outside the Fermi ball as a consequence of the energy asymptotics.}\\ 

\section{Main results}\label{sec: main and strategy}
Using the notations introduced in the previous section, we state the main theorem of the paper.
\begin{theorem}[Optimal energy asymptotics]\label{thm: optimal lw bd} Let $V, V_\infty$ be as in Assumption \ref{asu: potential V}. As $\rho\rightarrow 0$, it holds that
\begin{equation} \label{eq: opt lw bd}
  e(\rho_\uparrow, \rho_\downarrow) = \frac{3}{5}(6\pi^2)^{\frac{2}{3}} \big(\rho_\uparrow^{\frac{5}{3}} + \rho_\downarrow^{\frac{5}{3}}\big) + 8\pi a\rho_\uparrow\rho_\downarrow + \mathcal{O}\big(\rho^{\frac{7}{3}}\big),
  \end{equation}
  where $a$ is the scattering length of the interaction potential $V_\infty$.
\end{theorem}
The theorem provides an asymptotic for the ground-state energy density with an error estimate consistent with the order of the HY correction. This result can be viewed as a completion of the analysis in \cite[Theorem 2.1]{Gia1}, where we derived an optimal upper bound for $e(\rho_\uparrow, \rho_\downarrow)$. Here, we obtain the complementary optimal lower bound, inspired by the method later developed in \cite{GHNS24}, which employs two quasi-bosonic unitary operators to isolate the relevant correlations at order $\rho^{7/3}$. Note that in this paper we slightly enlarge the class of interaction potentials compared to those considered in \cite{Gia1} (as well as in \cite{FGHP, GHNS24, GHNS25, CWZ25}). 
As a consequence of Theorem \ref{thm: optimal lw bd}, we can refine the estimate of excitations outside the Fermi ball, obtaining a bound in agreement with that of \cite[Proposition 2.4]{GHNS25}. Before presenting this result, we introduce some additional notations.\\

\noindent\textbf{Fermionic Fock space}. Let $\mathcal{F}_{\mathrm{f}}$ be the fermionic Fock space, $\mathcal{F}_{\mathrm{f}} = \bigoplus_{n\geq 0} \mathcal{F}^{(n)}_{\mathrm{f}}$, with $\mathcal{F}^{(n)}_{\mathrm{f}} =  L^2(\Lambda; \mathbb{C}^2)^{\wedge n}$ and $\mathcal{F}^{(0)}_{\mathrm{f}} = \mathbb{C}$. We denote by $\Omega$ the vacuum. Furthermore, for any $f\in L^2(\Lambda; \mathbb{C}^2)$, i.e., $f = (f_\uparrow, f_\downarrow)$, we denote by $a^\ast(f)$ (resp. $a(f)$) the fermionic creation/annihilation operators, which satisfy canonical anti-commutation relations. For more details we refer to \cite{FGHP,Gia1,Gia2}. As it is well-known, the fermionic creation/annihilation operators are bounded, i.e., 
\[
  \|a(f)\|, \|a^\ast(f)\| = \|f\|_{L^2(\Lambda; \mathbb{C}^2)}.
\] 
In our analysis, we often use 
\[
  \hat{a}_{k,\sigma} = a_\sigma(f_k) = a(\delta_\sigma f_k) = \frac{1}{L^{\frac{3}{2}}}\int_{\Lambda}dx\, a_{x,\sigma}e^{-ik\cdot x}, \qquad (\hat{a}_{k,\sigma})^\ast = \hat{a}^\ast_{k,\sigma},
\]
where $(\delta_\sigma f_k)(x,\sigma^\prime)\in L^2(\Lambda; \mathbb{C}^2) $ and $(\delta_\sigma f_k)(x,\sigma^\prime) = \delta_{\sigma,\sigma^\prime}e^{ik\cdot x}/L^{\frac{3}{2}}$.
Formally, we can express the operator $a_{x,\sigma}$ as $a(\delta_{x,\sigma})$ with $\delta_{x,\sigma}(y,\sigma^\prime) = \delta_{\sigma,\sigma^\prime}\delta(x-y)$. Here, $\delta_{\sigma,\sigma^\prime}$ is the Kronecker delta, and
\[
  \delta(x-y) = \frac{1}{L^3}\sum_{k\in\frac{2\pi}{L}\mathbb{Z}^3} e^{ik\cdot (x-y)}.
\]
We denote the number operator as 
\[
  \mathcal{N} = \sum_{\sigma\in \{\uparrow, \downarrow\}}\mathcal{N}_\sigma = \sum_{\sigma\in \{\uparrow, \downarrow\}} \hat{a}_{k,\sigma}^\ast \hat{a}_{k,\sigma} = \sum_{\sigma\in \{\uparrow, \downarrow\}}\int dx\, a^\ast_{x,\sigma}a_{x,\sigma}.
\]
The second-quantized Hamiltonian is explicitly written as
\[
  \mathcal{H} = \sum_{\sigma \in \{\uparrow, \downarrow\}}\sum_{k\in\frac{2\pi}{L}\mathbb{Z}^3} |k|^2 \hat{a}^\ast_{k,\sigma}\hat{a}_{k,\sigma} + \frac{1}{2L^3}\sum_{\sigma,\sigma^\prime\in \{\uparrow, \downarrow\}} \sum_{k,p,q}\hat{V}(k)\hat{a}_{p+k,\sigma}^\ast\hat{a}_{q-k,\sigma^\prime}^\ast\hat{a}_{q,\sigma^\prime}\hat{a}_{p,\sigma}.
\]
Denoting by $\mathcal{F}_{\mathrm{f}}(N_\uparrow, N_\downarrow)\subset \mathcal{F}_{\mathrm{f}}$ the subset of $\mathcal{F}_{\mathrm{f}}$ given by $N-$particles states with $N = N_\uparrow + N_\downarrow$, the ground state energy of the system can be written as 
\[
  E_L(N_\uparrow, N_\downarrow)= \inf_{\psi \in \mathcal{F}_{\mathrm{f}}^{(N_\uparrow, N_\downarrow)}}\frac{\langle \psi, \mathcal{H}\psi\rangle}{\langle\psi, \psi\rangle}.
\]

We can now explicitly state the improved  estimates for the number of excitations.
\begin{theorem}[Excitation estimates]\label{thm: optimal number operator}
Let $\psi$ be a normalized $N-$particle fermionic state, with  $N = N_\uparrow + N_\downarrow$ such that
\begin{equation}\label{eq: new app gs}
  \left|\frac{\langle \psi,\mathcal{H}\psi\rangle}{L^3} - \frac{3}{5}(6\pi^2)^{\frac{2}{3}} (\rho_\uparrow^{\frac{5}{3}} + \rho_\downarrow^{\frac{5}{3}}) - 8\pi a\rho_\uparrow\rho_\downarrow \right|  \leq C\rho^{\frac{7}{3}}.
\end{equation}
Under the same assumptions as in Theorem \ref{thm: optimal lw bd}, it holds that
\begin{equation}\label{eq: optimal bound number op}
  \sum_{\sigma = \{\uparrow, \downarrow\}}\sum_{\substack{k\in\frac{2\pi}{L}\mathbb{Z}^3 \\ |k|> k_F^\sigma}}\langle \psi, \hat{a}_{k,\sigma}^\ast \hat{a}_{k,\sigma}\psi\rangle  =  \sum_{\sigma = \{\uparrow, \downarrow\}}\sum_{\substack{k\in\frac{2\pi}{L}\mathbb{Z}^3 \\ |k| \leq k_F^\sigma}}\langle \psi, \hat{a}_{k,\sigma} \hat{a}_{k,\sigma}^\ast \psi\rangle \leq CL^3\rho^{\frac{4}{3}}.
\end{equation}
Furthermore, let $0\leq\epsilon<1$ and 
\[
  \mathcal{N}_>^{(\epsilon)} = \sum_{\sigma = \{\uparrow, \downarrow\}}\sum_{|k| > k_F^\sigma + (k_F^\sigma)^{1+\epsilon}}\hat{a}_{k,\sigma}^\ast \hat{a}_{k,\sigma}.
\]
Under the same assumptions as above, it holds that
\[
  \langle \psi, \mathcal{N}_>^{(\epsilon)} \psi\rangle \leq CL^3\rho^{\frac{5}{3}  - \frac{\epsilon}{3}}.
\]  \end{theorem}

\begin{remark}[Hard core interaction]\label{rem: lw bd hardcore} It is interesting to extend the energy density asymptotics discussed in Theorem \ref{thm: optimal lw bd} to the hard core potential $V_{\mathrm{hc}}$, i.e, 
\begin{equation}\label{eq: hc int}
  V_{\mathrm{hc}}(x) = \begin{cases} \infty &\mbox{if}\,\,\, |x| \leq 1, \\ 0 &\mbox{if}\,\,\, |x| >1.\end{cases}
\end{equation}
 This type of interaction potential is indeed included in the analysis in \cite{LSS}.  Following the approximation techniques in \cite[Lemma 3.3, Remark 3.4]{FS2}, the hard-core potential $V_{\mathrm{hc}}$ can be approximated using integrable $L^1$ interaction potentials. By employing these approximations, the proof of Theorem \ref{thm: optimal lw bd} can be adapted (only at the level of a lower bound) to handle interactions of the form given by \eqref{eq: hc int}. However, it is important to note that in this case, we do not obtain an optimal lower bound.
\end{remark}

\subsection{Unitary transformations}

\textcolor{black}{In this section, we introduce the unitary transformations that will be used to prove Theorem~\ref{thm: optimal lw bd}. Our focus will be on the proof of the lower bound matching \eqref{eq: opt lw bd}; the main ideas underlying the corresponding upper bound are briefly discussed in Section~\ref{sec: up bd}.}

\textcolor{black}{The section is organized as follows. In Section~\ref{sec: particle-hole}, we introduce the particle–hole transformation, which allows us to isolate the correlation energy, defined as the difference between the ground state energy and the free Fermi gas (FFG) energy. In Section~\ref{sec: def T1}, we present the first quasi-bosonic unitary transformation, which is used to describe low-energy excitations around the Fermi ball. A closely related transformation was already employed in \cite{Gia1}, and it is sufficient, both here and in \cite{Gia1}, to obtain an optimal upper bound (see Section~\ref{sec: up bd}). Finally, in Section~\ref{sec: T2 unitary}, we introduce a second quasi-bosonic unitary transformation, inspired by the one used in \cite{GHNS24}, which is essential for deriving the optimal lower bound.}

\subsubsection{Particle-hole transformation}\label{sec: particle-hole}
In this section, we recall the definition of the \textit{particle-hole transformation}. To start, we fix the notation for the free Fermi gas. We write 
\begin{equation}\label{eq: FFG}
    \psi_{\mathrm{FFG}} = \prod_{\sigma\in \{\uparrow, \downarrow\}} \prod_{k\in \mathcal{B}_F^\sigma} \hat{a}_{k,\sigma}^\ast\Omega,
\end{equation}
  with the Fermi ball given by
\[
 \mathcal{B}_F^\sigma := \big\{ k\in (2\pi/L)\mathbb{Z}^3\, \,\,\vert\,\,\, |k| \leq k_F^\sigma\big\}.
\]
A direct calculation, see \cite{FGHP, Gia1}, gives (taking $L$ large enough)
  \begin{equation}\label{eq: FFG energy}
    E_{\mathrm{FFG}} = \langle \psi_{\mathrm{FFG}}, \mathcal{H}_N \psi_{\mathrm{FFG}}\rangle  = \frac{3}{5}(6\pi^2)^{\frac{2}{3}}\left(\rho_\uparrow^{\frac{5}{3}} + \rho_{\downarrow}^{\frac{5}{3}}\right) L^3+\hat{V}(0)\rho_\uparrow\rho_\downarrow L^3+ \mathcal{O}(L^3\rho^{\frac{8}{3}}).
\end{equation}
\begin{remark}[Filled Fermi sea]
We suppose the Fermi ball to be completely filled, i.e., we only consider the values of $N_\sigma$ such that $N_\sigma = |\mathcal{B}_F^\sigma|$. This is not a loss of generality since the corresponding densities $\rho_\sigma$ form a dense subset of $\mathbb{R}_+$ in the large volume limit. See \cite[Section 3]{FGHP} and \cite[Remark 1.7]{LS3} for more details.
\end{remark}
\begin{definition}[Particle-hole transformation]\label{def: ferm bog}
Let $u$,$v: L^2(\Lambda;\mathbb{C}^2)\rightarrow L^2(\Lambda;\mathbb{C}^2)$ be the operators with integral kernels given by 
\begin{equation} \label{eq: def u,v}
  v_{\sigma,\sigma^\prime} (x;y)  =  \frac{\delta_{\sigma,\sigma^\prime}}{L^3}\sum_{k\in\mathcal{B}_{F}^\sigma} e^{ik\cdot(x-y)}, \qquad u_{\sigma, \sigma^\prime}(x;y) = \frac{\delta_{\sigma,\sigma^\prime}}{L^3} \sum_{k\notin\mathcal{B}_F^{\sigma}}e^{ik\cdot(x-y)}.
\end{equation}
The particle-hole transformation is a unitary operator $R: \mathcal{F}\rightarrow\mathcal{F}$ such that the following properties hold: 
\begin{itemize}
  \item[(i)] The state $R\Omega$ is such that $(R\Omega)^{(n)} = 0$ whenever $n\neq N$ and $(R\Omega)^{(N)} = \Psi_{\mathrm{FFG}}$.
  \item[(ii)] It holds that
  \begin{equation}\label{eq: prop II R}
    R^\ast a_{x,\sigma}^\ast R = a_\sigma^\ast(u_x) + a_\sigma( v_x),
  \end{equation}
  where
  \[
    a_\sigma^\ast(u_x) = \int_{\Lambda} dy\, u_\sigma(y;x)a_{y,\sigma}^\ast, \qquad a_\sigma( v_x) = \int_{\Lambda} dy\,v_\sigma(y;x) a_{y,\sigma}.
  \] 
\end{itemize}
\end{definition}
\textcolor{black}{Note that equivalently, the transformation $R$ can be defined by the condition $(i)$ above and:
  \begin{equation}\label{eq: def R momentum space}
  R^\ast \hat{a}_{k,\sigma} R = \begin{cases} \hat{a}_{k,\sigma} &\mbox{if}\,\,\, k\notin\mathcal{B}_F^\sigma, \\ \hat{a}^\ast_{-k,\sigma} &\mbox{if}\,\,\, k\in \mathcal{B}_F^\sigma. \end{cases}
\end{equation} 
The next proposition\footnote{Note that our notation for the operators in $\mathbb{Q}$ differs from that used in \cite{FGHP, Gia1}.} allows us to isolate the correlation energy  by conjugating $\mathcal{H}$ with $R$. Here and in the following we use the notation $f_x = f(\cdot - x)$ (note that $u_x = u(\cdot\,; x) = u(\cdot - x)$ and $v_x = v(\cdot\,; x) = v(\cdot - x)$).}
\begin{proposition}[Conjugation by $R$]\label{pro: fermionic transf}
  Let $V$ be as in Assumption \ref{asu: potential V}. Let $\psi\in \mathcal{F}_{\mathrm{f}}$ be a normalized state, such that $\mathcal{N}_\sigma\psi = N_\sigma\psi$ and $N = N_\uparrow + N_\downarrow$. Then 
  \begin{equation} 
    \langle \psi, \mathcal{H}\psi\rangle \geq E_{\mathrm{FFG}} + \langle R^\ast\psi, \mathbb{H}_0 R^\ast\psi\rangle + \langle R^{\ast} \psi, \mathbb{Q}R^\ast\psi\rangle,
  \end{equation}
  where $E_{\mathrm{FFG}}$ is the energy of the free Fermi in \eqref{eq: FFG energy}. The operator $\mathbb{H}_0$ is given by 
\begin{equation}\label{eq: def H0}
  \mathbb{H}_0 = \sum_\sigma\sum_k ||k|^2 -(k_F^\sigma)^2| \hat{a}_{k,\sigma}^\ast \hat{a}_{k,\sigma}.
\end{equation}
The operator $\mathbb{Q}$ can be written as $\mathbb{Q} = \sum_{i=1}^4\mathbb{Q}_i$, with 
\begin{eqnarray}\label{eq: def Qi}
  \mathbb{Q}_4 &=& \frac{1}{2}\sum_{\sigma\neq \sigma^\prime}\int dxdy\, V(x-y) \left( {a}^\ast_\sigma(u_x){a}^\ast_{\sigma^\prime}(u_y){a}_{\sigma^\prime}(u_y){a}_\sigma(u_x)\right),
  \\
  \mathbb{Q}_1 &=& \sum_{\sigma\neq \sigma^\prime}\int dxdy\, V(x-y) a^\ast_\sigma(u_x)a^\ast_{\sigma}( v_x)a_{\sigma^\prime}( v_y)a_{\sigma^\prime}(u_y),\nonumber
  \\
  && + \frac{1}{2}\sum_{\sigma\neq\sigma^\prime}\int dxdy\, V(x-y) \left(a^\ast_\sigma( v_x)a^\ast_{\sigma^\prime}( v_y)a_{\sigma^\prime}( v_y)a_\sigma( v_x) -2a^\ast_\sigma(u_x)a^\ast_{\sigma^\prime}( v_y)a_{\sigma^\prime}( v_y)a_\sigma(u_x)\right)\nonumber,
 \\
  \mathbb{Q}_3 &=& -\sum_{\sigma\neq\sigma^\prime} \int dxdy\, V(x-y)\left(a^\ast_\sigma(u_x) a^\ast_{\sigma^\prime}(u_y) a^\ast_{\sigma}( v_x)a_{\sigma^\prime}(u_y) - a^\ast_\sigma(u_x) a^\ast_{\sigma^\prime}( v_y) a^\ast_{\sigma}( v_x)a_{\sigma^\prime}( v_y)\right) + \mathrm{h.c.}\nonumber,
  \\
  \mathbb{Q}_2 &=& \frac{1}{2}\sum_{\sigma\neq\sigma^\prime} \int dxdy\, V(x-y) a^\ast_\sigma(u_x)a^\ast_{\sigma^\prime}(u_y)a^\ast_{\sigma^\prime}( v_y)a^\ast_{\sigma}( v_x) + \mathrm{h.c.}\nonumber
\end{eqnarray}
Furthermore, we also have 
\begin{equation}\label{eq: EFFG Qsigmasigma'}
  \langle \psi, \mathcal{H}\psi\rangle = E_{\mathrm{FFG}} + \langle R^\ast \psi, \mathbb{H}_0 R^\ast \psi\rangle + \langle R^\ast \psi, \mathbb{Q}^{\sigma,\sigma^\prime}R^\ast\psi\rangle,
\end{equation}
where $\mathbb{Q}^{\sigma, \sigma^\prime} = \sum_{i=1}^4 \mathbb{Q}^{\sigma, \sigma^\prime}_i$ and each $\mathbb{Q}^{\sigma,\sigma^\prime}_i$ is defined from the corresponding $\mathbb{Q}_i$ by replacing $\sum_{\sigma\neq \sigma^\prime}$ with $\sum_{\sigma,\sigma^\prime}$.
\end{proposition}
\textcolor{black}{The proof of the above proposition follows from a straightforward computation based on the decomposition \eqref{eq: prop II R}. We refer to \cite{FGHP} for further details and related references.} Note that the fact that we can ignore the interaction between particles with the same spin is a consequence of the positivity of the interaction potential (see \cite{FGHP}) and will be used in the discussion of the lower bound. In Section \ref{sec: up bd}, where we discuss the main ideas for the upper bound, we will instead use \eqref{eq: EFFG Qsigmasigma'}.

Before proceeding further, we collect some a priori estimates valid for any approximate ground state. 
\begin{definition}[Approximate ground state]\label{def: approx gs} Let $\psi\in \mathcal{F}_{\mathrm{f}}$ be a normalized $N$--particle state with $N = N_\uparrow + N_\downarrow$. We say that $\psi$ is an approximate ground state of $\mathcal{H}$ if 
\[
  \left| \langle \psi, \mathcal{H}\psi\rangle - \sum_{\sigma = \uparrow, \downarrow}\sum_{k\in\mathcal{B}_F^\sigma} |k|^2\right|\leq CL^3\rho^2.
\]
\end{definition}
\begin{lemma}[A priori bound for $\mathbb{H}_0, \mathbb{Q}_4, \mathbb{Q}_4^{\sigma, \sigma^\prime}, \mathcal{N}$]\label{lem: a priori est} Let $\psi$ be an approximate ground state. Under the same assumptions as in Theorem \ref{thm: optimal lw bd}, we have 
\[
  \langle R^\ast\psi,\mathbb{H}_0 R^\ast \psi\rangle \leq CL^3\rho^2, \quad \langle R^\ast \psi, \mathbb{Q}_4 R^\ast \psi\rangle\leq CL^3\rho^2, \quad\langle R^\ast \psi, \mathbb{Q}_4^{\sigma,\sigma^\prime}R^\ast \psi\rangle\leq CL^3\rho^2, \quad  \langle R^\ast \psi, \mathcal{N}R^\ast \psi\rangle \leq CL^3 \rho^{\frac{7}{6}}.
\]
\end{lemma}
 For the proof of Lemma \ref{lem: a priori est}, we refer to \cite[Lemma 3.5, Lemma 3.9, Corollary 3.7]{FGHP}. Via the a priori bounds, one can easily prove (see \cite[Proposition 3.3]{FGHP}) that for any $\psi\in \mathcal{F}$, it holds that
\begin{equation}\label{eq: est Q1}
  |\langle \psi, \mathbb{Q}_1  \psi\rangle |\leq C\rho\langle  \psi, \mathcal{N} \psi\rangle.
\end{equation}
The estimate above will guarantee that $\mathbb{Q}_1$ is sub-leading (see Section \ref{sec: lw bd T1}). Furthermore, proceeding similarly as in \cite[Section 7.2]{Gia1}, we will also prove that for any approximate ground state $\psi$, it holds that
\begin{equation}\label{eq: est Q3}
  |\langle R^\ast \psi, \mathbb{Q}_3 R^\ast \psi\rangle | \leq CL^3\rho^{\frac{7}{3}}, 
\end{equation}
which implies that the effective correlation energy is given by 
\begin{equation}\label{eq: eff corr en T1}
  \mathcal{H}_{\mathrm{corr}}^{\mathrm{eff}} = \mathbb{H}_0 + \mathbb{Q}_2 + \mathbb{Q}_4.
\end{equation}
\subsubsection{First quasi-bosonic transformation}\label{sec: def T1}
In this section, we introduce the first quasi-bosonic transformation that will be used throughout our analysis. As already mentioned, this unitary operator is closely related to the one introduced in \cite[Sections 2.1 and 4]{Gia1}.

We begin by defining the quasi-bosonic creation and annihilation operators, which provide an effective description of the low-energy excitations around the Fermi ball. These operators are constructed from pairs of fermionic creation and annihilation operators, with one particle carrying momentum inside the Fermi ball and the other outside. More precisely, we define
\begin{equation}\label{eq: def bp bpr}
  b_{p,\sigma} = \sum_{k\in\frac{2\pi}{L}\mathbb{Z}^3} b_{p,k,\sigma} = \sum_{k\in\frac{2\pi}{L}\mathbb{Z}^3}\hat{u}_\sigma(k+p)\hat{v}_\sigma(k)\hat{a}_{k+p,\sigma}\hat{a}_{-k,\sigma}, \qquad b_{p,\sigma}^\ast = (b_{p,\sigma})^\ast,
\end{equation}
with $\hat{u}_\sigma,\hat{v}_\sigma$ being the Fourier coefficients of the kernels introduced in \eqref{eq: def u,v}:
\begin{equation}\label{eq: def u hat, v hat}
\hat{u}_\sigma(k) = \begin{cases}
  0 &\mbox{if}\,\,\, |k| \leq k_F^\sigma, \\ 1 &\mbox{if}\,\,\, |k| > k_F^\sigma,
  \end{cases}\qquad \hat{v}_\sigma(k) = \begin{cases} 1 &\mbox{if}\,\,\, |k| \leq k_F^\sigma, \\ 0 &\mbox{if}\,\,\, |k| > k_F^\sigma. \end{cases}
\end{equation}
Similarly as in \cite{Gia1}, in order to extract the $8\pi a \rho_\uparrow \rho_\downarrow$ correction, it is enough to use the almost bosonic operators with the $\hat{u}_\sigma(\cdot)$ replaced by $\hat{u}^>_\sigma(\cdot)$ given by 
\begin{equation}\label{eq: def u>}
  \hat{u}^>_\sigma(k) := \hat{u}_\sigma(k)\widehat{\eta}^>_{\sigma}(k),
\end{equation}
where $\widehat{\eta}^>_\sigma$ is the periodization of a smooth function from $\mathbb{R}^3$ to $\mathbb{R}$ such that 
\begin{equation}
  \mathcal{F}(\eta^>_\sigma)(k) = \begin{cases} 0 &\mbox{if}\,\,\, |k| < 2k_F^\sigma, \\ 1 &\mbox{if}\,\,\, |k| \geq 3k_F^\sigma,\end{cases}
\end{equation}
here $\mathcal{F}$ denotes here the Fourier transform in $\mathbb{R}^3$.
\begin{remark}[Comparison with \cite{Gia1}]\label{rem: comparisono Gia1 cut-off}
Unlike in \cite{Gia1}, we do not employ a cutoff function at high momenta (of order $\rho^{-\beta}$ for some $\beta >0$) in the definition of $\hat{u}^>$. This simplifies the analysis and, more importantly, allows us to work with a larger class of interaction potentials. As a consequence, in the estimate of many error terms we often write $\hat{u}^>_\sigma = 1 - \hat{\alpha}^<_\sigma$, with $\hat{\alpha}^<_\sigma$  being the periodization of a smooth function $\mathbb{R}^3\rightarrow\mathbb{R}$ such that
\begin{equation}\label{eq: def alpha<}
  \mathcal{F}({\alpha}^<_\sigma)(k) = \begin{cases} 1 &\mbox{if}\,\,\, |k| < 2k_F^\sigma, \\ 0 &\mbox{if}\,\,\, |k| \geq 3 k_F^\sigma. \end{cases}
\end{equation}
This implies $u^>_\sigma = \delta - \alpha^<_\sigma$. It will be convenient to use that $\|\alpha^<_\sigma\|_{L^1(\Lambda)} \leq C$, which can be proven as in \cite[Proposition 4.1]{Gia1}. We also emphasize that, in contrast to \cite{Gia1}, we do not localize $\hat{v}_\sigma$.
\end{remark}
We now define the family of unitary transformations that will be used in our analysis. These transformations are quadratic in the quasi-bosonic operators. For any $\lambda \in \mathbb{R}$, we define
\begin{eqnarray}\label{eq: def T1}
  T_{1;\lambda}&:= &\exp\left(\frac{\lambda}{L^3}\sum_{p,r,r^\prime\in \frac{2\pi}{L}\mathbb{Z}^3}\hat{\varphi}(p)b_{p,r,\uparrow}b_{-p,r^\prime, \downarrow}\widehat{\eta}_>(r+p)\widehat{\eta}_>(r^\prime - p) - \mathrm{h.c.}\right)\nonumber
  \\
  &=& \exp\left(\lambda\int_{\Lambda_L \times \Lambda_L}dzdz^\prime \, \varphi(z-z^\prime) a_\uparrow(u_{z}^>) a_\uparrow( v_z) a_\downarrow (u_{z^\prime}^>)a_\downarrow( v_{z^\prime}) - \mathrm{h.c.}\right) =: \exp(\lambda(B_1 - B_1^\ast)), 
\end{eqnarray}
where
\[
  \hat{\varphi}(p) = \int_{\Lambda} dx \, \varphi(x)e^{-ip\cdot x}.
\]
Here $\varphi(x)$ denotes the periodization of a localized version of the solution to the zero energy scattering equation in $\mathbb{R}^3$, as introduced in \cite[Eq. (2.11)--(2.13)]{Gia1}. For completeness, we recall its definition. We  write 
\begin{equation}\label{eq: def phi}
  \varphi(x) = \sum_{n\in\mathbb{Z}^3}\varphi_\infty(x + nL),
\end{equation}
or, equivalently,
\[
  \varphi(x) = \frac{1}{L^3}\sum_{p\in\frac{2\pi}{L}\mathbb{Z}^3}\hat{\varphi}_\infty(p)e^{ip\cdot x},\qquad \hat{\varphi}_\infty(p) = \int_{\mathbb{R}^3} dx\, \varphi_\infty(x) e^{-ip\cdot x}.
\]
The function $\varphi_\infty$ is defined as $\varphi_\infty = \varphi_0 \chi_{\sqrt[3]{\rho}}$, where $\varphi_0$ is the solution of the zero energy scattering equation in $\mathbb{R}^3$
\begin{equation}\label{eq: scatter eq phi0}
 2\Delta\varphi_0 + V_\infty(1-\varphi_0) = 0 \qquad \mbox{in}\,\,\,\mathbb{R}^3, \quad \mbox{with}\quad \varphi_0(x)\rightarrow 0 \quad \mbox{as}\,\,\, |x| \rightarrow \infty,
\end{equation}
and $\chi$ is a smooth cut-off function satisfying  $0\leq \chi\leq 1$ and
\begin{equation}\label{eq: def chi}
  \chi(x) = \chi(|x|) = \begin{cases} 1 &\mbox{if}\,\,\, |x| \leq 1 \\ 0, &\mbox{if}\,\,\, |x| \geq 2, \end{cases} \qquad \chi_{\sqrt[3]{\rho}}(\cdot) := \chi(\cdot /\rho^{1/3}).
\end{equation}

Because of the localization introduced by $\chi$,  the function $\varphi_\infty$ does not solve the zero energy scattering equation exactly. Instead, it satisfies (see \cite[Section 4, Eq. (4.7)]{Gia1})
\[
  2\Delta\varphi_\infty + V_\infty(1-\varphi_\infty) = \mathcal{E}_{\{\varphi_0, \chi_{\sqrt[3]{\rho}}\}}^\infty,
\]
where
\begin{equation}\label{eq: def scatt err infty}
  \mathcal{E}_{\{\varphi_0, \chi_{\sqrt[3]{\rho}}\}}^\infty(x) = -4\nabla \varphi_0(x) \nabla \chi_{\sqrt[3]{\rho}}(x) - 2\varphi_0(x) \Delta\chi_{\sqrt[3]{\rho}}(x).
\end{equation} 
We denote by $\mathcal{E}_{\{\varphi_0, \chi_{\sqrt[3]{\rho}}\}}(x)$ its periodization to the box $\Lambda$, i.e., 
\begin{equation}\label{eq: period error scattering}
  \mathcal{E}_{\{\varphi_0, \chi_{\sqrt[3]{\rho}}\}}(x) = \sum_{n\in\mathbb{Z}^3}\mathcal{E}_{\{\varphi_0, \chi_{\sqrt[3]{\rho}}\}}^\infty(x + nL).
\end{equation}
The following lemma collects several useful properties of $\varphi$, which we will be repeatedly used to estimate the error terms.
\begin{lemma}[Bounds for $\varphi$]\label{lem: bound phi} Let $V$ be as in Assumption \ref{asu: potential V}. Let $\varphi$ as in \eqref{eq: def phi}. Taking $L$ large enough, the following holds
\begin{itemize}
  \item[(i)] For all $x\in \Lambda$, there exists $C>0$ such that
  \begin{equation}\label{eq: uniform norm phi >}
   \qquad|\varphi(x)| \leq C.
  \end{equation}
  \item[(ii)] It holds that
\begin{equation}
  \|\varphi\|_{L^2(\Lambda)} \leq C\rho^{-\frac{1}{6}},\qquad \|\nabla\varphi\|_{L^2(\Lambda)}\leq C,
  \end{equation}
 and
  \begin{equation}  \|\varphi\|_{L^1(\Lambda)} \leq C\rho^{-\frac{2}{3}}, \qquad \|\nabla\varphi\|_{L^1(\Lambda)} \leq C\rho^{-\frac{1}{3}}, \qquad \|\Delta\varphi\|_{L^1(\Lambda)} \leq C. 
  \end{equation}
\end{itemize}
\end{lemma}
\textcolor{black}{The proof of Lemma \ref{lem: bound phi} can be found in \cite[Lemma 4.6]{Gia1}.}
\begin{remark}Compared to the unitary operator used in \cite{GHNS24} (see \cite[Eqs.(2.18)]{GHNS24}), we localize the solution of the scattering equation in configuration space (see\eqref{eq: def phi}), rather than introducing a momentum cut-off.
This choice simplifies the analysis of $\varphi$ and allows for an explicitly control of the localization error defined in~\eqref{eq: def scatt err infty}.
Since we are only interested in terms up to order~$\mathcal{O}(\rho^{7/3})$, this approach is particularly convenient.
However, because $\varphi$ is not localized in momentum space, we must use $\hat{u}^>_\sigma$ in place of~$\hat{u}_\sigma$ in $T_{1;\lambda}$.
In~\cite{GHNS24}, a  similar restriction was implicitly enforced through the momentum cut-off applied to the scattering solution, so that no such replacement was necessary.
\end{remark}
\subsubsection{Second quasi-bosonic transformation}\label{sec: T2 unitary}
We denote by $T_2$ the second quasi-bosonic unitary transformation we are going to use. We will use it only for momenta close to $k_F^\sigma$. More precisely, after the conjugation of the correlation energy under $T_1$, we are left with a new effective correlation energy:  
\[
  \widetilde{\mathcal{H}}_{\mathrm{corr}}^{\mathrm{eff}} := \mathbb{H}_0 + \widetilde{\mathbb{Q}}_>,
\]
which we will conjugate under $T_2$. Here $\widetilde{\mathbb{Q}}_>$ is defined as 
\begin{equation}\label{eq: def Qtilde>}
  \widetilde{\mathbb{Q}}_{>} = \int_{\Lambda \times \Lambda} dxdy\, V_\varphi(x-y)a_\uparrow(u^\gg_x) a_\uparrow( v_x^\gg)a_\downarrow(u^\gg_y)a_\downarrow( v_y^\gg) + \mathrm{h.c.},
\end{equation}
with $V_\varphi:= V(1-\varphi)$. Also, $u^{\gg}_\sigma$ and $v^\gg_\sigma$ have Fourier coefficients, $\hat{u}^\gg_\sigma$ and $\hat{v}^\gg_\sigma$,  defined as 
\begin{equation}\label{eq: def v>> u>>}
\hat{v}_\sigma^\gg(k):= \hat{v}_\sigma(k) \widehat{\eta}^\gg_\sigma(k), \qquad \hat{u}^\gg_\sigma(k) := \hat{u}^<_\sigma(k)\widehat{\chi}^\gg_\sigma(k),
\end{equation}
with $\hat{u}^>_\sigma = 1-\hat{u}^<_\sigma$, $\hat{v}_\sigma$ and $\hat{u}^>_\sigma$ as in \eqref{eq: def u,v} and \eqref{eq: def u>}, respectively, and $\widehat{\chi}^{\gg}_\sigma$, $\widehat{\eta}^{\gg}_\sigma$ being the periodization of \footnote{Note that the choice of the exponent $3/2$ turns out to be optimal.}
\[
  \mathcal{F}({\chi}^{\gg}_{\sigma})(k) := \begin{cases} 0 &\mbox{if}\,\,\, k_F^\sigma \leq |k| \leq k_F^\sigma + (k_F^\sigma)^{\frac{3}{2}}, 
  \\
  1 &\mbox{if}\,\,\, k_F^\sigma + (k_F^\sigma)^{\frac{3}{2}} < |k| \leq 3k_F^\sigma
  \end{cases}, \qquad \mathcal{F}({\eta}^{\gg}_{\sigma})(k) = \begin{cases} 0 &\mbox{if}\,\,\, k_F^\sigma - (k_F^\sigma)^{\frac{3}{2}}  \leq |k| \leq k_F^\sigma,
 \\
 1 &\mbox{if}\,\,\,  |k| < k_F^\sigma - (k_F^\sigma)^{\frac{3}{2}},
  \end{cases}
\] 
where we recall that $\mathcal{F}$ denotes the Fourier transform in $\mathbb{R}^3$.
We then introduce new quasi-bosonic operators:
\begin{equation}\label{eq: def b tilde}
  \widetilde{b}_{p,\sigma} = \sum_{k\in\frac{2\pi}{L}\mathbb{Z}^3} \widetilde{b}_{p,k,\sigma} = \sum_{k\in\frac{2\pi}{L}\mathbb{Z}^3} \hat{u}_\sigma^{\gg}(k+p)\hat{v}_\sigma^\gg(k)\hat{a}_{k+p,\sigma}\hat{a}_{-k,\sigma}, \qquad \widetilde{b}^\ast_{p,\sigma} = \left(\widetilde{b}_{p,\sigma}\right)^\ast,
\end{equation}
and, for any $\lambda\in \mathbb{R}$, we define
\begin{equation}\label{eq: def T2 unitary}
  T_{2;\lambda}:= \exp\left(\frac{\lambda}{L^3}\sum_{p,r,r^\prime} \hat{f}_{r,r^\prime}(p)\, \widetilde{b}_{p,r,\uparrow}\widetilde{b}_{-p,r^\prime, \downarrow} - \mathrm{h.c.} \right) =: \exp(\lambda (B_2 - B_2^\ast)),
\end{equation} 
with
\begin{equation}\label{eq: def frr'}
  \hat{f}_{r,r^\prime}(p) = \frac{2\widehat{V_\varphi}(p)}{|r+p|^2 -|r|^2 + |r^\prime - p|^2 - |r^\prime|^2}, \qquad \widehat{V_\varphi}(p) = \int_{\Lambda_L} dx\, V_\varphi(x)\, e^{ip\cdot x}.
\end{equation}
The choice of $\hat{f}_{r,r^\prime}$ ensures that $(1/2)[\mathbb{H}_0, B_2 - B_2^\ast] = \widetilde{\mathbb{Q}}_>$, which allows us to estimate the effective correlation energy near the Fermi surface.

In the estimates of the error terms, it will be useful to use the expression of $T_{2;\lambda}$ in configuration space, which we now discuss. \textcolor{black}{Proceeding as in  \cite{GHNS24},  for $|r+p|> k_F^\uparrow > |r|$ and for $|r^\prime - p|> k_F^\downarrow > |r^\prime|$, we can write}
  \begin{equation}\label{eq: integral dt}
  \frac{1}{|r+p|^2 - |r|^2 + |r^\prime - p|^2 - |r^\prime|^2} = \int_0^{\infty} dt \, e^{-t|r+p|^2} e^{t|r|^2} e^{-t|r^\prime - p|^2}e^{t|r^\prime|^2}.
  \end{equation}
  Introducing the notation 
 \begin{equation}\label{eq def ut vt}
    \hat{u}^{t}_\sigma(k) := e^{-t|k|^2}\hat{u}^\gg_\sigma(k), \qquad \hat{v}^{t}_\sigma(k) := e^{t|k|^2}\hat{v}_{\sigma}^\gg(k),
  \end{equation}
  and using the definition of $\hat{f}_{r,r^\prime}$, we can write
  \begin{equation}\label{eq: B2 conf space}
  B_2 = 2\int_0^{\infty}dt\int_{\Lambda \times \Lambda} dxdy  V_\varphi(x-y) a_\uparrow(u^t_x)a_\uparrow( v^t_x)a_\downarrow(u^t_y)a_\downarrow( v^t_y),\nonumber
  \end{equation}
  where 
  \begin{equation}\label{eq: def ut vt x space}
    u^t_\sigma(x;y) = \frac{1}{L^3}\sum_{k\in \frac{2\pi}{L}\mathbb{Z}^3} \hat{u}_\sigma^t(k)e^{ik\cdot (x-y)}, \qquad v^t_\sigma(x;y) = \frac{1}{L^3}\sum_{k\in \frac{2\pi}{L}\mathbb{Z}^3} \hat{v}_\sigma^t(k)e^{ik\cdot (x-y)}.
  \end{equation}
The rigorous conjugation of $\widetilde{\mathcal{H}}^{\mathrm{eff}}_{\mathrm{corr}}$ under $T_2$, will be discussed in Section \ref{sec: T2 conjugation}.
\begin{remark}[Comparison with \cite{GHNS24}]\label{comparison new up bd} The unitary transformation $T_{2;\lambda}$  introduced above can be compared with the one defined in \cite[Eq. (2.19)]{GHNS24}. There are two main differences between the two constructions. First, in the denominator of $B_{2}$ we have $V_\varphi$, whereas in \cite[Eq. (2.19)]{GHNS24} the authors directly use $8\pi a$, which corresponds to the small momentum approximation of $V_\varphi$. The second, and more substantial, difference is that in our definition \eqref{eq: def T2 unitary} (see also \eqref{eq: def frr'}) there is no singularity in the denominator. This regularity follows from the support properties of $\hat{u}^\gg_\sigma$ and $\hat{v}^\gg_\sigma$. Discarding momenta near $k_F^\uparrow$ and $k_F^\downarrow$ simplifies the analysis but only allows us to obtain optimal estimates: essentially because we are neglecting momenta that contribute to the Huang–Yang correction of order $\rho^{7/3}$.
\end{remark}
\section{Useful bounds}\label{sec: useful bounds}
In this section we collect some bounds which we are going to use often to estimate the error terms. Let $g_{z}(z^\prime):= g(z-z^\prime)\in L^1(\Lambda)$. In the next lemma we prove some operator norm bounds for the following operators 
\begin{equation}\label{eq: def op b g}
  b_\sigma (g_z):= \int_{\Lambda} dz^\prime\, g(z-z^\prime)a_\sigma(u_{z^\prime}^>)a_{\sigma}( v_{z^\prime}), \qquad b^\ast_\sigma(g_z) := \int_{\Lambda} dz^\prime\, \overline{g(z-z^\prime)} a_\sigma^\ast( v_{z^\prime})a_\sigma^\ast(u_{z^\prime}^>),
\end{equation}
and 
\begin{equation}\label{eq: def op bj g}
  b_{\ell,\sigma}(g_z):= \int_{\Lambda} dz^\prime\, g(z-z^\prime)a_\sigma(u_{z^\prime}^>)a_{\sigma}(\partial_\ell v_{z^\prime}), \quad b^\ast_{\ell,\sigma}(g_z) := \int_{\Lambda} dz^\prime\, \overline{g(z-z^\prime)} a_\sigma^\ast(\partial_\ell v_{z^\prime})a_\sigma^\ast(u_{z^\prime}^>), \quad \ell = 1,2,3.
\end{equation}
\begin{lemma}\label{lem: bound b phi}Let $\varphi$ as in \eqref{eq: def phi}. For $j,\ell = 1,2,3$, it holds that
  \begin{equation} \label{eq: op bound b phi}
      \| b_\sigma(\varphi_z)\| \leq C \rho^{\frac{1}{3}}, \quad \| {b}_{\ell,\sigma}(\varphi_z)\| \leq C \rho^{\frac{2}{3}},  \quad \| b_\sigma(\partial_j\varphi_z)\| \leq C \rho^{\frac{1}{2}}, \quad \| {b}_{\ell,\sigma}(\partial_j\varphi_z)\| \leq C \rho^{\frac{5}{6}},
  \end{equation}
  and the same is true for the adjoint operators. Furthermore, let $\mathcal{E}_{\{\varphi_0, \chi_{\sqrt[3]{\rho}}\}}$ be as in \eqref{eq: period error scattering}. Then 
\begin{equation}\label{eq: bound b err scatt}
  \|b_\sigma((\mathcal{E}_{\{\varphi_0, \chi_{\sqrt[3]{\rho}}\}})_z)\| \leq C\rho,
\end{equation}
and the same holds for the adjoint operator.
\end{lemma}
The proof of Lemma \ref{lem: bound b phi} can be done similarly as the one of \cite[Lemma 4.8, Lemma 4.9]{Gia1}. However, a different and simpler proof can be found in \cite[Lemma 5.2]{LS1}, for which one needs only to bound the $L^1$ and $L^2$ norms of $\varphi$, $\partial_j \varphi$, $\mathcal{E}_{\{\varphi_0, \chi_{\sqrt[3]{\rho}}\}}$.
\begin{proposition}[Bounds for the number operator]\label{pro: N} Let $\lambda\in [0,1]$ and let $\psi$ be an approximate ground state in the sense of Definition \ref{def: approx gs}. It holds  that
\begin{equation}\label{eq: est N}
\langle T^\ast_{1;\lambda} R^\ast \psi, \mathcal{N}T^\ast_{1;\lambda} R^\ast \psi\rangle \leq C\langle T_1^\ast R^\ast \psi, \mathcal{N} T^\ast_1 R^\ast \psi\rangle + CL^{3}\rho^{\frac{5}{3}}, \qquad \langle T^\ast_{1;\lambda} R^\ast \psi, \mathcal{N} T^\ast_{1;\lambda} R^\ast \psi\rangle \leq CL^3\rho^{\frac{7}{6}}.
\end{equation}
Furthermore, let
\begin{equation}\label{eq: def N>}
  \mathcal{N}_> := \sum_\sigma\sum_{|k|\geq 2k_F} \hat{a}^\ast_{k,\sigma}\hat{a}_{k,\sigma}.
\end{equation}
It holds that
\begin{equation}\label{eq: N>}
  \langle T^\ast_{1;\lambda} R^\ast \psi , \mathcal{N}_> T^\ast_{1;\lambda} R^\ast \psi\rangle \leq C\rho^{-\frac{2}{3}}\langle T^\ast_{1;\lambda} R^\ast \psi, \mathbb{H}_0 T^\ast_{1;\lambda} R^\ast \psi\rangle,
\end{equation}
and
\begin{equation}\label{eq: N> lambda 1}
  \langle T^\ast_{1;\lambda} R^\ast \psi, \mathcal{N}_> T^\ast_{1;\lambda} R^\ast \psi\rangle \leq C\rho^{-\frac{2}{3}}\langle T^\ast_{1} R^\ast \psi, \mathbb{H}_0 T^\ast_{1} R^\ast \psi\rangle + CL^3 \rho^{\frac{5}{3}}.
\end{equation}
\end{proposition}
\textcolor{black}{The proof of Proposition \ref{pro: N} is as in \cite[Proposition 4.11, Proposition 4.15]{Gia1}. In particular, to prove \eqref{eq: N> lambda 1}, one has to use the propagation bounds proved in Proposition \ref{pro: propagation est}.} Note that, similarly as in \cite{Gia1}, the operator $\mathcal{N}_>$ will appear often in the estimates of the error terms, since
\begin{equation}\label{eq: u> wrt N>}
\int_{\Lambda_L}dx\, \|a_\sigma(u_x^>)\psi\|^2 \leq C\sum_{|k|\geq 2k_F}\langle \psi, \hat{a}_{k,\sigma}^\ast \hat{a}_{k,\sigma}\psi\rangle.
\end{equation}


\section{Conjugation under the  transformation $T_1$}\label{sec: conj T1}
In this section, we conjugate the correlation energy by the quasi-bosonic unitary transformation introduced in Section~\ref{sec: def T1}. Our main focus will be the proof of the lower bound corresponding to the energy density asymptotics stated in Theorem~\ref{thm: optimal lw bd}. At the end of the section, we briefly outline how the argument can be adapted to obtain the upper bound. Before entering the rigorous analysis, we first provide an intuitive explanation of the effect of the conjugation under $T_1$.\\

\noindent\textbf{Ideas on the conjugation under $T_1$.}\label{sec: strategy T1} As anticipated in Section \ref{sec: particle-hole} (see \eqref{eq: est Q1} and \eqref{eq: est Q3}), by applying the particle-hole transformation from Definition \ref{def: ferm bog} together with the priori estimates in Lemma \ref{lem: a priori est}, we obtain (see Section \ref{sec: Q3} and Section \ref{sec: lw bd T1} for the rigorous proof)
\[
  \langle R^\ast \psi, (\mathbb{H}_{0} + \mathbb{Q}) R^\ast \psi\rangle  \sim \langle R^\ast\psi, (\mathbb{H}_0 + \mathbb{Q}_2 + \mathbb{Q}_4)R^\ast \psi\rangle,
\]
for any $\psi$ approximate ground state with $\mathbb{H}_0$ and $\mathbb{Q}$ as in Proposition \ref{pro: fermionic transf}. Denoting $\mathcal{H}_{\mathrm{corr}}^{\mathrm{eff}} = \mathbb{H}_0 + \mathbb{Q}_2 + \mathbb{Q}_4$ and conjugating it under $T_1$, we get 
\begin{eqnarray*}
  \langle R^\ast \psi, \mathcal{H}_{\mathrm{corr}}^{\mathrm{eff}} R^\ast \psi\rangle\hspace{-0.25cm} &=& \hspace{-0.25cm}\langle T^\ast_1 R^\ast \psi, (\mathbb{H}_0 + \mathbb{Q}_2 + \mathbb{Q}_4) T^\ast_1 R^\ast \psi\rangle 
   - \int_0^1 d\lambda\, \partial_\lambda \langle T^\ast_{1;\lambda} R^\ast \psi, (\mathbb{H}_0 + \mathbb{Q}_4 + \mathbb{Q}_2) T^\ast_{1;\lambda} R^\ast \psi\rangle 
   \\
   &=& \hspace{-0.25cm} \langle T^\ast_1 R^\ast \psi, (\mathbb{H}_0 + \mathbb{Q}_2 + \mathbb{Q}_4) T^\ast_1 R^\ast \psi\rangle  + \int_0^1 d\lambda\, \langle T^\ast_{1;\lambda} R^\ast \psi, [\mathbb{H}_0 + \mathbb{Q}_4 + \mathbb{Q}_2 , B_1 - B_1^\ast] T^\ast_{1;\lambda} R^\ast \psi\rangle.
\end{eqnarray*}
The choice of $\varphi$ in the definition of $T_1$ implies that (see Section \ref{sec: scattering}) $\mathbb{Q}_2 + [\mathbb{H}_0 + \mathbb{Q}_4, B_1 - B_1^\ast]$ is almost zero up to error terms supported close to $k_F^\sigma$ ($\sigma = \uparrow, \downarrow$) due to the cut-off in \eqref{eq: def u>}. More precisely, neglecting other errors $\mathcal{O}(L^3\rho^{7/3})$, we have 
\[
  \mathbb{Q}_2 + [\mathbb{H}_0 + \mathbb{Q}_4, B_1 - B_1^\ast] \sim \widetilde{\mathbb{Q}},
\]
where
\[
  \widetilde{\mathbb{Q}} = \frac{1}{L^3}\sum_{p,r,r^\prime}(\hat{V} - \hat{V}\ast \hat{\varphi})(p)b_{p,r,\uparrow}b_{-p,r^\prime,\downarrow}\hat{u}^<_\uparrow(r+p)\hat{u}^<_\downarrow(r^\prime - p) + \mathrm{h.c.},
\]
and $\hat{u}^<_\sigma(k) := 1 - \hat{u}^>_\sigma(k)$ with $\hat{u}^>_\sigma$ defined as in \eqref{eq: def u>}.
Applying the Duhamel's formula again  (see Proposition \ref{pro: T2<} and Section \ref{sec: lw bd T1}) up to errors of the order $\mathcal{O}(L^3\rho^{7/3})$, we find that 
\begin{align*}
  \langle R^\ast \psi, \mathcal{H}^{\mathrm{eff}}_{\mathrm{corr}} R^\ast \psi\rangle &\sim \langle T^\ast_{1} R^\ast \psi, (\mathbb{H}_0 + \mathbb{Q}_4) T^\ast_{1} R^\ast \psi\rangle + \langle T^\ast_{1} R^\ast \psi, \widetilde{\mathbb{Q}}  T^\ast_{1} R^\ast \psi\rangle
   + \int_0^1 d\lambda\, \langle T^\ast_{1;\lambda} R^\ast \psi, [\mathbb{Q}_2, B_1 - B_1^\ast] T^\ast_{1;\lambda} R^\ast \psi\rangle
  \\
  & \quad  +\int_0^1d\lambda\int_1^\lambda d\lambda^\prime  \langle T^\ast_{1;\lambda^\prime} R^\ast \psi, [\mathbb{Q}_2, B_1 - B_1^\ast] T^\ast_{1;\lambda} R^\ast \psi\rangle,
\end{align*}
and (see Proposition \ref{pro: Q2}),
\begin{multline*}
 \int_0^1 d\lambda\, \langle T^\ast_{1;\lambda} R^\ast \psi, [\mathbb{Q}_2, B_1 - B_1^\ast] T^\ast_{1;\lambda} R^\ast \psi\rangle
 +\int_0^1d\lambda\int_1^\lambda d\lambda^\prime \langle T^\ast_{1;\lambda^\prime} R^\ast \psi, [\mathbb{Q}_2, B_1 - B_1^\ast] T^\ast_{1;\lambda} R^\ast \psi\rangle 
 \\
 \sim   -\rho_\uparrow\rho_\downarrow\sum_{p\in\frac{2\pi}{L}\mathbb{Z}^3}\hat{V}(p)\hat{\varphi}(p).
\end{multline*}
In conclusion, after the conjugation under $T_1$ and using the explicit formula for the FFG energy in \eqref{eq: FFG energy}, for $L$ large enough, we find that
\[
 E_L(N_\uparrow, N_\downarrow) \geq  \frac{3}{5}(6\pi^2)^{\frac{2}{3}}\left (\rho_\uparrow^{\frac{5}{3}} + \rho_\downarrow^{\frac{5}{3}}\right)L^3 + 8\pi a  \rho_\uparrow \rho_\downarrow L^3 +  \frac{1}{2}\langle T^\ast_{1}R^\ast \psi, (\mathbb{H}_0   + 2\widetilde{\mathbb{Q}})T^\ast_{1}R^\ast \psi\rangle + \mathcal{E}_{T_1}, \qquad |\mathcal{E}_{T_1} |\leq CL^3\rho^{\frac{7}{3}},
\]
where the factor $1/2$ in front of $\mathbb{H}_0$ is due to the fact that we partially used the positivity of $\mathbb{H}_0$ to estimate many error terms. To prove an optimal lower bound,  we need to refine our analysis for $\widetilde{\mathbb{Q}}$. This is exactly what was missing in \cite{Gia1}. To do that, we first prove that (see Section \ref{sec: reduction new eff corr en})
\[
  \langle T^\ast_1 R^\ast \psi, \widetilde{\mathbb{Q}}T^\ast_1 R^\ast \psi\rangle = \langle T^\ast_1 R^\ast \psi, \widetilde{\mathbb{Q}}_>T^\ast_1 R^\ast \psi\rangle + \mathcal{E}, \qquad |\mathcal{E}| \leq CL^3\rho^{\frac{7}{3}},
\]
with $\widetilde{\mathbb{Q}}_{>}$ as in \eqref{eq: def Qtilde>}. The correlations in $\widetilde{\mathbb{Q}}_{>}$ will be treated with the quasi-bosonic Bogoliubov transformation introduced in \eqref{eq: def T2 unitary}. 
\subsection{Propagation estimates}\label{sec: propagation estimates}
In this section we collect a number of propositions useful to propagate the bounds in Lemma \ref{lem: a priori est} via a Gr\"onwall-like argument. The proofs presented here are similar to those in \cite[Section 5]{Gia1}. However, due to the different localization we employ for the momenta inside or outside the Fermi ball, we recall some of the relevant estimates for completeness. Through this section (and in what follows), we use the shorthand notation $\sum_\sigma$ for $\sum_{\sigma = \uparrow, \downarrow}$, $\sum_p$ for $\sum_{p\in (2\pi/L)\mathbb{Z}^3}$, $\int dx$ for $\int_{\Lambda} dx$ and $\|\cdot\|_p$ for $\|\cdot\|_{L^p(\Lambda)}$.
In the following propositions, we will repeatedly use the bounds
\begin{equation}\label{eq: bounds av}
  \|a_\sigma(\partial^n v)\| = \|\partial^n v_\sigma\|_2 \leq C\rho^{\frac{n}{3} + \frac{1}{2}},\qquad \|v\|_\infty \leq C\rho,
\end{equation}
which follow from the support properties of $\hat{v}_\sigma$.
Moreover, we frequently write $u^>_\sigma = \delta - \alpha^<_\sigma$ with $\alpha^<_\sigma$ defined in \eqref{eq: def alpha<}, and $(u^>_\sigma)^2 = \delta  - \nu_\sigma$ where $(u^>_\sigma)^2$ denotes the function whose Fourier coefficients are given by $(\hat{u}^>_\sigma)^2(\cdot)$. Here, $\hat{\nu}_\sigma(\cdot)$ is the periodization of a smooth function on $\mathbb{R}^3\rightarrow \mathbb{R}$ satisfying 
\begin{equation}\label{eq: def nu}
  \mathcal{F}({\nu}_\sigma)(k) = \begin{cases} 1 &\mbox{if}\,\,\, |k| < 2k_F^\sigma \\ 0 &\mbox{if}\,\,\, |k| > 3k_F^\sigma.\end{cases}
\end{equation}
We recall that $\mathcal{F}$ denotes the Fourier transform on $\mathbb{R}^3$.
Throughout the paper, we repeatedly make use of the estimates
\begin{equation}\label{eq: bounds alpha< nu}
  \|\alpha^<_\sigma\|_1 \leq C, \qquad \|\nu_\sigma\|_1 \leq C, \qquad \|\alpha^<_\sigma\|_\infty \leq C\rho, \qquad \|\nu_\sigma\|_\infty \leq C\rho,
\end{equation}
where the $L^\infty$ bounds follow from the support properties of $\hat{\alpha}^<_\sigma$ and $\hat{\nu}_\sigma$, while the $L^1$ estimates can be obtained as explained in \cite[Proposition 4.1]{Gia1}.
\begin{proposition}[Propagation estimate for $\mathbb{H}_0$ - Part I]\label{pro: H0} Let $\lambda\in [0,1]$ and let $\psi$ be an approximate ground state as in Definition \ref{def: approx gs}. Under the assumptions of Theorem \ref{thm: optimal lw bd}, it holds  that
\begin{equation}\label{eq: prop H0}
  \partial_\lambda T_{1;\lambda}\mathbb{H}_0 T^\ast_{1;\lambda} =  T_{1;\lambda}( \mathbb{T}_1 + \mathcal{E}_{\mathbb{H}_0}) T^\ast_{1;\lambda},
\end{equation}
with
\begin{equation}\label{eq: def T1 kinetic}
  \mathbb{T}_1 =- 2\int dxdy\,  \Delta\varphi(x-y) a_\uparrow(u^>_x)a_\uparrow( v_x)a_\downarrow(u^>_y) a_\downarrow( v_y) + \mathrm{h.c.}
\end{equation}
and
\begin{equation}\label{eq: error kin energy}
  \mathcal{E}_{\mathbb{H}_0} = 2\sum_{j=1}^3\int dxdy\, \partial_j \varphi(x-y) \left[a_\uparrow(u^>_x)a_\uparrow(\partial_j v_x)a_\downarrow(u^>_y) a_\downarrow( v_y) - a_\uparrow(u^>_x)a_\downarrow( v_x)a_\downarrow(u^>_y) a_\downarrow(\partial_j v_y)\right] + \mathrm{h.c.}
\end{equation}
Furthermore, for any $0 <\delta <1$, it holds that 
\begin{align}\label{eq: est err kin}
  |\langle T^\ast_{1;\lambda}R^\ast \psi,\mathcal{E}_{\mathbb{H}_0} T^\ast_{1;\lambda}R^\ast \psi\rangle| &\leq CL^3\rho^{\frac{7}{3}} + \delta\langle T^\ast_1 R^\ast \psi, \mathbb{H}_0 T^\ast_1 R^\ast \psi\rangle\nonumber
  \\
  &\quad+ C\int_{1}^\lambda d\lambda^\prime\,  (\rho^{\frac{1}{2}}\langle  T^\ast_{1;\lambda^\prime} R^\ast \psi, \mathbb{H}_0T^\ast_{1;\lambda^\prime} R^\ast \psi\rangle + \rho^{\frac{4}{3}}\langle T^\ast_{1;\lambda^\prime} R^\ast \psi, \mathcal{N}T^\ast_{1;\lambda^\prime} R^\ast \psi\rangle).
\end{align}
\end{proposition}
\begin{proof} 
 To prove \eqref{eq: prop H0}, we compute
  \begin{equation*} 
   \partial_\lambda T_{1;\lambda} \mathbb{H}_0T^\ast_{1;\lambda} = - T_{1;\lambda} [\mathbb{H}_0,B_1]T_{1;\lambda}^\ast + \mathrm{h.c.}
  \end{equation*}
  The equality in \eqref{eq: prop H0} follows directly from a computation of the commutator. We skip this calculation and refer to \cite[Proposition 5.1]{Gia1} for details. The estimate in \eqref{eq: est err kin} is the one proved in\cite[Lemma 5.2]{Gia1}. We discuss it again here in view of some technical differences with respect to the proof in \cite[Lemma 5.2]{Gia1}.  As in \cite[Lemma 5.2]{Gia1}, it is enough to estimate
\begin{equation}\label{eq: def A}
\mathbb{A}:=  2\sum_{j=1}^3\int\, dxdy\, \partial_j \varphi(x-y)\, a_\uparrow^\ast(u_x)a_\uparrow^\ast(\partial_j v_x)a_\downarrow^\ast(u_y) a_\downarrow^\ast( v_y).
\end{equation}
To have a shorter notation, we set $\xi_\lambda := T^\ast_{1;\lambda}R^\ast \psi$. It is convenient (see also the discussion in \cite[Section 5.1]{Gia1}) to use Duhamel's formula to write 
\begin{equation}\label{eq: prop A}
  \langle \xi_\lambda, \mathbb{A} \xi_\lambda \rangle  = \langle \xi_1, \mathbb{A}\xi_1\rangle  + \int_1^{\lambda} d\lambda^\prime \partial_\lambda\langle\xi_{\lambda^\prime}, \mathbb{A}\xi_{\lambda^\prime}\rangle = \langle \xi_1, \mathbb{A}\xi_1\rangle  - \int_1^{\lambda} d\lambda^\prime\langle\xi_{\lambda^\prime}, [\mathbb{A}, B_1]\xi_{\lambda^\prime}\rangle.
\end{equation}
The estimate of $\langle \xi_1, \mathbb{A}\xi_1\rangle$ can be done as in \cite[Lemma 5.2]{Gia1}, see \cite[Eqs. (5.60)--(5.63)] {Gia1}. Omitting then the details, we directly write that for any $0<\delta<1$, we have
\begin{equation}
  |\langle \xi_1, \mathbb{A}\xi_1 \rangle | \leq CL^3 \rho^{\frac{7}{3}} + \delta \langle \xi_1, \mathbb{H}_0 \xi_1\rangle.
\end{equation}
We now estimate the error terms coming from the commutator $[\mathbb{A}, B_1]$:
\begin{multline*}
   [\mathbb{A}, B_1] = \sum_{j=1}^3\int dxdydzdz^\prime \partial_j\varphi(x-y)\varphi(z-z^\prime) \, [a_\uparrow^\ast( u^>_x) a_\downarrow^\ast(u^>_y)a_\downarrow^\ast( v_y)a_\uparrow^\ast(\partial_j v_x) ,a_\uparrow(u^>_z)a_\uparrow( v_z)a_\downarrow(u^>_{z^\prime})a_\downarrow( v_{z^\prime})] .
\end{multline*}
We compute
\begin{align}\label{eq: comm B A}
&[a_\uparrow^\ast(u^>_x) a_\downarrow^\ast(u^>_y)a_\downarrow^\ast( v_y)a_\uparrow^\ast(\partial_j v_x) ,a_\uparrow(u^>_z)a_\uparrow( v_z)a_\downarrow(u^>_{z^\prime})a_\downarrow( v_{z^\prime})] 
  \\
  &=  -a_\uparrow^\ast( u^>_x) a_\downarrow^\ast(u^>_y)\big(-  \partial_jv_\uparrow(x;z)v_\downarrow (y;z^\prime) + \partial_jv_\uparrow(x;z)a^\ast_\downarrow( v_y) a_\downarrow ( v_{z^\prime})
  + v_\downarrow(y;z^\prime) a^\ast_\uparrow(\partial_j v_x) a_\uparrow( v_z)\big)a_\uparrow(u^>_z)a_\downarrow(u^>_{z^\prime})\nonumber
  \\
  &- a_\uparrow( v_z)a_\downarrow( v_{z^\prime})\big((u^>_\downarrow)^2(z^\prime;y)(u^>_\uparrow)^2(z;x) - (u^>_\uparrow)^2(z;x) a^\ast_\downarrow(u^>_y)a_\downarrow(u^>_{z^\prime}) -(u^>_\downarrow)^2(z^\prime;y)a^\ast_\uparrow(u^>_x)a_\uparrow(u^>_z)\big)a_\downarrow^\ast(v_y)a_\uparrow(\partial_j v_x),\nonumber
\end{align}
In the following, we estimate all the type of error terms. From the first line in the right hand side of $\eqref{eq: comm B A}$, we have three types of error terms, which we denote by $\mathrm{I}_a$, $\mathrm{I}_b$ and $\mathrm{I}_c$.  The first one is 
\begin{align}\label{eq: term Ia kin err}
  \mathrm{I}_a &=  \sum_{j=1}^3\int dxdydzdz^\prime\, \partial_j\varphi(x-y)\varphi(z-z^\prime)\partial_j v_\uparrow(x;z) v_\downarrow (y;z^\prime) \langle\xi_\lambda, a^\ast_\uparrow(u^>_x)a^\ast_\downarrow(u^>_y)  a_\uparrow(u^>_z)a_\downarrow(u^>_{z^\prime})\xi_\lambda\rangle\nonumber
\\
& = \sum_{j=1}^3\int dydz\, \left\langle \left(\int dx\, \partial_j \varphi(x-y) \partial_j v_\uparrow (x;z)a_\uparrow(u^>_x)\right) a_\downarrow(u^>_y)\xi_\lambda, \left(\int dz^\prime\, \varphi(z-z^\prime) v_\downarrow(y;z^\prime) a_\downarrow(u^>_{z^\prime})\right) a_\uparrow(u^>_{z})\xi_\lambda\right\rangle.
\end{align}
Proceeding similarly as in \cite[Proposition 4.2 (Eq. (4.16))]{GHNS24}, i.e., using that $0\leq \hat{u}_\uparrow^> \leq 1$, we can bound
\[
   \left\| \int dx\, \partial_j \varphi(x-y) \partial_\ell v_\uparrow (x;z)a_\uparrow(u^>_x)\right\|\leq \int dt\, |\partial_j\varphi(y-t)|^2 |\partial_\ell v_\uparrow(z;t)|^2,
\]
and
\begin{equation}\label{eq: conv L2 norm}
  \left\| \int dz^\prime\, \varphi(z-z^\prime) v_\downarrow(y;z^\prime) a_\downarrow(u^>_{z^\prime})\right\| \leq \int dt\, |\varphi(z-t)|^2 | v_\downarrow(y;t)|^2.
\end{equation}
Therefore, by Cauchy-Schwarz's inequality, we have 
\begin{multline*}
  |\mathrm{I}_a| 
\leq \sum_{j=1}^3\left(\int dzdy dt\, |\partial_j \varphi(y-t)|^2 |\partial_\ell v_\uparrow(t-z)|^2 \|a_\downarrow(u^>_y)\xi_\lambda\|^2\right)^{\frac{1}{2}}\times 
\\
\times \left( \int dzdy dt\, |\varphi(z-t)|^2 |v_\downarrow(y;t)|^2  \|a_\uparrow(u^>_z)\xi_\lambda\|^2 \right)^{\frac{1}{2}}.
\end{multline*}
Using now Lemma \ref{lem: bound phi} and \eqref{eq: bounds av}, we get that 
\[
  |\mathrm{I}_a| \leq\sum_{j=1}^3 \|\varphi\|_2 \|\partial_j\varphi\|_2 \|\partial_\ell v_\uparrow\|2\|v_\downarrow\|_2 \langle \xi_\lambda, \mathcal{N}_> \xi_\lambda\rangle \leq C\rho^{\frac{4}{3} - \frac{1}{6}}\langle   \xi_\lambda, \mathcal{N}_> \xi_\lambda\rangle \leq C\rho^{\frac{2}{3} - \frac{1}{6}}\langle \xi_\lambda, \mathbb{H}_0 \xi_\lambda\rangle, 
\]
where we used \eqref{eq: u> wrt N>} and \eqref{eq: N>}. Therefore, all together we proved that 
\begin{equation}
  |\mathrm{I}_a| \leq C\rho^{\frac{1}{2}}\langle \xi_\lambda, \mathbb{H}_0 \xi_\lambda\rangle.
\end{equation}
Another type of error term is: 
\[
  \mathrm{I}_b = \sum_{j=1}^3\int dxdz\,\partial_j v_\uparrow(x;z) \langle\xi_\lambda, a^\ast_\uparrow( u^>_x)b^\ast_{\downarrow}(\partial_j\varphi_x)  a_\uparrow(u^>_z)b_\downarrow(\varphi_z)\xi_\lambda\rangle.
\]
To estimate $\mathrm{I}_{b}$ it is convenient to rewrite it as 
\[
  \mathrm{I}_b = \frac{1}{L^3}\sum_{j=1}^3\sum_k (-ik_j)\hat{v}_\uparrow(k) \left\langle \left(\int dx\,  e^{-ik\cdot x} b_{\downarrow}(\partial_j\varphi_x)a_\uparrow( u^>_x)\xi_\lambda\right), \left(\int dz\, e^{-ik\cdot z}\, b_\downarrow(\varphi_z)a_\uparrow(u^>_z)\xi_\lambda\right)\right\rangle.
\]
Using then that $|(-ik_j)\hat{v}_\uparrow(k)|\leq C\rho^{1/3}$ and the  Cauchy-Schwarz inequality, we find
\begin{align*}
  |\mathrm{I}_b| &\leq C\rho^{\frac{1}{3}}\sum_{j=1}^3\sum_k \left\|\int dx\,  e^{-ik\cdot x} b_{\downarrow}(\partial_j\varphi_x)a_\uparrow( u^>_x)\xi_\lambda\right\|\left\|\int dz\, e^{-ik\cdot z}\, b_\downarrow(\varphi_z)a_\uparrow(u^>_z)\xi_\lambda\right\|
  \\
  &\leq C\sum_{j=1}^3\rho^{\frac{1}{3}}\left(\int dx \|b_{\downarrow}(\partial_j\varphi_x)a_\uparrow( u^>_x)\xi_\lambda\|^2\right)^{\frac{1}{2}}\left( \int dz\, \|b_\downarrow(\varphi_z)a_\uparrow(u^>_z)\xi_\lambda \|^2\right)^{\frac{1}{2}}
  \\
  &\leq  C\rho^{\frac{7}{6}}\langle \xi_\lambda, \mathcal{N}_> \xi_\lambda\rangle \leq C\rho^{\frac{1}{2}}\langle \xi_\lambda,\mathbb{H}_0 \xi_\lambda\rangle,
\end{align*}
where we used also \eqref{eq: u> wrt N>} and \eqref{eq: N>} together with the bounds in Lemma \ref{lem: bound b phi}.
The last type of error term coming from the first line in the right hand side of \eqref{eq: comm B A} is 
\[
  \mathrm{I}_c= \sum_{j=1}^3\int dydz^\prime\,v_\downarrow(y;z^\prime)\langle \xi_\lambda, a^\ast_\downarrow(u^>_y){b}_{j,\uparrow}^\ast(\partial_j\varphi_y) b_\uparrow(\varphi_{z^\prime}) a_\downarrow(u^>_{z^\prime})\xi_\lambda\rangle.
\]
  The estimate for $\mathrm{I}_c$ can be done similarly as the one for $\mathrm{I}_b$, we omit the details and we write $|\mathrm{I}_c|\leq  C\rho^{\frac{1}{2}}\langle \xi_\lambda, \mathbb{H}_0 \xi_\lambda\rangle$.
We now estimate all the contributions coming from the second line in the right hand side of \eqref{eq: comm B A}. To do that is sufficient to consider two different types of terms. The first one is, for fixed $j=1,2,3$,  
\begin{equation}
  \mathrm{II}_a = \sum_{j=1}^3\int dxdydzdz^\prime \partial_j\varphi(x-y) \varphi (z-z^\prime) (u^>_\uparrow)^2(z;x)  \langle \xi_\lambda, a^\ast_\downarrow(u^>_y) a_\uparrow( v_z)a_\downarrow( v_{z^\prime}) a^\ast_\downarrow( v_y)a^\ast_\uparrow(\partial_j v_x)a_\downarrow(u^>_{z^\prime}) \xi_\lambda\rangle.
\end{equation}
In order to estimate $\mathrm{II}_a$, similarly as what explained in Remark \ref{rem: comparisono Gia1 cut-off}, it is convenient to write $(\hat{u}^>_\sigma)^2 = \delta - {\nu}_\sigma$, with $\hat{\nu}_\sigma$ introduced in \eqref{eq: def nu}. To estimate $\mathrm{II}_a$ we then have to bound two terms: 
\begin{equation*}
  \mathrm{II}_{a;1} =   \sum_{j=1}^3\int dxdydz^\prime \partial_j\varphi(x-y) \varphi (x-z^\prime) \langle \xi_\lambda, a^\ast_\downarrow(u^>_y) a_\uparrow( v_x)a_\downarrow( v_{z^\prime}) a^\ast_\downarrow( v_y)a^\ast_\uparrow(\partial_j v_x)a_\downarrow(u^>_{z^\prime}) \xi_\lambda\rangle,
  \end{equation*}
  and
  \begin{equation*}
  \mathrm{II}_{a;2} = -\sum_{j=1}^3 \int dxdydzdz^\prime \partial_j\varphi(x-y) \varphi (z-z^\prime) \nu_\uparrow(z;x)  \langle \xi_\lambda, a^\ast_\downarrow(u^>_y) a_\uparrow( v_z)a_\downarrow( v_{z^\prime}) a^\ast_\downarrow( v_y)a^\ast_\uparrow(\partial_j v_x)a_\downarrow(u^>_{z^\prime}) \xi_\lambda\rangle.
  \end{equation*}
  We consider first $\mathrm{II}_{a;1}$. Using \eqref{eq: bounds av}, we have 
  \begin{align*}
  |\mathrm{II}_{a;1}| &\leq C\rho^{2+\frac{1}{3}}\sum_{j=1}^3 \int dxdydz^\prime |\partial_j\varphi(x-y)||\varphi(x-z^\prime)| \|a_\downarrow( u^>_y)\xi_\lambda\|\|a_\downarrow( u^>_{z^\prime}) \xi_\lambda\rangle\| \leq C\rho^{2+\frac{1}{3}} \|\partial_j\varphi\|_1\|\varphi\|_1 \langle \xi_\lambda, \mathcal{N}_>\xi_\lambda\rangle
  \\
  &\leq C\rho^{\frac{2}{3}}\langle \xi_\lambda, \mathbb{H}_0\xi_\lambda\rangle,
  \end{align*}
  where we also used Lemma \ref{lem: bound phi} and  \eqref{eq: u> wrt N>}, \eqref{eq: N>}. The estimate for $\mathrm{II}_{a;2}$ can be done similarly, using also the bound in \eqref{eq: bounds alpha< nu}. We therefore have 
\[
  |\mathrm{II}_{a;2}| \leq  C\rho^{1+\frac{1}{3}} \sum_{j=1}^3\int dxdydzdz^\prime |\partial_j\varphi(x-y)||\varphi(z-z^\prime)||\nu_\uparrow(z;x) |\|a_\downarrow( u^>_y)\xi_\lambda\|\|a_\downarrow( u^>_{z^\prime}) \xi_\lambda\rangle\|\leq C\rho^{\frac{2}{3}}\langle \xi_\lambda, \mathbb{H}_0\xi_\lambda\rangle.
\]
The last term to consider is 
\[
  \mathrm{II}_b = -\sum_{j=1}^3\int dxdydzdz^\prime \partial_j\varphi(x-y)\varphi(z-z^\prime)(u^>_\downarrow)^2(z^\prime;y)(u^>_\uparrow)^2(z;x) \langle \xi_\lambda, a_\uparrow( v_z)a_\downarrow( v_{z^\prime})a^\ast_\downarrow( v_y)a^\ast_\uparrow(\partial_j v_x)\xi_\lambda\rangle.
\]
It is convenient to write it partially in normal order: 
\[
a_\uparrow( v_z)a_\downarrow( v_{z^\prime})a^\ast_\downarrow( v_y)a^\ast_\uparrow(\partial_j v_x)=  - v_\downarrow(z^\prime;y) a_\uparrow^\ast(\partial_j v_x) a_\uparrow(v_z) - a^\ast_\downarrow( v_y) a_\uparrow( v_z)  a^\ast_\uparrow(\partial_j v_x)a_\downarrow( v_{z^\prime})+ v_\downarrow(z^\prime;y) \partial_j v_\uparrow(z;x).
\]
As a consequence, we write $\mathrm{II}_b$ as a sum of three terms, $\mathrm{II}_b = \mathrm{II}_{b;1} + \mathrm{II}_{b;2} + \mathrm{II}_{b;3}$. The first one is 
\[
  \mathrm{II}_{b;1} = \sum_{j=1}^3\int dxdydzdz^\prime \partial_j\varphi(x-y)\varphi(z-z^\prime)(u^>_\downarrow)^2(z^\prime;y)(u^>_\uparrow)^2(z;x)v_\downarrow(z^\prime;y)\langle \xi_\lambda, a_\uparrow^\ast (\partial_j v_x)a_\uparrow(v_z)\xi_\lambda\rangle.
\]
As before, we write $(u^>_\sigma)^2 = \delta - \nu_\sigma$. Therefore, we write $\mathrm{II}_{b;1} = \mathrm{II}_{b;1;0} + \mathrm{II}_{b;1;1} + \mathrm{II}_{b;1;2}$, where $\mathrm{II}_{b;1;\ast}$ denotes the part of $\mathrm{II}_{b;1}$ having exactly $\ast\in \{0,1,2\}$ $\delta$. We first consider  
\[
\mathrm{II}_{b;1;2} = \sum_{j=1}^3\int dxdy\,  \partial_j\varphi(x-y)\varphi(x-y)v_\downarrow(y;y) \langle \xi_\lambda,a^\ast_\uparrow(\partial_j v_x) a_\uparrow ( v_x)\xi_\lambda\rangle.
\]
Writing then $\varphi \partial_j\varphi = \frac{1}{2}\partial_j \varphi^2$ and integrating by parts, we find 
\begin{align*}
  |\mathrm{II}_{b;1;2}| &\leq C\rho\sum_{j=1}^3\int dxdy\, |\varphi(x-y)|^2 \left(\|a_\uparrow(\partial_j^2 v_x) \xi_\lambda\|\|a_\uparrow ( v_x)\xi_\lambda\| + \|a_\uparrow(\partial_j v_x)\psi\|^2\right)
  \\
  &\leq C\rho^{1 +\frac{2}{3}} \|\varphi\|_2 \langle \xi_\lambda, \mathcal{N}\xi_\lambda\rangle \leq C\rho^{\frac{4}{3}}\langle \xi_\lambda, \mathcal{N}\xi_\lambda\rangle ,
\end{align*}
where we also used Lemma \ref{lem: bound phi}. The term $\mathrm{II}_{b;1;1}$ can be estimated similarly as $\mathrm{II}_{b;1;2}$ using also the bounds in \eqref{eq: bounds av} and \eqref{eq: bounds alpha< nu}. Therefore, we find 
\[
  |\mathrm{II}_{b;1;1}| \leq C\rho^{\frac{4}{3}}\langle \xi_\lambda, \mathcal{N}\xi_\lambda\rangle.
\]
For the term $\mathrm{II}_{b;1;0}$ we have:
\begin{align*}
  |\mathrm{II}_{b;1;0}| &\leq\sum_{j=1}^3 \int dxdydzdz^\prime |\partial_j\varphi(x-y)||\varphi(z-z^\prime)||\nu_\downarrow(z^\prime;y)||\nu_\uparrow(z;x)||v_\downarrow(z^\prime;y)||a_\uparrow (\partial_j v_x)\xi_\lambda\|\|a_\uparrow(v_z)\xi_\lambda\|
  \\
  &\leq C\rho^{2+\frac{1}{3}}\sum_{j=1}^3\|\nu_\sigma\|_1 \|\partial_j\varphi\|_1 \|\varphi\|_1\langle \xi_\lambda, \mathcal{N}\xi_\lambda\rangle \leq C\rho^{\frac{4}{3}}\langle \xi_\lambda, \mathcal{N}\xi_\lambda\rangle.
\end{align*}
We therefore conclude that 
\[
  |\mathrm{II}_{b;1}|\leq C\rho^{\frac{4}{3}}\langle \xi_\lambda, \mathcal{N}\xi_\lambda\rangle.
\]
The term $\mathrm{II}_{b;2}$: 
\[
  \mathrm{II}_{b;2} = \sum_{j=1}^3\int dxdydzdz^\prime \partial_j\varphi(x-y)\varphi(z-z^\prime)(u^>_\downarrow)^2(z^\prime;y)(u^>_\uparrow)^2(z;x)\langle \xi_\lambda, a^\ast_\downarrow(v_y) a_\uparrow(v_z) a_\uparrow^\ast (\partial_j v_x)a_\downarrow(v_{z^\prime})\xi_\lambda\rangle,
\]
can be estimated similarly as $\mathrm{II}_{b;1}$, using that $\|a_\uparrow(v_z) a_\uparrow^\ast (\partial_j v_x)\|\leq C\rho^{1+1/3}$ in place of $\|v_\sigma\|_\infty\leq C\rho$. We omit the details. We have
\[
  | \mathrm{II}_{b;2}|\leq C\rho^{\frac{4}{3}}\langle \xi_\lambda, \mathcal{N}\xi_\lambda\rangle.
\]
It remains to consider the constant term:
\begin{equation}\label{eq: term IIIf H0}
  \mathrm{II}_{b;3} = -\sum_{j=1}^3\int dxdydzdz^\prime \partial_j\varphi(x-y)\varphi(z-z^\prime)(u^>_\downarrow)^2(z^\prime;y)(u^>_\uparrow)^2(z;x)\partial_j v_\uparrow(x;z)v_\downarrow (y;z^\prime).
\end{equation} 
Writing as before $(u^>_\sigma)^2 = \delta - \nu_\sigma$, we decompose  $\mathrm{II}_{b;3} = \mathrm{II}_{b;3;0} + \mathrm{II}_{b;3;1} + \mathrm{II}_{b;3;2}$ with $\mathrm{II}_{b;3;\ast}$ given by the term having $\ast \in \{0,1,2\}$ $\delta$. We note that, similarly as for the term in \cite[Eq. (5.54)]{Gia1}, we have $\mathrm{II}_{b;3;2} =0$, which can be seen by rewriting it in momentum space (see \cite{Gia1} for more details). The term $\mathrm{II}_{b;2;1}$ is
\begin{equation}
  \mathrm{II}_{b;2;1} = \sum_{j=1}^3\int dxdydzdz^\prime \partial_j\varphi(x-y)\varphi(z-z^\prime)\left(\delta(z^\prime;y)\nu_\uparrow(z;x) + \delta(z;x)\nu_\downarrow(z^\prime;y)\right)\partial_j v_\uparrow(x;z)v_\downarrow (y;z^\prime),
\end{equation}
and can be estimated using Lemma \ref{lem: bound phi} together with the bound in \eqref{eq: bounds alpha< nu} for $\|\nu_\sigma\|_\infty$:
\begin{eqnarray}
  |\mathrm{II}_{b;2;1}| &\leq& C\rho^{\frac{7}{3} +1}  \sum_{j=1}^3\int dxdydz |\partial_j\varphi(x-y)| |\varphi(z-y)|\leq CL^3\rho^{\frac{7}{3} + 1}\sum_{j=1}^3\|\partial_j\varphi\|_1 \|\varphi\|_1  \leq CL^3\rho^{\frac{7}{3}}\nonumber.
\end{eqnarray}
To conclude, we estimate
\begin{equation}
  \mathrm{II}_{b;2;0} =  -\sum_{j=1}^3\int dxdydzdz^\prime  \partial_j\varphi(x-y)\varphi(z-z^\prime)\nu_\downarrow(z^\prime;y)\nu_\uparrow(z;x)\partial_j v_\uparrow(x;z)v_\downarrow (y;z^\prime).
\end{equation}
Using the bounds in \eqref{eq: bounds alpha< nu} for $\nu_\sigma$ and proceeding similar as above,  we find that
 \[
 |\mathrm{II}_{b;2;0}| \leq CL^3\rho^{\frac{7}{3}} \sum_{j=1}^3\|\partial_j\varphi\|_1 \|\varphi\|_1 \|\nu_\downarrow\|_1\|\nu_\uparrow\|_\infty\leq CL^3\rho^{\frac{7}{3}}. 
 \] 
 Collecting all the bounds, we conclude the proof. 
\end{proof}
We now propagate the estimate for the operator $\mathbb{Q}_{4}$ introduced in \eqref{eq: def Qi}.
\begin{proposition}[Propagation estimate for $\mathbb{Q}_{4}$]\label{pro: Q4}
Let $\lambda\in [0,1]$. Let $\psi$ be an approximate ground state as in Definition \ref{def: approx gs}. Under the  assumptions of Theorem \ref{thm: optimal lw bd}, it holds that  
\begin{equation}
\partial_\lambda T_{1;\lambda} \mathbb{Q}_{4}T^\ast_{1;\lambda} =  T_{1;\lambda}\mathbb{T}_2T_{1;\lambda}^\ast + \mathcal{E}_{\mathbb{Q}_{4}},
\end{equation}  
where 
\begin{equation}\label{eq: def T2}
  \mathbb{T}_2 := \int\, dxdy\, V(x-y)\varphi(x-y)a_\uparrow(u_x)a_\uparrow( v_x)a_\downarrow(u_y)a_\downarrow( v_y) + \mathrm{h.c.},
\end{equation}
and
\begin{equation} \label{eq: est err Q4 B1}
  |\langle T^\ast_{1;\lambda}\psi, \mathcal{E}_{\mathbb{Q}_{4}}T^\ast_{1;\lambda}\psi| \leq   CL^{\frac{3}{2}}\rho^{\frac{4}{3}}\|\mathbb{Q}^{\frac{1}{2}}_{4}T^\ast_{1;\lambda}\psi\|.
\end{equation}
\end{proposition}
\begin{proof} 
Similarly as in Proposition \ref{pro: H0}, we compute
  \begin{equation} \label{eq: der lambda Q1}
    \partial_\lambda T_{1;\lambda} \mathbb{Q}_{4} T^\ast_{1;\lambda} = -T_{1;\lambda} [\mathbb{Q}_{4}, B_1]T^\ast_{1;\lambda} + \mathrm{h.c.}
  \end{equation}
We have
  \begin{multline*} 
    \hspace{-0.3cm}[\mathbb{Q}_{4},B_1] = \frac{1}{2}\sum_{\sigma\neq\sigma^\prime} \int dxdydzdz^\prime\, V(x-y)\varphi(z-z^\prime)\times
    \\
    \times [a^\ast_\sigma(u_x)a^\ast_{\sigma^\prime}(u_y)a_{\sigma^\prime}a(u_y)a_{\sigma}(u_x), a_\uparrow(u^>_z)a_\uparrow( v_z)a_\downarrow(u^>_{z^\prime})a_\downarrow( v_{z^\prime})],
    \end{multline*}
with
    \begin{align}\label{eq: comm Q4 B1}
     & [a^\ast_\sigma(u_x)a^\ast_{\sigma^\prime}(u_y)a_{\sigma^\prime}(u_y)a_{\sigma}(u_x), a_\uparrow(u^>_z)a_\uparrow( v_z)a_\downarrow(u^>_{z^\prime})a_\downarrow( v_{z^\prime})]\nonumber
      \\
     & = -\big(\delta_{\sigma^\prime,\uparrow} u^>_{\sigma^\prime}(z;y)a^\ast_\sigma(u_x)a_\downarrow(u^>_{z^\prime}) - \delta_{\sigma^\prime,\downarrow}u^>_{\sigma^\prime}(z^\prime;y)a^\ast_\sigma(u_x)a_\uparrow(u^>_z)-\delta_{\sigma,\uparrow}u^>_\sigma(z;x) a^\ast_{\sigma^\prime}(u_y)a_\downarrow(u^>_{z^\prime}) \nonumber
     \\
     &\quad + \delta_{\sigma,\downarrow}u^>_{\sigma}(z^\prime;x) a^\ast_{\sigma^\prime}(u_y)a_\uparrow(u^>_z)+ \delta_{\sigma,\uparrow}\delta_{\sigma^\prime,\downarrow} u^>_\sigma(z;x)u^>_{\sigma^\prime}(z^\prime;y) - \delta_{\sigma,\downarrow}\delta_{\sigma^\prime,\uparrow} u^>_\sigma(z^\prime;x)u^>_{\sigma^\prime}(z;y)\big) \times \nonumber
     \\
     &\quad \quad \times a_\uparrow( v_z)a_\downarrow( v_{z^\prime})a_{\sigma^\prime}(u_y)a_\sigma(u_x).
    \end{align}
 As for the previous proposition, the proof is similar to the one in \cite[Proposition 5.4]{Gia1}, we redo it since we are using a different localization for $\hat{u}_\sigma$. In the following, we write $\xi_\lambda := T^\ast_{1;\lambda}R^\ast \psi$ for simplicity. The first four terms in \eqref{eq: comm Q4 B1} can be estimated in the same way. We consider for instance: 
\[
  \mathrm{I}  =  \int dxdydz\, V(x-y)u^>_\uparrow(z;y) \langle\xi_\lambda, a_\downarrow^\ast(u_x)b_\downarrow(\varphi_z)a_\uparrow( v_z) a_\uparrow(u_y)a_\downarrow(u_x)\xi_\lambda\rangle .
\]
We write $u^>_\uparrow(z;y) = \delta(z-y) - \alpha^<_\uparrow(z;y)$, with $\hat{\alpha}^<_\sigma$ as in \eqref{eq: def alpha<} and we split $\mathrm{I}$ in two parts: 
\begin{eqnarray*}
  \mathrm{I} &=& \int dxdy\, V(x-y) \langle\xi_\lambda, a_\downarrow^\ast(u_x)b_\downarrow(\varphi_y)a_\uparrow( v_y) a_\uparrow(u_y)a_\downarrow(u_x)\xi_\lambda\rangle
  \\
  && - \int dxdydz\, V(x-y)\alpha^<_\uparrow(z;y) \langle\xi_\lambda, a_\downarrow^\ast(u_x)b_\downarrow(\varphi_z)a_\uparrow( v_z) a_\uparrow(u_y)a_\downarrow(u_x)\xi_\lambda,\rangle = \mathrm{I}_1 + \mathrm{I}_2.
\end{eqnarray*}
We first consider $\mathrm{I}_1$. Using Lemma \ref{lem: bound b phi} to bound $\|b_\downarrow(\varphi_y)\|$ and the first of the two bounds in \eqref{eq: bounds av}, with the aid of Cauchy-Schwarz inequality, we get
\[
  |\mathrm{I}_1| \leq C\rho^{\frac{1}{2} + \frac{1}{3}}\int dxdy\, V(x-y)\|a_\downarrow (u_x)\xi_\lambda\|\|a_\uparrow(u_y)a_\downarrow(u_x)\xi_\lambda\| \leq C\|V\|_1^{\frac{1}{2}}\rho^{\frac{5}{6}} \|\mathcal{N}^{\frac{1}{2}}\xi_\lambda\|\|\mathbb{Q}_4^{\frac{1}{2}}\xi_\lambda\|\leq C\rho^{\frac{5}{6}} \|\mathcal{N}^{\frac{1}{2}}\xi_\lambda\|\|\mathbb{Q}_4^{\frac{1}{2}}\xi_\lambda\|.
\]
The estimate for $\mathrm{I}_2$ can be done similarly, using that also the bound for $\|\alpha^<_\uparrow\|_1$ in \eqref{eq: bounds alpha< nu}. Therefore, all together, we find that 
\[
  |\mathrm{I}| \leq C\rho^{\frac{5}{6}} \|\mathcal{N}^{\frac{1}{2}}\xi_\lambda\|\|\mathbb{Q}_4^{\frac{1}{2}}\xi_\lambda\|.
\]
Combining the last two contributions in \eqref{eq: comm Q4 B1}, we get
\begin{equation} 
  \mathrm{II} = -\int dxdydzdz^\prime\, V(x-y)\varphi(z-z^\prime)u^>_\uparrow(z;x)u^>_\downarrow(z^\prime;y) a_\uparrow(u_x) a_\uparrow( v_z)a_\downarrow(u_y) a_\downarrow( v_{z^\prime}).
\end{equation}
As above, we write $u^>_\sigma = \delta - \alpha^<_\sigma$ and we write  $\mathrm{II} = \mathrm{II}_0 + \mathrm{II}_{1} + \mathrm{II}_{2}$, where the index $\ast$ in $\mathrm{II}_\ast$ denotes the number of $\delta$. The main term is $\mathrm{II}_2$, which corresponds to $\mathbb{T}_2$ in \eqref{eq: def T2}. More precisely, $\mathrm{II}_2 = - \mathbb{T}_2$. The term $\mathrm{II}_1$ can be estimated by Cauchy-Schwarz and using again the first bound in \eqref{eq: bounds av} and the one for $\|\alpha_\sigma^<\|_\infty$ in \eqref{eq: bounds alpha< nu}. We indeed have 
\begin{multline*}
  |\mathrm{II}_1| \leq C\rho\int dxdydz^\prime |V(x-y)||\varphi(x-z^\prime)| |\alpha^<(z^\prime;y)|\| a_\uparrow(u_x) a_\downarrow(u_y)\xi_\lambda\|
  \\
  \leq CL^{\frac{3}{2}}\rho^2\|V\|_1\|\varphi\|_1 \|\mathbb{Q}_4^{\frac{1}{2}}\xi_\lambda\|
  \leq CL^{\frac{3}{2}}\rho^{\frac{4}{3}}\|\mathbb{Q}_4^{\frac{1}{2}}\xi_\lambda\|,
\end{multline*}
where we also used Lemma \ref{lem: bound phi}. The estimate for the term $\mathrm{II}_0$ can be done similarly using also that $\| \alpha_\sigma^<\|_1 \leq C$. Therefore we find that $|\mathrm{II}_0| \leq CL^{\frac{3}{2}}\rho^{\frac{4}{3}}\|\mathbb{Q}_4^{\frac{1}{2}}\xi_\lambda\|$.
Combining all the estimates and using the non optimal bound for the number operator in \eqref{eq: est N}, i.e., $\langle T_{1;\lambda}^\ast R^\ast\psi, \mathcal{N} T_{1;\lambda}^\ast R^\ast\psi\rangle \leq CL^3\rho^{7/6}$, we proved the estimate in \eqref{eq: est err Q4 B1}.
\end{proof}

Proceeding as in the proof of Proposition \ref{pro: Q4}, we derive propagation bounds for the operator $\mathbb{Q}_4^{\sigma,\sigma^\prime}$, defined in Proposition \ref{pro: fermionic transf}. These bounds will be used in the proof of the upper bound in Section \ref{sec: up bd}. 
\begin{proposition}\label{pro: Q4sigmasigma'}
Let $\lambda\in [0,1]$. Let $\psi$ be an approximate ground state as in Definition \ref{def: approx gs}. Under the assumptions of Theorem \ref{thm: optimal lw bd}, it holds that  
\begin{equation}
\partial_\lambda T_{1;\lambda} \mathbb{Q}_4^{\sigma,\sigma^\prime}T^\ast_{1;\lambda} =  T_{1;\lambda} (\mathbb{T}_2 + \mathcal{E}_{ \mathbb{Q}_4^{\sigma,\sigma^\prime}}) T^\ast_{1;\lambda},
\end{equation}  
with $\mathbb{T}_2$ as in \eqref{eq: def T2} and 
\begin{equation} 
  |\mathcal{E}_{ \mathbb{Q}_4^{\sigma,\sigma^\prime}}| \leq   CL^{\frac{3}{2}}\rho^{\frac{4}{3}}\| (\mathbb{Q}_4^{\sigma,\sigma^\prime})^{\frac{1}{2}}\xi_\lambda\|.
\end{equation}
\end{proposition}
Since the proof of Proposition \ref{pro: Q4sigmasigma'} can be done similarly as the one of Proposition \ref{pro: Q4sigmasigma'}, we omit the details.
We now prove some propagation estimates for the operator $\mathbb{Q}_2$ defined in \eqref{eq: def Qi}. We also consider a regularized version of it, which we will need to reconstruct the scattering equation (see Proposition \ref{pro: scatt canc}). We then define
\begin{equation}
	\mathbb{Q}_{2;>} = \frac{1}{2}\sum_{\sigma\neq\sigma^\prime} \int dxdy\, V(x-y) a^\ast_\sigma(u^>_x)a^\ast_\sigma(u_y^>)a^\ast_{\sigma^\prime}( v_y)a^\ast_{\sigma^\prime}( v_x) + \mathrm{h.c.}
\end{equation}
\begin{proposition}[Propagation estimate for $\mathbb{Q}_2, \mathbb{Q}_{2;>}$]\label{pro: Q2} Let $\lambda\in [0,1]$. Let $\psi$ be an approximate ground state as in Definition \eqref{def: approx gs}. Under the assumptions of Theorem \ref{thm: optimal lw bd}, it holds that  
 \begin{equation} 
      \partial_\lambda T_{1;\lambda} \mathbb{Q}_{2;>}T^\ast_{1;\lambda} = 2\rho_\uparrow \rho_\downarrow\int_{\Lambda_L\times \Lambda_L} dxdy V(x-y)\varphi(x-y) + T_{1;\lambda}\mathcal{E}_{\mathbb{Q}_{2;>}}T^\ast_{1;\lambda}, 
    \end{equation}
    \begin{equation} 
      \partial_\lambda\langle T^\ast_{1;\lambda}R^\ast \psi, \mathbb{Q}_2 T^\ast_{1;\lambda}R^\ast \psi\rangle = 2\rho_\uparrow \rho_\downarrow\int_{\Lambda_L\times \Lambda_L} dxdy V(x-y)\varphi(x-y) + T_{1;\lambda}\mathcal{E}_{\mathbb{Q}_2}T^\ast_{1;\lambda}, 
    \end{equation}
    with,  
    \begin{multline} 
  |\langle T^\ast_{1;\lambda}R^\ast\psi, \mathcal{E}_{\mathbb{Q}_2}T_{1;\lambda}R^\ast\psi\rangle|,\, |\langle T^\ast_{1;\lambda}R^\ast\psi,  \mathcal{E}_{\mathbb{Q}_{2;>}}T_{1;\lambda}R^\ast\psi\rangle|  \leq CL^3\rho^{\frac{7}{3}} + C\rho^{\frac{1}{2}}\|\mathbb{Q}_4^{\frac{1}{2}}T^\ast_{1;\lambda}R^\ast\psi\|\|\mathbb{H}_0^{\frac{1}{2}}T^\ast_{1;\lambda}R^\ast\psi\| 
  \\
  +C\rho\langle T^\ast_{1;\lambda}R^\ast\psi, \mathcal{N}T^\ast_{1;\lambda}R^\ast\psi\rangle +  C\rho^{\frac{2}{3}}\langle T^\ast_{1;\lambda}R^\ast\psi,\mathbb{H}_0 T^\ast_{1;\lambda}R^\ast\psi\rangle .
    \end{multline}
\end{proposition}

\begin{proof} 
We discuss the proof for the operator $\mathbb{Q}_{2;>}$, the one for $\mathbb{Q}_2$ is simpler and it uses the same ideas. As in the previous propositions, we use the notation $\xi_\lambda := T^\ast_{1;\lambda}R^\ast \psi$. We compute
  \begin{equation} 
    \partial_\lambda T_{1;\lambda} \mathbb{Q}_{2;>}T_{1;\lambda}^\ast = -T_{1;\lambda}[\mathbb{Q}_{2;>}, B_1]T^\ast_{1;\lambda} + \mathrm{h.c.}
  \end{equation}
  Following the calculations in \cite[Proposition 5.5]{Gia1}, we get 
  \begin{multline} 
    [\mathbb{Q}_{2;>},B_1]
    \\ = \frac{1}{2}\sum_{\sigma\neq \sigma^\prime} \int dxdydzdz^\prime\, V(x-y)\varphi(z-z^\prime)[a^\ast_{\sigma}(u_x^>)a^\ast_{\sigma^\prime}(u_y^>)a^\ast_{\sigma^\prime}( v_y)a^\ast_\sigma( v_x), a_\uparrow(u^>_z)a_\uparrow( v_z)a_\downarrow(u^>_{z^\prime})a_\downarrow( v_{z^\prime})],
  \end{multline}
  with
  \begin{align} \label{eq: comm Q4}
    &[a^\ast_{\sigma}(u_x^>)a^\ast_{\sigma^\prime}(u_y^>)a^\ast_{\sigma^\prime}( v_y)a^\ast_\sigma( v_x), a_\uparrow(u^>_z)a_\uparrow( v_z)a_\downarrow(u^>_{z^\prime})a_\downarrow( v_{z^\prime})]
    \\
    &= -a^\ast_{\sigma}(u_x^>)a^\ast_{\sigma^\prime}(u_y^>)\big( \delta_{\sigma,\uparrow}v_\sigma(x;z)a^\ast_{\sigma^\prime}( v_y)a_\downarrow( v_{z^\prime}) - \delta_{\sigma, \downarrow} v_\sigma(x;z^\prime)a^\ast_{\sigma^\prime}( v_y)a_\uparrow( v_z) - \delta_{\sigma^\prime, \uparrow} v_{\sigma^\prime}(y;z) a^\ast_{\sigma}( v_x) a_\downarrow( v_{z^\prime}) \nonumber
    \\
    &+ \delta_{\sigma^\prime, \downarrow} v_{\sigma^\prime}(y;z^\prime)a^\ast_\sigma( v_x)a_\uparrow( v_{z}) + \delta_{\sigma^\prime, \uparrow}\delta_{\sigma,\downarrow} v_{\sigma^\prime}(y;z) v_{\sigma}(x;z^\prime) - \delta_{\sigma^\prime,\downarrow}\delta_{\sigma,\uparrow} v_{\sigma^\prime}(y;z^\prime) v_{\sigma}(x;z)\big)a_\uparrow(u^>_z)a_\downarrow(u^>_{z^\prime})\nonumber
    \\
    &\quad - a_\uparrow(v_z) a_\downarrow(v_{z^\prime}) \big( \delta_{\sigma, \uparrow}\delta_{\sigma^\prime, \downarrow}(u^>_\downarrow)^2(z^\prime;y)(u^>_\uparrow)^2(z;x)  - \delta_{\sigma, \downarrow}\delta_{\sigma^\prime, \uparrow}(u^>_\uparrow)^2(z;y)(u^>_\downarrow)^2(z^\prime;x) - (u^>_\uparrow)^2(z;x) a^\ast_\downarrow(u^>_y)a_\downarrow(u^>_{z^\prime}) \nonumber
    \\
    &+ (u^>_\downarrow)^2(z^\prime;x) a^\ast_\uparrow(u^>_y) a_\uparrow(u^>_z)+ (u^>_\uparrow)^2(z;y)a^\ast_\downarrow(u^>_x)a_\downarrow(u^>_{z^\prime}) -(u^>_\downarrow)^2(z^\prime;y)a^\ast_\uparrow(u^>_x)a_\uparrow(u^>_z)\big)
a^\ast_{\sigma^\prime}( v_y)a^\ast_\sigma(v_x).\nonumber
  \end{align}
From the first two lines in the right hand side above, we find two types of error terms. The first one is 
\begin{equation} \label{eq: term I Q4}
  \mathrm{I}_a:= \int dxdydzdz^\prime\, V(x-y)\varphi(z-z^\prime)v_\uparrow(x;z) \langle \xi_\lambda, a_\uparrow^\ast(u_x^>)a_\downarrow^\ast(u_y^>) a^\ast_\downarrow( v_y)a_\downarrow( v_{z^\prime}) a_\downarrow(u_{z^\prime}^>)a_\uparrow(u_{z}^>)\xi_\lambda\rangle.
\end{equation}
Following ideas similar to those introduced in \cite[Proposition 6.1]{FGHP} and in \cite[Proposition 5.5]{Gia1}, we first replace $u^>_x$, $u^>_y$ by $u_x$ and $u_y$, respectively. More precisely, we decompose $u_\sigma = u^>_\sigma + u^<_\sigma$, where $u^<_\sigma$ is the periodization of a smooth function on $\mathbb{R}^3$ satisfying 
\begin{equation}\label{eq: definition u<}
  \mathcal{F}(\hat{u}^<_\sigma)(k) = \begin{cases} 1 &\mbox{if}\,\,\, k_F^\sigma < |k| < 2k_F^\sigma, \\ 0&\mbox{if}\,\,\, |k| > 3k_F^\sigma.\end{cases}
\end{equation}
We can then rewrite $\mathrm{I}_a$ as: 
\begin{equation}
  \mathrm{I}_a = \int dxdydz\, V(x-y)v_\uparrow(x;z) \langle \xi_\lambda, a_\uparrow^\ast(u_x)a_\downarrow^\ast(u_y) a^\ast_\downarrow( v_y)b_\downarrow(\varphi_z)a_\uparrow(u_{z})\xi_\lambda\rangle + \mathrm{I}_{a;2} =: \mathrm{I}_{a;1} + {\mathrm{I}}_{a;2},
\end{equation}
where $\mathrm{I}_{a;2}$ is an error term. We start by estimating $\mathrm{I}_{a;1}$. We can write 
\[
  \mathrm{I}_{a;1} = \frac{1}{L^3}\sum_k \hat{v}_\uparrow(k) \left\langle\left(\int dxdy\, e^{-ik\cdot x} V(x-y) a_\downarrow( v_y) a_\downarrow(u_y) a_\uparrow(u_x)\xi_\lambda\right), \left(\int dz\, e^{-ik\cdot z}\, b_\downarrow(\varphi_z)a_\uparrow(u_z^>)\xi_\lambda\right) \right\rangle.
\]
Therefore, by Cauchy-Schwarz's inequality and using that $|\hat{v}_\sigma(k)|\leq 1$, we get 
\[
  |\mathrm{I}_{a;1}| \leq\sqrt{\frac{1}{L^3}\sum_k\left\|\int dxdy\, e^{-ik\cdot x} V(x-y) a_\downarrow( v_y) a_\downarrow(u_y) a_\uparrow(u_x)\xi_\lambda\right\|^2}\sqrt{\frac{1}{L^3}\sum_k\left\| \int dz\, e^{-ik\cdot z}\, b_\downarrow(\varphi_z)a_\uparrow(u_z^>)\xi_\lambda\right\|^2}.
\]
The first term above can be estimated as follows by Cauchy-Schwarz again:
\[
  \frac{1}{L^3}\sum_k\left\|\int dxdy\, e^{-ik\cdot x} V(x-y) a_\downarrow( v_y) a_\downarrow(u_y) a_\uparrow(u_x)\xi_\lambda\right\|^2 \leq C\rho \|V\|_1 \langle\xi_\lambda, \mathbb{Q}_4 \xi_\lambda\rangle.
\]
Moreover, using Lemma \ref{lem: bound b phi} together with \eqref{eq: N>} and \eqref{eq: u> wrt N>}, we get
\begin{equation}\label{eq: CS Ia int dz Q2}
  \frac{1}{L^3}\sum_k\left\| \int dz\, e^{-ik\cdot z}\, b_\downarrow(\varphi_z)a_\uparrow(u_z^>)\xi_\lambda\right\|^2 = \int dz \|b_\downarrow(\varphi_z)a_\uparrow(u^>_z)\xi_\lambda\|^2 \leq C\rho^{\frac{2}{3}}\langle \xi_\lambda, \mathcal{N}_> \xi_\lambda\rangle \leq C\langle \xi_\lambda, \mathbb{H}_0 \xi_\lambda\rangle.
\end{equation}
All together, we find that 
\[
  |\mathrm{I}_{a;1} | \leq C\|V\|_1^{\frac{1}{2}}\rho^{\frac{1}{2}}\|\mathbb{Q}_4^{\frac{1}{2}}\xi_\lambda\|\|\mathbb{H}_0^{\frac{1}{2}}\xi_\lambda\| \leq  C\rho^{\frac{1}{2}}\|\mathbb{Q}_4^{\frac{1}{2}}\xi_\lambda\|\|\mathbb{H}_0^{\frac{1}{2}}\xi_\lambda\|.
\]
We now consider $\mathrm{I}_{a;2}$: there are two types of terms, with one or two $a_\sigma(u^<_{\cdot})$. Both can be estimated similarly as $\mathrm{I}_{a;1}$. More precisely, proceeding as above, we can write 
\begin{align*}
  &\left| \int dxdydzdz^\prime\, V(x-y)\varphi(z-z^\prime)v_\uparrow(x;z)\langle \xi_\lambda, a^\ast_\uparrow(u^<_x)a^\ast_{\downarrow}(u_y) a^\ast_{\downarrow}( v_y)a_\downarrow ( v_{z^\prime})a_\downarrow(u^>_{z^\prime})a_\uparrow(u^>_{z})\xi_\lambda\rangle\right|
  \\
  &\qquad \leq \frac{1}{L^3}\sum_k\left\|\int dxdy\, e^{-ik\cdot x} V(x-y)  a_\downarrow( v_y) a_\downarrow(u_y) a_\uparrow(u^<_x)\xi_\lambda\right\|\left\| \int dz\, e^{-ik\cdot z}\, b_\downarrow(\varphi_z)a_\uparrow(u_z^>)\xi_\lambda\right\|
  \\
  &\qquad \leq C\rho \|\mathcal{N}^{\frac{1}{2}}\xi_\lambda\|\|\mathbb{H}_0^{\frac{1}{2}}\xi_\lambda\|,
\end{align*}
where we also used 
\begin{equation}\label{eq: east au<}
  \|a_\sigma(u^<_\cdot)\|\leq C\rho^{1/2}, 
\end{equation}
due to the support properties of $\hat{u}^{<}_\sigma$. All the other terms in $\mathrm{I}_{a;2}$ can be estimated similarly, we omit the details. We find that 
\[
  |\mathrm{I}_a| \leq C\|V\|_1^{\frac{1}{2}}\rho^{\frac{1}{2}}\|\mathbb{Q}_4^{\frac{1}{2}}\xi_\lambda\|\|\mathbb{H}_0^{\frac{1}{2}}\xi_\lambda\| + C \|V\|_1\rho \|\mathcal{N}^{\frac{1}{2}}\xi_\lambda\|\|\mathbb{H}_0^{\frac{1}{2}}\xi_\lambda\| \leq C\rho^{\frac{1}{2}}\|\mathbb{Q}_4^{\frac{1}{2}}\xi_\lambda\|\|\mathbb{H}_0^{\frac{1}{2}}\xi_\lambda\| + C\rho \|\mathcal{N}^{\frac{1}{2}}\xi_\lambda\|\|\mathbb{H}_0^{\frac{1}{2}}\xi_\lambda\|.
\]
The other type of error term coming from the first two lines in \eqref{eq: comm Q4} is 
\begin{equation} \label{eq: term Ib Q4}
  \mathrm{I}_b:= \int dxdydzdz^\prime\, V(x-y)\varphi(z-z^\prime)v_\uparrow(x;z)v_\downarrow(y;z^\prime)\langle \xi_\lambda, a^\ast_\uparrow(u^>_x)a^\ast_{\downarrow}(u^>_y) a_\uparrow(u^>_{z})a_\downarrow(u^>_{z^\prime})\xi_\lambda\rangle.
\end{equation}
To bound this error term, we replace $a_\uparrow(u^>_x)a_\downarrow(u^>_y)$ by $a_\uparrow(u_x)a_\downarrow(u_y)$, similarly as above. We write 
\[
  \mathrm{I}_{b} = \int dxdydzdz^\prime\, V(x-y)\varphi(z-z^\prime)v_\uparrow(x;z)v_\downarrow(y;z^\prime)\langle \xi_\lambda, a^\ast_\uparrow(u_x)a^\ast_{\downarrow}(u_y) a_\uparrow(u^>_{z})a_\downarrow(u^>_{z^\prime})\xi_\lambda\rangle + \mathrm{I}_{b;2} =: \mathrm{I}_{b;1} + \mathrm{I}_{b;2}.
\]
We first estimate $\mathrm{I}_{b;1}$. We now proceed similarly as for the term in \eqref{eq: term Ia kin err} in Proposition \ref{pro: H0}. We write
\[
  \mathrm{I}_{b;1} = \int dxdydz^\prime\, V(x-y)v_\downarrow(y;z^\prime)\left\langle a_\uparrow(u_x)a_{\downarrow}(u_y)\xi_\lambda, \left(\int dz\, \varphi(z-z^\prime)v_\uparrow(x;z) a_\uparrow(u^>_{z})\right) a_\downarrow(u^>_{z^\prime})\xi_\lambda\right\rangle.
\]
By Cauchy-Schwarz's inequality and using \eqref{eq: conv L2 norm}, we find that 
\begin{multline*}
  |\mathrm{I}_{b;1}| \leq \left(\int dxdydz^\prime V(x-y)|v_\downarrow(y;z^\prime)|^2 \|a_\uparrow(u_x)a_{\downarrow}(u_y)\xi_\lambda\|^2\right)^{\frac{1}{2}}\times 
  \\
  \times \left(\int dxdydz^\prime V(x-y) \left\|\int dz\, \varphi(z-z^\prime)v_\uparrow(x;z) a_\uparrow(u^>_{z})\right\|^2\|a_\downarrow(u^>_{z^\prime})\xi_\lambda\|^2 \right)^{\frac{1}{2}}.
\end{multline*}
Thus
\[
  |\mathrm{I}_{b;1}| \leq C \|V\|_1^{\frac{1}{2}}\|\varphi\|_2\|v_\uparrow\|_2\|v_\downarrow\|_2 \|\mathbb{Q}_4^{\frac{1}{2}}\xi_\lambda\|\|\mathcal{N}_>^{\frac{1}{2}}\xi_\lambda\| \leq C\rho^{1-\frac{1}{6}}\|\mathbb{Q}_4^{\frac{1}{2}}\xi_\lambda\|\|\mathcal{N}_>^{\frac{1}{2}}\xi_\lambda\|\leq  C\rho^{\frac{1}{2}}\|\mathbb{Q}_4^{\frac{1}{2}}\xi_\lambda\|\|\mathbb{H}_0^{\frac{1}{2}} \xi_\lambda\|,
\]
using also \eqref{eq: u> wrt N>} and \eqref{eq: N>}. To estimate $\mathrm{I}_{b;2}$ we can proceed similarly. In $\mathrm{I}_{b;2}$ there are terms with one or two $a_\sigma(u^<_\cdot)$. Proceeding as for the term $\mathrm{I}_{b;1}$ and using also \eqref{eq: east au<} together with 
\[
  \int dx\, \|a_\sigma(u_x)\xi_\lambda\|^2\leq C\langle \xi_\lambda, \mathcal{N}\xi_\lambda\rangle, \qquad \int dx\, \|a_\sigma(u^<_x)\xi_\lambda\|^2\leq C\langle \xi_\lambda, \mathcal{N}\xi_\lambda\rangle, 
\]
we get that 
\[
  |\mathrm{I}_{b;2}| \leq C\rho\|\mathcal{N}^{\frac{1}{2}}\xi_\lambda\|\|\mathbb{H}_0^{\frac{1}{2}}\xi_\lambda\|,
\]
which implies that 
\[
  |\mathrm{I}_b| \leq C\rho^{\frac{1}{2}}\|\mathbb{Q}_4^{\frac{1}{2}}\xi_\lambda\|\|\mathbb{H}_0^{\frac{1}{2}}\xi_\lambda\| + C\rho\|\mathcal{N}^{\frac{1}{2}}\xi_\lambda\|\|\mathbb{H}_0^{\frac{1}{2}}\xi_\lambda\|.
\]
We now consider the error terms coming from the last two lines in the right hand side in \eqref{eq: comm Q4}. We have to estimate two types of error terms. The first type is, for instance,
\begin{equation}\label{eq: IIb Q2}
  \mathrm{II}_a = \int dxdydzdz^\prime\, V(x-y)\varphi(z-z^\prime)(u^>_\uparrow)^2(z;x) \langle \xi_\lambda, a^\ast_\downarrow(u^>_y)a_\downarrow(u^>_{z^\prime}) a_\uparrow( v_z)a_\downarrow( v_{z^\prime}) a^\ast_{\downarrow}( v_y)a^\ast_\uparrow( v_x)\xi_\lambda\rangle.
\end{equation}
Similarly as in previous propositions, we write $(u^>_\sigma)^2 = \delta - \nu_\sigma$. We have
\begin{eqnarray*}
  |\mathrm{II}_a| &\leq& \int dxdydz^\prime\, |V(x-y)||\varphi(x-z^\prime)| |\langle \xi_\lambda, a^\ast_\downarrow(u^>_y)a_\downarrow(u^>_{z^\prime}) a_\uparrow( v_x)a_\downarrow( v_{z^\prime}) a^\ast_{\downarrow}( v_y)a^\ast_\uparrow( v_x)\xi_\lambda\rangle| 
  \\
  && + \int dxdydzdz^\prime\, |V(x-y)||\varphi(z-z^\prime)||\nu_\uparrow(z;x)|| \langle \xi_\lambda, a^\ast_\downarrow(u^>_y)a_\downarrow(u^>_{z^\prime}) a_\uparrow( v_z)a_\downarrow( v_{z^\prime}) a^\ast_{\downarrow}( v_y)a^\ast_\uparrow( v_x)\xi_\lambda\rangle|
  \\
  &\leq& C\rho^2\int dxdydz^\prime\, |V(x-y)||\varphi(z-z^\prime)| (\|a_\downarrow(u^>_y)\xi_\lambda\|^2 + \|a_\downarrow(u^>_{z^\prime})\xi_\lambda\|^2)
  \\
  && + C\rho^2\int dxdydz^\prime\, |V(x-y)||\varphi(z-z^\prime)| |\nu_\uparrow(z;x)| (\|a_\downarrow(u^>_y)\xi_\lambda\|^2 + \|a_\downarrow(u^>_{z^\prime})\xi_\lambda\|^2)
  \\
  &\leq& C\rho^{2}\|V\|_1\|\varphi\|_1 \langle \xi_\lambda,\mathcal{N}_>\xi_\lambda\rangle +  C\rho^{2}\|V\|_1\|\varphi\|_1 \|\nu_\uparrow\|_1\langle \xi_\lambda,\mathcal{N}_>\xi_\lambda\rangle \leq C\rho^{2}\|\varphi\|_1 \langle \xi_\lambda,\mathcal{N}_>\xi_\lambda\rangle ,
\end{eqnarray*}
where we used that $\|a_\uparrow( v_z)a_\downarrow( v_{z^\prime})a^\ast_{\downarrow}( v_y)a^\ast_\uparrow( v_x)\|\leq C\rho^2$ and the bound for $\|\nu_\sigma\|_1$ in \eqref{eq: bounds alpha< nu}. Therefore, using also \eqref{eq: N>}, we find that
\[
  |\mathrm{II}_a|\leq C\rho^{\frac{2}{3}}\langle \xi_\lambda,\mathbb{H}_0 \xi_\lambda\rangle.
\]
Note that in the estimate for the same error term for $\mathbb{Q}_{2}$, we would have $a_\downarrow(u_y)$ in place of $a_\downarrow(u_y^>)$. As a consequence, for $\mathbb{Q}_{2}$, the estimate for $\mathrm{II}_a$ is $C\rho\|\mathcal{N}^{\frac{1}{2}}\xi_\lambda\|\|\mathbb{H}_0^{\frac{1}{2}}\xi_\lambda\|$. The last four terms in the right hand side of \eqref{eq: comm Q4} can be all estimated as $\mathrm{II}_a$. We now consider the first two terms in the third line of the right hand side of \eqref{eq: comm Q4}. After a change of variable and using that $V(x-y) = V(y-x)$, we can rewrite $\mathrm{II}_b$ as
\begin{equation}\label{eq: term IIa Q2}
  \mathrm{II}_b = -2\int dxdydzdz^\prime\,  V(x-y)\varphi(z-z^\prime)(u^>_\downarrow)^2(z^\prime;y)(u^>_\uparrow)^2(z;x) \langle \xi_\lambda, a_\uparrow( v_z)a_\downarrow( v_{z^\prime})a^\ast_{\downarrow}( v_y)a^\ast_\uparrow( v_x)\xi_\lambda\rangle.
\end{equation}
To extract the main contribution we partially put $\mathrm{II}_a$ in normal order. Therefore, we get 
\begin{eqnarray*}
  \mathrm{II}_b &=&  2\int dxdydzdz^\prime\,  V(x-y)\varphi(z-z^\prime)(u^>_\downarrow)^2(z^\prime;y)(u^>_\uparrow)^2(z;x) \langle \xi_\lambda, a^\ast_{\downarrow}( v_y)a_\uparrow( v_z)a^\ast_\uparrow( v_x)a_\downarrow( v_{z^\prime})\xi_\lambda\rangle
  \\
  && 2\int dxdydzdz^\prime\,  V(x-y)\varphi(z-z^\prime)(u^>_\downarrow)^2(z^\prime;y)(u^>_\uparrow)^2(z;x) v_\downarrow(z^\prime;y)\langle \xi_\lambda, a^\ast_\uparrow( v_x) a_\uparrow( v_z)\xi_\lambda\rangle
  \\
  && - 2\int dxdydzdz^\prime\,  V(x-y)\varphi(z-z^\prime)(u^>_\downarrow)^2(z^\prime;y)(u^>_\uparrow)^2(z;x) v_\downarrow(z^\prime;y) v_\uparrow(z;x)
\\
&=:& \mathrm{II}_{b;1} + \mathrm{II}_{b;2} + \mathrm{I}_{\mathrm{main}}.
\end{eqnarray*}
From the term $\mathrm{I}_{\mathrm{main}}$, we will extract the main constant; the other terms instead are errors which we estimate below. We start by considering $\mathrm{II}_{b;1}$. To estimate it, similarly as above, we write $(u^>_\sigma)^2 = \delta - \nu_\sigma$, with $\nu_\sigma$ as in \eqref{eq: def nu}. Correspondingly, we have $\mathrm{II}_{b;1} = \mathrm{II}_{b;1;0} + \mathrm{II}_{b;1;1} + \mathrm{II}_{b;1;2}$, where $\mathrm{II}_{b;1;\ast}$ is such that there are $\ast \in \{0,1,2\}$ $\delta$ in the term. We start by estimating $\mathrm{II}_{b;1;2}$. Using that $\|a_\sigma( v_x)a^\ast_\sigma(v_x)\|\leq C\rho$ and Lemma \ref{lem: bound phi}, we can then estimate 
\begin{multline*}
  |\mathrm{II}_{b;1;2}| \leq 2\int dxdy\, |V(x-y)||\varphi(x-y)|| \langle \xi_\lambda, a^\ast_{\downarrow}( v_y)a_\uparrow( v_x)a^\ast_\uparrow( v_x)a_\downarrow( v_{y})\xi_\lambda\rangle| 
  \\
  \leq C\rho \int dxdy\, |V(x-y)||\varphi(x-y)||\| a_{\downarrow}( v_y)\xi_\lambda\|^2\leq C\rho\|V\|_1\|\varphi\|_\infty \langle\xi_\lambda,\mathcal{N}\xi_\lambda\rangle \leq C\rho\langle\xi_\lambda,\mathcal{N}\xi_\lambda\rangle.
\end{multline*}
For the term $\mathrm{II}_{b;1;1}$, we have
\begin{multline}
  |\mathrm{II}_{b;1;1}|\leq C\rho \int dxdydz^\prime |V(x-y)||\varphi(x-z^\prime)||\nu_\downarrow(z^\prime;y)(\| a_{\downarrow}( v_y)\xi_\lambda\|^2 +\|a_\downarrow( v_{z^\prime})\xi_\lambda\|^2) 
  \\
  \leq C\rho\|V\|_1 \|\varphi\|_1 \|\nu_\downarrow\|_\infty\langle \xi_\lambda,\mathcal{N}\xi_\lambda\rangle \leq C\rho^{\frac{4}{3}}\langle \xi_\lambda,\mathcal{N}\xi_\lambda\rangle,
\end{multline}
where we used again Lemma \ref{lem: bound phi} together with the bound for $\|\nu_\sigma\|_\infty$ in \eqref{eq: bounds alpha< nu}. Proceeding in a similar way and using also that $\|\nu_\sigma\|_1 \leq C$ (see \eqref{eq: bounds alpha< nu}), we get that 
\[
  |\mathrm{II}_{b;1;0}| \leq C\rho \|V\|_1\|\varphi\|_1  \|\nu_\uparrow\|_1 \|\nu_\downarrow\|_\infty \langle \xi_\lambda,\mathcal{N}\xi_\lambda\rangle \leq C\rho^{\frac{4}{3}}\langle \xi_\lambda,\mathcal{N}\xi_\lambda\rangle.
\]
Therefore, we find $|\mathrm{II}_{b;1}| \leq C\rho\langle \xi_\lambda,\mathcal{N}\xi_\lambda\rangle$.
The term $\mathrm{II}_{b;2}$ can be treated in a similar way writing $(u_\sigma^>)^2 = \delta - \nu_\sigma$ and using that $v_\uparrow(x;x) = \rho$, $\|v_\sigma\|_\infty \leq C\rho$ (see \eqref{eq: bounds av}) in place of $\|a_\uparrow(v_x)a_\uparrow^\ast (v_x)\|, \|a_\uparrow(v_x)a^\ast_\uparrow(v_z)\|\leq C\rho$, we omit the details. We have 
\[
  |\mathrm{II}_{b;2}| \leq C\rho\langle \xi_\lambda,\mathcal{N}\xi_\lambda\rangle.
\]
To conclude the proof, we now consider the constant term, which is given by 
\begin{equation} \label{eq: I main Q4}
  \mathrm{I}_{\mathrm{main}}  = - 2 \int dxdydzdz^\prime V(x-y) \varphi(z-z^\prime) (u^>_\uparrow)^2(z;x)(u^>_\downarrow)^2(z^\prime;y)v_\uparrow(x;z)v_\downarrow(y;z^\prime).
\end{equation}
As above, we write $(u^>_\sigma)^2 = \delta - \nu_\sigma$. We then have $\mathrm{I}_{\mathrm{main}} = \mathrm{I}_{\mathrm{main}; \ast}$ with $\ast\in \{0,1,2\}$ denoting the number of $\delta$ in the term. Using again the bounds in \eqref{eq: bounds alpha< nu}, we can prove that 
\[
  |\mathrm{I}_{\mathrm{main}; 1} + \mathrm{I}_{\mathrm{main}; 0} | \leq C\rho^2\int dxdydz^\prime |V(x-y)| |\varphi(z-z^\prime)| |\nu_\downarrow(z^\prime;y)|\leq CL^3\rho^2\|V\|_1 \|\varphi\|_1 \|\nu_\downarrow\|_\infty \leq CL^3\rho^{\frac{7}{3}},
\]
which implies
\begin{align*}
  \mathrm{I}_{\mathrm{main}} &=  -2 \int dxdydzdz^\prime V(x-y) \varphi(z-z^\prime)v_\uparrow(x;x)v_\downarrow(y;y) + \mathcal{E}
  \\
  & = -2\rho_\uparrow\rho_\downarrow \int dxdy\,  V(x-y) \varphi(x-y) + \mathcal{E}, \qquad\qquad |\mathcal{E}|\leq CL^3\rho^{\frac{7}{3}}.
\end{align*}
Note that to deal with $\mathrm{II}_a$ and $\mathrm{II}_b$ for $\mathbb{Q}_2$, we can proceed similarly as above with $(u^>_\sigma)^2$ replaced by $u^>_\sigma$ (see Remark \ref{rem: comparisono Gia1 cut-off}).
Putting all the estimates together, we conclude the proof.
\end{proof}

\subsection{Scattering equation cancellation}\label{sec: scattering}
In the conjugation under $T_1$, we need to use the scattering equation cancellation which we discuss in this section. To do that, we need to regularize both $\mathbb{T}_2$ and $\mathbb{Q}_{2}$, to combine these operators with $\mathbb{T}_1$. We then define
\begin{equation}\label{eq: def T2r}
  \mathbb{T}_2^> = \int dxdy\, V(x-y)\varphi(x-y) a_\uparrow(u_x^>)a_\uparrow( v_x)a_\downarrow(u_y^>) a_\downarrow( v_y)+ \mathrm{h.c.},
\end{equation}
and 
\begin{equation}\label{eq: def T2<}
  \mathbb{T}_2^< = \int dxdy\, V(x-y)\varphi(x-y) a_\uparrow(u_x^<)a_\uparrow( v_x)a_\downarrow(u_y^<) a_\downarrow( v_y)+ \mathrm{h.c.},
\end{equation}
with $u^<_\sigma$ as in \eqref{eq: definition u<}, i.e., $u^<_\sigma$ is the periodization of a smooth function $\mathbb{R}^3 \rightarrow \mathbb{R}$ and such that 
\begin{equation}\label{eq: def u<}
\mathcal{F}({u}^<_\sigma)(k) = \begin{cases} 1 &\mbox{if}\,\,\, k_F^\sigma\leq |k| \leq 2k_F^\sigma, \\ 0 &\mbox{if}\,\,\,  |k| \geq 3k_F^\sigma.\end{cases}
\end{equation}
We also recall that 
\begin{equation}
  \mathbb{Q}_{2;>} = \int dxdy\, V(x-y) a_\uparrow(u_x^>)a_\uparrow ( v_x)a_\downarrow (u_y^>)  a_\downarrow( v_y) + \mathrm{h.c.},
\end{equation}
and we define 
\begin{equation}\label{eq: Q2<}
   \mathbb{Q}_{2;<} = \int dxdy\, V(x-y) a_\uparrow(u_x^<)a_\uparrow ( v_x)a_\downarrow (u_y^<)  a_\downarrow( v_y) + \mathrm{h.c.}
\end{equation}
\begin{proposition}\label{pro: reg} Let $\lambda\in [0,1]$. Let $\psi$ be an approximate ground state as in Definition \eqref{def: approx gs}. Under the assumptions of Theorem \ref{thm: optimal lw bd}, it holds that  
\begin{equation}\label{eq: T2reg up bd}
  \mathbb{T}_2  - \mathbb{T}_2^> = \mathbb{T}_2^< + \mathfrak{E}_{\mathbb{T}_2},\qquad \mathbb{Q}_2 - \mathbb{Q}_{2;>} = \mathbb{Q}_{2;<} + \mathfrak{E}_{\mathbb{Q}_2},
\end{equation}
with 
\begin{equation}\label{eq: error T2}
  |\langle T^\ast_{1;\lambda}R^\ast\psi, \mathfrak{E}_{\mathbb{T}_2}T^\ast_{1;\lambda}R^\ast\psi\rangle| \leq CL^{\frac{3}{2}}\rho^{\frac{3}{2}}\|\mathcal{N}_>^{\frac{1}{2}}\xi_\lambda\| \leq  CL^3\rho^{\frac{7}{3}} + \delta\langle T^\ast_{1}R^\ast \psi, \mathbb{H}_0 T^\ast_{1}R^\ast \psi\rangle,
\end{equation}
and 
\begin{equation}\label{eq: error Q2}
  |\langle T^\ast_{1;\lambda}R^\ast\psi, \mathfrak{E}_{\mathbb{Q}_2}T^\ast_{1;\lambda}R^\ast\psi\rangle| \leq CL^{\frac{3}{2}}\rho^{\frac{3}{2}}\|\mathcal{N}_>^{\frac{1}{2}}\xi_\lambda\|\leq   CL^3\rho^{\frac{7}{3}} + \delta\langle T^\ast_{1}R^\ast \psi, \mathbb{H}_0 T^\ast_{1}R^\ast \psi\rangle,
\end{equation}
for any $0<\delta<1$.
\end{proposition}
\begin{remark}
The main difference with respect to \cite{Gia1} lies in the fact that, in order to obtain a lower bound, we do not estimate $\mathbb{T}_2^<$ and $\mathbb{Q}_{2;<}$, but rather introduce a new unitary transformation to handle these terms (see Section \ref{sec: T2 conjugation}).
\end{remark}
\begin{proof}
  We start by writing 
  \[
    \hat{u}_\sigma(k) = \hat{u}^>_\sigma(k) + \hat{u}^<_\sigma(k), 
   \]
   with $\hat{u}^>_\sigma$, $\hat{u}^<_\sigma$ as in \eqref{eq: def u>} and \eqref{eq: definition u<}, respectively. Therefore, 
\begin{align*}
  \mathbb{T}_2 - \mathbb{T}_2^> &= \int dxdy\, V(x-y)\varphi(x-y) a_\uparrow(u_x^<)a_\uparrow( v_x)a_\downarrow(u_y^<) a_\downarrow( v_y)+ \mathrm{h.c.}
  \\
  & \quad +\int dxdy\, V(x-y)\varphi(x-y) a_\uparrow(u_x^>)a_\uparrow( v_x)a_\downarrow(u_y^<) a_\downarrow( v_y)+ \mathrm{h.c.}
  \\
  & \quad +\int dxdy\, V(x-y)\varphi(x-y) a_\uparrow(u_x^<)a_\uparrow( v_x)a_\downarrow(u_y^>) a_\downarrow( v_y)+ \mathrm{h.c.}= \mathrm{\mathbb{T}}_2^< +\mathcal{E}_{\mathbb{T}_2},
\end{align*}
with $\mathbb{T}_2^<$ as in \eqref{eq: def T2<}.
Both the last two operators can be estimated in the same way. As in the previous proofs, we write $\xi_\lambda$ in place of $T^{\ast}_{1;\lambda}R^\ast\psi$. We consider for instance the last term in the right hand side above. Using that $\|a_\sigma( v_\cdot)\|, \|a_\sigma(u^<_\cdot)\|\leq C\rho^{1/2}$ (see \eqref{eq: bounds av} and \eqref{eq: east au<}), with the aid of Cauchy-Schwarz, we get
\begin{multline*}
  \left|\int dxdy\, V(x-y)\varphi(x-y)\langle \xi_\lambda, a_\uparrow(u_x^<)a_\uparrow( v_x)a_\downarrow(u_y^>) a_\downarrow( v_y)\xi_\lambda\rangle \right|
  \\
  \leq C\rho^{\frac{3}{2}}\int dxdy\, V(x-y)\varphi(x-y)\|a_\downarrow(u_y^>)\xi_\lambda\|\leq CL^{\frac{3}{2}}\rho^{\frac{3}{2}}\|V\|_1\|\varphi\|_\infty \|\mathcal{N}_>^{\frac{1}{2}}\xi_\lambda\|.
\end{multline*}
Therefore, From Lemma \ref{lem: bound phi} together with \eqref{eq: N> lambda 1}, we find that 
\[
  |\langle \xi_\lambda, \mathcal{E}_{\mathbb{T}_2} \xi_\lambda\rangle \leq CL^{\frac{3}{2}}\|V\|_1\rho^{\frac{3}{2}}\|\mathcal{N}_>^{\frac{1}{2}}\xi_\lambda\| \leq  C L^{\frac{3}{2}} \|V\|_1\rho^{\frac{3}{2} - \frac{1}{3}}\|\mathbb{H}_0^{\frac{1}{2}}\xi_1\|  \leq CL^3\rho^{\frac{7}{3}} + \delta\langle \xi_1, \mathbb{H}_0 \xi_1\rangle.
\]
The proof for $\mathbb{Q}_2 - \mathbb{Q}_{2;>}$ can be done in the same way, we omit the details.
\end{proof}
\begin{proposition}[Scattering equation cancellation]\label{pro: scatt canc} Let $\lambda\in [0,1]$. Let $\psi$ be an approximate ground state as in Definition \eqref{def: approx gs}. Under the assumptions of Theorem \ref{thm: optimal lw bd}, it holds that
\begin{equation}\label{eq: est err scatt eq canc}
  |\langle T^\ast_{1;\lambda}R^\ast \psi, (-\mathbb{T}_1 - \mathbb{T}_2^> + \mathbb{Q}_{2;>})T^\ast_{1;\lambda}R^\ast \psi \rangle|\leq CL^{\frac{3}{2}}\rho^{\frac{3}{2}}\|\mathcal{N}_>^{\frac{1}{2}} T^\ast_{1;\lambda}R^\ast \psi\| \leq  CL^3 \rho^{\frac{7}{3}} + \delta\langle T^\ast_{1}R^\ast \psi, \mathbb{H}_0 T^\ast_{1}R^\ast \psi\rangle,
\end{equation}
for any $0< \delta <1$.
\end{proposition}
\begin{proof} The proof is the same as the one in \cite[Proposition 5.8]{Gia1}. We recall here the main steps. We write $\xi_\lambda:= T^\ast_{1;\lambda}R^\ast\psi$. Similarly as in \cite[Proposition 5.8]{Gia1}, taking $L$ large enough, we find 
\[
  \langle\xi_\lambda, (-\mathbb{T}_1 - \mathbb{T}_2^> + \mathbb{Q}_{2;>})\xi_\lambda \rangle 
  = \int dy \langle \xi_\lambda, b_\uparrow((\mathcal{E}_{\varphi_0, \chi_{\sqrt[3]\rho}})_y) a_\downarrow(u^>_y)a_\downarrow( v_y)\xi_\lambda\rangle + \mathrm{c.c.}
\]
where $\mathcal{E}_{\varphi_0, \chi_{\sqrt[3]\rho}}(x-y)$ is as in \eqref{eq: period error scattering} and where we used the notation introduced in \eqref{eq: def op b g}.
Using the bounds in \eqref{eq: bounds av} and \eqref{eq: bound b err scatt}, by Cauchy-Schwarz, we find for any $0<\delta<1$
\begin{align*}
  | \langle\xi_\lambda, (-\mathbb{T}_1 - \mathbb{T}_2^> + \mathbb{Q}_{2;>})\xi_\lambda \rangle| &\leq \int dy\, \|b_\uparrow((\mathcal{E}_{\varphi_0, \chi_{\sqrt[3]\rho}})_y)\|\|a_\downarrow( v_y)\|\|a_\downarrow(u^>_y)\xi_\lambda\| \leq  CL^{\frac{3}{2}}\rho^{\frac{3}{2}}\|\mathcal{N}_>^{\frac{1}{2}}\xi_\lambda\|
  \\
  &\leq CL^{\frac{3}{2}}\rho^{\frac{3}{2} -\frac{1}{3}}\|\mathbb{H}_0^{\frac{1}{2}}\xi_1\| + CL^3\rho^{\frac{7}{3}} \leq CL^3\rho^{\frac{7}{3}} + \delta\langle \xi_1 , \mathbb{H}_0 + \xi_1 \rangle,
\end{align*}
where we used also \eqref{eq: N> lambda 1}.
\end{proof}
We now consider the conjugation of the operator $\mathbb{T}_2^<$, defined in \eqref{eq: def T2<}, by $T_1$. This step is required before conjugating the operator by the transformation $T_2$ in Section \ref{sec: T2 unitary}.
\begin{proposition}\label{pro: T2<}
Let $\lambda\in [0,1]$. Let $\psi$ be an approximate ground state as in Definition \eqref{def: approx gs}. Under the  assumptions of Theorem \ref{thm: optimal lw bd}, it holds that  
 \begin{equation}\label{eq: prop T2<} 
     \left| \partial_\lambda\langle T^\ast_{1;\lambda}R^\ast \psi, \mathbb{T}_{2}^<T^\ast_{1;\lambda}R^\ast \psi\rangle \right|\leq CL^3\rho^{\frac{7}{3}} + C\rho\|\mathcal{N}^{\frac{1}{2}}T^\ast_{1;\lambda}R^\ast \psi\|\|\mathbb{H}_0^{\frac{1}{2}}T^\ast_{1;\lambda}R^\ast \psi\| + C\rho^{\frac{4}{3}}\langle T^\ast_{1;\lambda}R^\ast \psi, \mathcal{N}T^\ast_{1;\lambda}R^\ast \psi\rangle.
\end{equation}
\end{proposition}
\begin{proof}
As in the previous proofs, we will write $\xi_\lambda$ in place of $T^\ast_{1;\lambda}R^\ast\psi$. Moreover, to have a shorter notation, we write $\widetilde{V}$ to denote $V\varphi$; we will use several times that $\|\widetilde{V}\|_1 \leq \|V\|_1\|\varphi\|_\infty \leq C$. We now compute 
\begin{equation}\label{eq: deriv lambda T2<}
  \partial_\lambda\langle \xi_\lambda, \mathbb{T}^<_2 \xi_\lambda\rangle = -\langle \xi_\lambda, [\mathbb{T}_{2}^<, B_1]\xi_\lambda\rangle + \mathrm{c.c.}
\end{equation}
The structure of the error terms is the same as those in Proposition \ref{pro: Q2}. We omit some details and we directly estimate them. The first error term corresponds to \eqref{eq: term I Q4} and it is 
\[
  \mathrm{I}_a =\int dxdydz\, \widetilde{V}(x-y)\varphi(z-z^\prime)v_\uparrow(x;z) \langle \xi_\lambda, a_\uparrow^\ast(u_x^<)a_\downarrow^\ast(u_y^<) a^\ast_\downarrow( v_y)b_\downarrow(\varphi_z) a_\uparrow(u_{z}^>)\xi_\lambda\rangle.
\]
Proceeding similarly as for the term \eqref{eq: term I Q4} in Proposition \ref{pro: Q2}, we can then write 
\[
  \mathrm{I}_{a} \leq \sqrt{\frac{1}{L^3}\sum_k\left\|\int dxdy\, e^{-ik\cdot x} \widetilde{V}(x-y) a_\downarrow( v_y) a_\downarrow(u_y^<) a_\uparrow(u_x^<)\xi_\lambda\right\|^2}\sqrt{\frac{1}{L^3}\sum_k\left\| \int dz\, e^{-ik\cdot z}\, b_\downarrow(\varphi_z)a_\uparrow(u_z^>)\xi_\lambda\right\|^2}.
\]
Using now \eqref{eq: CS Ia int dz Q2} together with $\|a_\sigma(u^<_\cdot)\|\leq C\rho^{1/2}$ and
\begin{equation}\label{eq: int vv 6as T2<}
\frac{1}{L^3}\sum_k\left\|\int dxdy\, e^{-ik\cdot x} \widetilde{V}(x-y) a_\downarrow( v_y) a_\downarrow(u_y^<) a_\uparrow(u_x^<)\xi_\lambda\right\|^2
\leq C\rho^2 \|\widetilde{V}\|_1^2 \langle\xi_\lambda, \mathcal{N} \xi_\lambda\rangle,
\end{equation}
we get 
\[
  |\mathrm{I}_a| \leq C\|V\|_1\rho\|\mathcal{N}^{\frac{1}{2}}\xi_\lambda\|\|\mathbb{H}_0^{\frac{1}{2}}\xi_\lambda\| \leq C\rho\|\mathcal{N}^{\frac{1}{2}}\xi_\lambda\|\|\mathbb{H}_0^{\frac{1}{2}}\xi_\lambda\|.
\]
Another possible error term is the one corresponding to \eqref{eq: term Ib Q4}, i.e., 
\begin{equation} 
  \mathrm{I}_b = \int dxdydzdz^\prime\, \widetilde{V}(x-y)\varphi(z-z^\prime)v_\uparrow(x;z)v_\downarrow(y;z^\prime)\langle \xi_\lambda, a^\ast_\uparrow(u^<_x)a^\ast_{\downarrow}(u^<_y) a_\downarrow(u^>_{z^\prime})a_\uparrow(u^>_{z})\xi_\lambda\rangle.
\end{equation}
We can proceed as for \eqref{eq: term Ib Q4} in Proposition \ref{pro: Q2}. By Cauchy-Schwarz inequality and \eqref{eq: conv L2 norm}, we find that 
\begin{multline*}
  |\mathrm{I}_{b}| \leq \left(\int dxdydz^\prime\, \widetilde{V}(x-y)|v_\downarrow(y;z^\prime)|^2 \|a_\uparrow(u_x^<)a_{\downarrow}(u_y^<)\xi_\lambda\|^2\right)^{\frac{1}{2}}\times 
  \\
  \times \left(\int dxdydz^\prime\, \widetilde{V}(x-y) \left\|\int dz\, \varphi(z-z^\prime)v_\uparrow(x;z) a_\uparrow(u^>_{z})\right\|^2\|a_\downarrow(u^>_{z^\prime})\xi_\lambda\|^2 \right)^{\frac{1}{2}}.
\end{multline*}
Thus,
\[
  |\mathrm{I}_{b}| \leq C \|\widetilde{V}\|_1\|\varphi\|_2\|v_\uparrow\|_2\|v_\downarrow\|_2 \|u_\uparrow^<\|_2\|\mathcal{N}^{\frac{1}{2}}\xi_\lambda\|\|\mathcal{N}_>^{\frac{1}{2}}\xi_\lambda\| \leq C\rho^{1+\frac{1}{2}-\frac{1}{6}}\|\mathcal{N}^{\frac{1}{2}}\xi_\lambda\|\|\mathcal{N}_>^{\frac{1}{2}}\xi_\lambda\|\leq  C\rho\|\mathcal{N}^{\frac{1}{2}}\xi_\lambda\|\|\mathbb{H}_0^{\frac{1}{2}} \xi_\lambda\|,
\]
using also \eqref{eq: u> wrt N>} and \eqref{eq: N>}. The next error term we consider, corresponds to \eqref{eq: IIb Q2} in Proposition \ref{pro: Q2}. Explicitly:
\begin{eqnarray*}
  \mathrm{II}_a &=& \int dxdydzdz^\prime\, \widetilde{V}(x-y)\varphi(z-z^\prime)u^{<,>}_\uparrow(z;x) \langle \xi_\lambda, a^\ast_\downarrow(u^<_y)a_\downarrow(u^>_{z^\prime}) a_\uparrow( v_z)a_\downarrow( v_{z^\prime}) a^\ast_{\downarrow}( v_y)a^\ast_\uparrow( v_x)\xi_\lambda\rangle,
\end{eqnarray*}
where we wrote $u^{<,>}_\sigma$ to denote: 
\[
  u^{<,>}_\sigma(x;y) = \frac{1}{L^3}\sum_k \hat{u}^<_\sigma(k)\hat{u}^>_\sigma(k) e^{ik\cdot(x-y)}.
\]
To estimate $\mathrm{II}_a$, it is convenient to rewrite it in normal order. We have
\begin{eqnarray*}
\mathrm{II}_a &=& \int dxdydzdz^\prime\, \widetilde{V}(x-y)\varphi(z-z^\prime)u^{<,>}_\uparrow(z;x) \bigg\{\langle \xi_\lambda, a^\ast_\downarrow(u^<_y) a^\ast_{\downarrow}( v_y)a^\ast_\uparrow( v_x)a_\downarrow(u^>_{z^\prime}) a_\uparrow( v_z)a_\downarrow( v_{z^\prime}) \xi_\lambda\rangle
\\
&&\qquad  + v_\downarrow(z^\prime;y)v_\uparrow(z;x)\langle \xi_\lambda, a^\ast_\downarrow(u^<_y)a_\downarrow(u^>_{z^\prime})  \xi_\lambda\rangle - v_\downarrow(z^\prime;y)\langle \xi_\lambda, a^\ast_\downarrow(u^<_y)a_\downarrow(u^>_{z^\prime})a^\ast_\uparrow(v_x)a_\uparrow(v_z)  \xi_\lambda\rangle
\\
&&\qquad + v_\uparrow(z;x)\langle \xi_\lambda, a^\ast_\downarrow(u^<_y)a^\ast_\downarrow(v_y)a_\downarrow(u^>_{z^\prime})a_\downarrow(v_{z^\prime})  \xi_\lambda\rangle\bigg\} = \mathrm{II}_{a;1} + \mathrm{II}_{a;2} + \mathrm{II}_{a;3} + \mathrm{II}_{a;4}. 
\end{eqnarray*}
We now estimate all of them. Similarly as for the term $\mathrm{I}_a$ above, we can use that $|\hat{u}^{<}_\sigma(k) \hat{u}^>_\sigma(k)| \leq 1$ together with Cauchy-Schwarz's inequality and write 
\begin{multline}
  |\mathrm{II}_{a;1}| 
  \leq \sqrt{\frac{1}{L^3}\sum_k \left \| \int dxdy\, e^{ik\cdot y} \widetilde{V}(x-y) a_\downarrow(u^<_y) a_\uparrow( v_x) a_{\downarrow}( v_y)\xi_\lambda \right\|^2}\times 
  \\
  \times\sqrt{\frac{1}{L^3}\sum_k\left\|\int dzdz^\prime e^{ik\cdot z}\varphi(z-z^\prime)a_\downarrow(u^>_{z^\prime})a_\uparrow( v_z)a_\downarrow( v_{z^\prime})\xi_\lambda\right\|^2}.
\end{multline}
Using then the first bound in \eqref{eq: bounds av} and proceeding similarly as in \eqref{eq: int vv 6as T2<} for both the terms in the right hand side, we get 
\[
  |\mathrm{II}_{a;1}| \leq C\rho^2 \|V\|_1\|\varphi\|_1 \|\mathcal{N}^{\frac{1}{2}}\xi_\lambda\|\|\mathcal{N}_>^{\frac{1}{2}}\xi_\lambda\|\leq C\rho^2 \|\varphi\|_1 \|\mathcal{N}^{\frac{1}{2}}\xi_\lambda\|\|\mathcal{N}_>^{\frac{1}{2}}\xi_\lambda\|\leq C\rho\|\mathcal{N}^{\frac{1}{2}}\xi_\lambda\|\|\mathbb{H}_0^{\frac{1}{2}}\xi_\lambda\|,
\]
where we also used Lemma \ref{lem: bound phi} and \eqref{eq: N>}.
We now estimate $\mathrm{II}_{b;2}$. Using the bounds in \eqref{eq: bounds av}, with the aid of Cauchy-Schwarz inequality, we get
\begin{eqnarray*}
  |\mathrm{II}_{a;2}| &\leq& C\rho \left(\int dxdydzdz^\prime, \widetilde{V}(x-y) |\varphi(z-z^\prime)||u^{<,>}_\uparrow(z;x)|| v_\downarrow(z^\prime;y)|\|a_\downarrow(u^<_y)\xi_\lambda\|^2\right)^{\frac{1}{2}}\times 
  \\
  && \times \left( \int dxdydzdz^\prime, \widetilde{V}(x-y) |\varphi(z-z^\prime)||u^{<,>}_\uparrow(z;x)|| v_\downarrow(z^\prime;y)| \|a_\downarrow(u^>_{z^\prime})\xi_\lambda\|^2\right)^{\frac{1}{2}}
  \\
  &\leq& C\rho \|\widetilde{V}\|_1 \|\varphi\|_1 \|u^{<,>}_\uparrow\|_2 \|v_\downarrow\|_2 \|\mathcal{N}^{\frac{1}{2}}\xi_\lambda\| \|\mathcal{N}_>^{\frac{1}{2}} \xi_\lambda\| \leq C\rho^{\frac{4}{3}}\|\mathcal{N}^{\frac{1}{2}}\xi_\lambda\| \|\mathcal{N}_>^{\frac{1}{2}} \xi_\lambda\| \leq C\rho\|\mathcal{N}^{\frac{1}{2}}\xi_\lambda\| \|\mathbb{H}_0^{\frac{1}{2}} \xi_\lambda\|,
\end{eqnarray*}
where we used $\|u^{<,>}_\uparrow\|_2 \leq C\rho^{1/2}$, Lemma \ref{lem: bound phi} and 
\begin{equation}\label{eq: CS Vuv}
  \int dxdy\, f(x-y)||u^{<,>}_\uparrow(z;x)|| v_\downarrow(z^\prime;y)|\leq \|\widetilde{V}\|_1\|u^{<,>}_\uparrow\|_2 \|v_\downarrow\|_2,
\end{equation}
with $f= \widetilde{V}$ or $f=\varphi$.
The estimate for $\mathrm{II}_{a;3}$ can be done in the same way, using $\|a_\sigma(v_\cdot)a_\sigma^\ast(v_\cdot)\|\leq C\rho$ in place of $\|v_\sigma\|_\infty\leq C\rho$. We omit the details. For the term $\mathrm{II}_{a;4}$, we can use also that 
\[
  \left|\int dx \, u^{<,>}_\uparrow(z;x) v_\uparrow(z;x)\right| = \left|\int dz\, u^{<,>}_\uparrow(z;x) v_\uparrow(z;x)\right|\leq \|u^{<,>}_\uparrow\|_2 \|v_\uparrow\|_2,
\]
and conclude that 
\[
  |\mathrm{II}_{a;4}| \leq C\rho \|\mathcal{N}^{\frac{1}{2}}\xi_\lambda\|\|\mathbb{H}_0^{\frac{1}{2}}\xi_\lambda\|,
\]
which implies  
\[
  |\mathrm{II}_a| \leq C\rho \|\mathcal{N}^{\frac{1}{2}}\xi_\lambda\|\|\mathbb{H}_0^{\frac{1}{2}}\xi_\lambda\|.
\]
We now consider the error term which corresponds to \eqref{eq: term IIa Q2}: 
\[
  \mathrm{II}_b = -2\int dxdydzdz^\prime\,  V(x-y)\varphi(z-z^\prime)u^{<,>}_\downarrow(z^\prime;y)u^{<,>}_\uparrow(z;x) \langle \xi_\lambda, a_\uparrow( v_z)a_\downarrow( v_{z^\prime})a^\ast_{\downarrow}( v_y)a^\ast_\uparrow( v_x)\xi_\lambda\rangle.
\]
We can then put $\mathrm{II}_a$ partially in normal order and get, as for $\eqref{eq: term IIa Q2}$: 
\begin{eqnarray*}
  \mathrm{II}_b &=& 2\int dxdydzdz^\prime\,  V(x-y)\varphi(z-z^\prime)u^{<,>}_\downarrow(z^\prime;y)u^{<,>}_\uparrow(z;x) \bigg\{\langle \xi_\lambda, a^\ast_{\downarrow}( v_y)a_\uparrow( v_z)a^\ast_\uparrow( v_x)a_\downarrow( v_{z^\prime})\xi_\lambda\rangle
  \\
  && \qquad + v_\downarrow(z^\prime;y)\langle \xi_\lambda, a^\ast_\uparrow( v_x) a_\uparrow( v_z)\xi_\lambda\rangle
  -v_\downarrow(z^\prime;y) v_\uparrow(z;x)\bigg\}=: \mathrm{II}_{b;1} + \mathrm{II}_{b;2} + \mathrm{I}_{\mathrm{main}}.
\end{eqnarray*}
We  start to estimate $\mathrm{II}_b$ by considering $\mathrm{II}_{b;1}$. In the following, we are going to use that $\hat{u}^{<,>}$ is supported only for momenta in $2k_F^\sigma < |k| < 3k_F^\sigma$. We have 
\begin{eqnarray*}
  |\mathrm{II}_{b;1}| &\leq&  C\rho\int dxdydzdz^\prime\, |\widetilde{V}(x-y)||\varphi(z-z^\prime)||u^{<,>}_\downarrow(z^\prime;y)||u^{<,>}_\uparrow(z;x)|(\|a_{\downarrow}( v_y)\xi_\lambda\|^2 + \|a_\downarrow( v_{z^\prime})\xi_\lambda\|^2)
  \\
  &\leq& C\rho\|u^{<,>}_\uparrow\|_2 \|u^{<,>}_\downarrow\|_2 \|\varphi\|_1\|\widetilde{V}\|_1\langle \xi_\lambda,\mathcal{N}\xi_\lambda\rangle \leq C\|V\|_1\rho^{\frac{4}{3}}\langle \xi_\lambda,\mathcal{N}\xi_\lambda\rangle \leq C\rho^{\frac{4}{3}}\langle \xi_\lambda,\mathcal{N}\xi_\lambda\rangle,
\end{eqnarray*}
where we also used the first bound in \eqref{eq: bounds av}, Lemma \ref{lem: bound phi}, together with analogous estimates to the one in \eqref{eq: CS Vuv}.
Next, we consider $\mathrm{II}_{b;2}$. Using again $\|v_\downarrow\|_\infty \leq C\rho$ in place of $\|a_\uparrow( v_z)a^\ast_\uparrow( v_x)\|\leq C\rho$, we can proceed as for $\mathrm{II}_{b;1}$ and get:
\[
  |\mathrm{II}_{b;2}|  \leq  C\rho^{\frac{4}{3}}\langle \xi_\lambda,\mathcal{N}\xi_\lambda\rangle.
\]
To conclude the proof, we need to estimate the constant term:
\[
  \mathrm{I}_{\mathrm{main}}  = -2 \int dxdydzdz^\prime \, \widetilde{V}(x-y) \varphi(z-z^\prime) u^{<,>}_\uparrow(z;x)u^{<,>}_\downarrow(z^\prime;y)v_\uparrow(x;z)v_\downarrow(y;z^\prime).
\]
Using again the bounds in \eqref{eq: CS Vuv} and \eqref{eq: bounds av}, we get 
\[
  |\mathrm{I}_{\mathrm{main}}| \leq CL^3\rho^2 \|\widetilde{V}\|_1\|\varphi\|_1 \|u^{<,>}_\uparrow\|_2\|u^{<,>}_\downarrow\|_2 \leq C\|V\|_1L^3 \rho^{\frac{7}{3}} \leq CL^3\rho^{\frac{7}{3}}.
\]
Putting all the estimates together, we can conclude the proof.
\end{proof}
\subsection{Propagation of the estimates: conclusions}
We are now in a position to establish propagation estimates for the operators $\mathbb{H}_0$ and $\mathbb{Q}_4$, which will be used to control the error terms appearing in Propositions \ref{pro: H0}, \ref{pro: Q4} and \ref{pro: Q2} and to prove the bound \eqref{eq: N> lambda 1} in Proposition \ref{pro: N}.
\begin{proposition}\label{pro: propagation est}
  Let $\lambda\in [0,1]$. Let $\psi$ be an approximate ground state as in Definition \ref{def: approx gs}. Under the  assumptions of Theorem \ref{thm: optimal lw bd}, it holds that
  \begin{equation}\label{eq: propagation H0 and Q1}
   \langle T^\ast_{1;\lambda}R^\ast \psi,\mathbb{H}_0T^\ast_{1;\lambda}R^\ast \psi\rangle \leq CL^3\rho^2, \qquad  \langle T^\ast_{1;\lambda}R^\ast \psi,\mathbb{Q}_{4}T^\ast_{1;\lambda}R^\ast \psi\rangle \leq CL^3\rho^2.
  \end{equation}
\end{proposition}
\begin{proof} 
The proof is similar to the one in \cite[Proposition 5.9]{Gia1} and \cite[Proposition 5.8]{FGHP}, we recall the main steps.  We write $\xi_\lambda$ in place of $T^\ast_{1;\lambda}R^\ast\psi$.
  From Proposition \ref{pro: H0} and Proposition \ref{pro: Q4},  we have
  \begin{eqnarray}
    \partial_\lambda\langle\xi_\lambda, (\mathbb{H}_0 + \mathbb{Q}_{4})\xi_\lambda\rangle  &=&  \langle\xi_\lambda, (\mathbb{T}_1 + \mathbb{T}_2)\xi_\lambda\rangle + \langle \xi_\lambda, \mathcal{E}_{\mathbb{H}_0}\xi_\lambda\rangle + \langle \xi_\lambda, \mathcal{E}_{\mathbb{Q}_{4}}\xi_\lambda\rangle
    \\
    &=&  \langle\xi_\lambda, (\mathbb{T}_1 + \mathbb{T}_2^>)\xi_\lambda\rangle - \langle\xi_\lambda (\mathbb{T}_2^> - \mathbb{T}_2)\xi_\lambda\rangle + \langle \xi_\lambda, \mathcal{E}_{\mathbb{H}_0}\xi_\lambda\rangle + \langle \xi_\lambda, \mathcal{E}_{\mathbb{Q}_{4}}\xi_\lambda\rangle.\nonumber
      \end{eqnarray}
We now estimate the four quantities in the right hand side above. From Proposition \ref{pro: scatt canc}, we know that 
\[
  |\langle \xi_\lambda, (\mathbb{T}_1 + \mathbb{T}_2^>)\xi_\lambda\rangle | \leq C|\langle \xi_\lambda, \mathbb{Q}_{2;>}\xi_\lambda\rangle | + CL^{\frac{3}{2}}\rho^{\frac{3}{2}}\|\mathcal{N}_>^{\frac{1}{2}}\xi_\lambda\|\leq C|\langle \xi_\lambda, \mathbb{Q}_{2;>}\xi_\lambda\rangle | + C\langle \xi_\lambda,\mathbb{H}_0 \xi_\lambda\rangle + CL^3\rho^{\frac{7}{3}},
\]
where we also used \eqref{eq: N>} and Cauchy-Schwarz. Moreover, from Proposition \ref{pro: reg}, we get
\begin{eqnarray}
        |\langle\xi_\lambda, \mathbb{Q}_{2;>}\xi_\lambda\rangle| &\leq& |\langle\xi_\lambda, (\mathbb{Q}_{2;>} - \mathbb{Q}_2)\xi_\lambda\rangle| + |\langle\xi_\lambda, \mathbb{Q}_2\xi_\lambda\rangle|\leq |\langle \xi_\lambda, \mathbb{Q}_{2;<}\xi_\lambda\rangle| + CL^{\frac{3}{2}}\rho^{\frac{3}{2}}\|\mathcal{N}_>^{\frac{1}{2}}\xi_\lambda\| + C\langle \xi_\lambda, \mathbb{Q}_2 \xi_\lambda\rangle  \nonumber
        \\
        &\leq& |\langle \xi_\lambda, \mathbb{Q}_{2;<}\xi_\lambda\rangle| + |\langle \xi_\lambda, \mathbb{Q}_{2}\xi_\lambda\rangle| + CL^3\rho^{\frac{7}{3}} +  C\langle \xi_\lambda, \mathbb{H}_0 \xi_\lambda\rangle,\nonumber
        \end{eqnarray} 
where in the last estimate we used \eqref{eq: N>} and Cauchy-Schwarz. It is then easy to see by Cauchy-Schwarz's inequality and using that $\|a_\sigma( v_\cdot)\|, \|a_\sigma(u^<_\cdot)\|\leq C\rho^{1/2}$, that 
\[
  |\langle\xi_\lambda, \mathbb{Q}_{2;<} \xi_\lambda\rangle |\leq C \int dxdy\, |V(x-y)||\langle a_\uparrow(u_x^<)a_\uparrow( v_x)a_\downarrow(u_y^<) a_\downarrow( v_y)\xi_\lambda\rangle \leq CL^{\frac{3}{2}}\rho^{\frac{3}{2}} \|V\|_1  \|\mathcal{N}^{\frac{1}{2}}\xi_\lambda\|\leq CL^3\rho^{2+\frac{1}{12}},
\]
\[
  |\langle\xi_\lambda, \mathbb{Q}_{2} \xi_\lambda\rangle |\leq C \int dxdy\, |V(x-y)||\langle a_\uparrow(u_x)a_\uparrow( v_x)a_\downarrow(u_y) a_\downarrow( v_y)\xi_\lambda\rangle \leq CL^{\frac{3}{2}}\rho \|\mathbb{Q}_4^{\frac{1}{2}}\xi_\lambda\|\leq CL^3\rho^2 + \langle \xi_\lambda, \mathbb{Q}_4 \xi_\lambda\rangle.
\]
Note the in the first estimate above, we used the non-optimal estimate $\langle \xi_\lambda,\mathcal{N}\xi_\lambda\rangle \leq CL^3\rho^{7/6}$ (see \eqref{eq: est N}).
All together, we then find that 
\[
  |\langle \xi_\lambda, (\mathbb{T}_1 + \mathbb{T}_2^>)\xi_\lambda\rangle | \leq CL^3\rho^{2} + \langle \xi_\lambda,(\mathbb{H}_0 + \mathbb{Q}_4) \xi_\lambda\rangle.
\]
    From Proposition \ref{pro: reg} togehter with \eqref{eq: N>}, we get
    \begin{equation}
        |\langle \xi_\lambda, (\mathbb{T}_2^>  - \mathbb{T}_2)\xi_\lambda\rangle | \leq  |\langle\xi_\lambda, \mathbb{T}_2^< \xi_\lambda\rangle | + CL^{\frac{3}{2}}\rho^{\frac{3}{2}}\|\mathcal{N}^{\frac{1}{2}}_>\xi_\lambda\| \leq CL^3\rho^{2+\frac{1}{12}} + C\langle\xi_\lambda, \mathbb{H}_0 \xi_\lambda\rangle,
\end{equation} 
where we also estimated $|\langle\xi_\lambda, \mathbb{T}_2^< \xi_\lambda\rangle |\leq C L^3\rho^{2+1/12}$ proceeding similarly as for $\mathbb{Q}_{2;<}$ and using that $\|\varphi\|_\infty \leq C$ (see Lemma \ref{lem: bound phi}).
The term $\langle \xi_\lambda, \mathcal{E}_{\mathbb{H}_0}\xi_\lambda\rangle$ can be estimated as in \cite[Eq. (5.127)]{Gia1}, we omit the details:
\[
  |\langle \xi_\lambda, \mathcal{E}_{\mathbb{H}_0}\xi_\lambda\rangle|\leq CL^3\rho^2 + C\langle \xi_\lambda,\mathbb{H}_0 \xi_\lambda\rangle.
\]
Moreover, from Proposition \ref{pro: Q4}, we get
\[
        |\langle\xi_\lambda, \mathcal{E}_{\mathbb{Q}_{4}}\xi_\lambda\rangle| \leq CL^{\frac{3}{2}}\rho^{\frac{4}{3}}\|\mathbb{Q}^{\frac{1}{2}}_{4}\xi_\lambda\|\leq CL^3\rho^{\frac{8}{3}} + C\langle \xi_\lambda,  \mathbb{Q}_{4}\xi_\lambda\rangle. 
\]
Combining all the bounds together, we have
\begin{equation} 
\partial_\lambda\langle \xi_\lambda, (\mathbb{H}_0 + \mathbb{Q}_{4})\xi_\lambda\rangle \leq CL^3\rho^{2} + \langle\xi_\lambda, (\mathbb{H}_0 +\mathbb{Q}_{4})\xi_\lambda\rangle.
\end{equation}
By Gr\"onwall's Lemma and the a priori bounds in Lemma \ref{lem: a priori est}, we conclude the proof. 
\end{proof}
\begin{remark}[Propagation estimates for $\mathbb{Q}_4^{\sigma, \sigma^\prime}$]\label{rem: prop Q4sigmasigma'} Let $\mathbb{Q}_4^{\sigma,\sigma^\prime}$ be the operator defined in Proposition \ref{pro: fermionic transf}. Proceeding as in Proposition \ref{pro: propagation est} (this time considering $\mathbb{H}_0 + \mathbb{Q}_4^{\sigma,\sigma^\prime}$ in place of  $\mathbb{H}_0 + \mathbb{Q}_4$) and using Proposition \ref{pro: Q4sigmasigma'} together with Lemma \ref{lem: a priori est}, we obtain that for any $\lambda\in [0,1]$, 
\[
|\langle T^\ast_{1;\lambda} R^\ast \psi, \mathbb{Q}_{4}^{\sigma, \sigma^\prime}T^\ast_{1;\lambda} R^\ast \psi\rangle| \leq CL^3\rho^{2},
\]
for any approximate ground state $\psi$.
\end{remark}
\subsection{Estimate of the cubic term $\mathbb{Q}_3$}\label{sec: Q3}
\textcolor{black}{The main goal of this section is to show that the cubic term $\mathbb{Q}_3$ defined in \eqref{eq: def Qi} is of order $\mathcal{O}(L^3\rho^{7/3})$, up to some positive contributions. The analysis closely follows that of \cite[Section 7.2]{Gia1}. We nevertheless repeat the argument here, since we employ a  different cut-off for $\hat{u}^>_\sigma$ and do not introduce any cut-off for $\hat{v}_\sigma$, see Remark \ref{rem: comparisono Gia1 cut-off}. In the next proposition we estimate $|\langle R^\ast\psi, \mathbb{Q}_3 R^\ast\psi\rangle |$.}
\begin{proposition}[Estimate for $\mathbb{Q}_3$] \label{pro: Q3 final} Let $\psi$ be an approximate ground state as in  Definition \ref{def: approx gs}. Under the  assumptions of Theorem \ref{thm: optimal lw bd}, it holds that 
\begin{eqnarray*}
  |\langle R^\ast\psi, \mathbb{Q}_3 R^\ast\psi\rangle | &\leq&  CL^3 \rho^{\frac{7}{3}} + \delta\langle T^\ast_1 R^\ast \psi, (\mathbb{H}_0 +\mathbb{Q}_4)T^\ast_1 R^\ast \psi\rangle + C\rho\langle R^\ast\psi, \mathcal{N} R^\ast \psi\rangle +  C\rho\langle T^\ast_1 R^\ast\psi, \mathcal{N} T^\ast_1 R^\ast \psi\rangle
  \\
  && +C\int_0^1 d\lambda\,  \left(\rho\langle T^\ast_{1;\lambda}R^\ast \psi,\mathcal{N}T^\ast_{1;\lambda} R^\ast \psi \rangle + \rho^{\frac{1}{2}}\|\mathbb{Q}_4^{\frac{1}{2}}T^\ast_{1;\lambda}R^\ast \psi\|\|\mathbb{H}_0^{\frac{1}{2}}T^\ast_{1;\lambda}R^\ast \psi\|\right),
\end{eqnarray*} 
for any $0<\delta<1$.
\end{proposition}
\begin{proof} 
\textcolor{black}{Before starting the proof, recall that
\begin{equation}\label{eq: def Q3}
  \mathbb{Q}_3 = -\sum_{\sigma\neq\sigma^\prime} \int\, dxdy\, V(x-y)\left(a^\ast_\sigma(u_x) a^\ast_{\sigma^\prime}(u_y) a^\ast_{\sigma}( v_x)a_{\sigma^\prime}(u_y) - a^\ast_\sigma(u_x) a^\ast_{\sigma^\prime}( v_y) a^\ast_{\sigma}( v_x)a_{\sigma^\prime}( v_y)\right) + \mathrm{h.c.} \nonumber
\end{equation}
We split $\mathbb{Q}_{3}$ in two parts by writing $a_{\sigma^\prime}(u_y) = a_{\sigma^\prime}(\mathfrak{u}^{<}_y) + a_{\sigma^\prime}(\mathfrak{u}_y^{>})$, with $\hat{\mathfrak{u}}^{<}_\sigma(k)$ and $\hat{\mathfrak{u}}^{>}_\sigma(k)$ such that 
\begin{equation}\label{eq: def mathfrak u<>}
  \hat{\mathfrak{u}}^<_\sigma(k) = \begin{cases} 1 &k\notin\mathcal{B}_F^\sigma, \, |k| < (3/2)k_F^\sigma, \\ 0 &|k| \geq (3/2)k_F^\sigma. \end{cases}, \qquad \hat{\mathfrak{u}}^{>}_\sigma(k) = \begin{cases} 0 &k\notin\mathcal{B}_F^\sigma, \, |k| < (3/2)k_F^\sigma, \\ 1 &|k| \geq (3/2)k_F^\sigma.\end{cases}
\end{equation}
Thus, we write $\mathbb{Q}_3 = {\mathbb{Q}}_{3;1}^< + \widetilde{\mathbb{Q}}_{3}$, with
\begin{equation}\label{eq: Q 31 <}
  \mathbb{Q}_{3;1}^< = -\sum_{\sigma\neq\sigma^\prime}\int\, dxdy\, V(x-y) \langle R^\ast \psi, a^\ast_\sigma(u_x) a^\ast_{\sigma^\prime}(u_y) a^\ast_\sigma( v_x) a_{\sigma^\prime}(\mathfrak{u}_y^{<})R^\ast\psi\rangle.
\end{equation}}
 The estimates for $\widetilde{\mathbb{Q}}_{3}$ can be carried out as in \cite[Proposition 7.2]{Gia1}. We therefore omit the details and directly state the final resulting bound, which holds for any $0<\delta<1$:  
\begin{equation}\label{eq: est Q32 & Q31>}
  |\langle R^\ast\psi, \widetilde{\mathbb{Q}}_{3} R^\ast\psi\rangle \leq CL^3\rho^{\frac{7}{3}} + C\rho\langle R^\ast\psi,\mathcal{N}R^\ast\psi\rangle + \delta\langle T^\ast_1 R^\ast \psi,\mathbb{H}_0 T^\ast_1 R^\ast\psi\rangle.
\end{equation}
To estimate $\langle R^\ast \psi, \mathbb{Q}_{3;1}^< R^\ast \psi\rangle$, we need a more refined analysis than the one for $\widetilde{\mathbb{Q}}_{3}$, similarly as in \cite[Proposition 7.2]{Gia1}. By Duhamel's formula, and using the notation $\xi_\lambda = T^\ast_{1;\lambda}R^\ast \psi$, we have
  \begin{equation}\label{eq: Duhamel Q31<}
    \langle R^\ast\psi, \mathbb{Q}_{3;1}^< R^\ast \psi\rangle = \langle T^\ast_1 R^\ast\psi, \mathbb{Q}_{3;1}^< T^\ast R^\ast\psi\rangle - \int_0^1\, d\lambda\, \partial_\lambda\langle \xi_\lambda, \mathbb{Q}_{3;1}^< \xi_\lambda\rangle.
  \end{equation}
  The first contribution in the right side above can be estimated by Cauchy-Schwarz, as in \cite[Proposition 7.2]{Gia1}. More precisely, for any $0<\delta<1$, we have that
  \begin{eqnarray}
    \langle T^\ast_1 R^\ast\psi, \mathbb{Q}_{3;1}^< T^\ast R^\ast\psi\rangle \leq C\rho^{\frac{1}{2}}\|\mathbb{Q}_{4}^{\frac{1}{2}}T^\ast_1 R^\ast \psi\|\|\mathcal{N}^{\frac{1}{2}}T^\ast_1 R^\ast \psi\| \nonumber
    \leq \delta\langle T^\ast_1 R^\ast \psi, \mathbb{Q}_{4} T^\ast_1 R^\ast \psi \rangle + C\rho\langle T^\ast_1 R^\ast \psi,\mathcal{N} T^\ast_1 R^\ast \psi\rangle.
  \end{eqnarray}
  Since $\partial_\lambda\langle \xi_\lambda, \mathbb{Q}_{3;1}^< \xi_\lambda\rangle = -\langle \xi_\lambda, [\mathbb{Q}_{3;1}^< , B_1 - B_1^\ast]\xi_\lambda\rangle$,  we now compute:
  \begin{equation}\label{eq: two parts comm Q3}
    [\mathbb{Q}_{3;1}^<, B_1-B^\ast_1]   = -\int\, dxdydzdz^\prime\, V(x-y)\varphi(z-z^\prime) [a^\ast_\sigma(u_x) a^\ast_{\sigma^\prime}(u_y) a^\ast_\sigma( v_x) a_{\sigma^\prime}(\mathfrak{u}_y^<), a_\uparrow(u^>_z)a_\uparrow( v_z)a_\downarrow(u^>_{z^\prime})a_\downarrow( v_{z^\prime})],
  \end{equation} 
 We have
  \begin{align}\label{eq: 1 part comm Q3}
    &[a^\ast_\sigma(u_x) a^\ast_{\sigma^\prime}(u_y) a^\ast_\sigma( v_x) a_{\sigma^\prime}(\mathfrak{u}_y^<), a_\uparrow(u^>_z)a_\uparrow( v_z)a_\downarrow(u^>_{z^\prime})a_\downarrow( v_{z^\prime})] 
    \\
    &= -\bigg( \delta_{\sigma, \uparrow}\delta_{\sigma^\prime,\downarrow} u^>_\uparrow(z;x)u^>_\downarrow(z^\prime;y) - \delta_{\sigma, \downarrow}\delta_{\sigma^\prime, \uparrow} u^>_\downarrow(z^\prime, x) u^>_\uparrow(z;y)  -  \delta_{\sigma, \uparrow} u^>_\uparrow(z;x)a^\ast_\downarrow(u_y)a_\downarrow(u^>_{z^\prime}) \nonumber
    \\ 
    &\quad + \delta_{\sigma, \downarrow} u^>_\downarrow(z^\prime;x) a^\ast_\uparrow(u_y)a_\uparrow(u^>_z) + \delta_{\sigma^\prime, \uparrow}u^>_\uparrow(z;y) a_\downarrow^\ast(u_x)a_\downarrow(u^>_{z^\prime}) - \delta_{\sigma^\prime, \downarrow} u^>_\downarrow(z^\prime;y) a_\uparrow^\ast(u_x)a_\uparrow(u^>_{z})\bigg)a_\uparrow( v_z)a_\downarrow( v_{z^\prime})\times \nonumber
    \\
    &\quad \times a^\ast_\sigma( v_x) a_{\sigma^\prime}(\mathfrak{u}_y^<)- a^\ast_\sigma(u_x) a^\ast_{\sigma^\prime}(u_y)a_\uparrow(u^>_z)a_\downarrow(u^>_{z^\prime})\bigg(\delta_{\sigma,\uparrow} v_\uparrow(x;z) a_\downarrow( v_{z^\prime}) - \delta_{\sigma, \downarrow}v_\downarrow(x;z^\prime)a_\uparrow( v_z)\bigg)a_{\sigma^\prime}(\mathfrak{u}_y^<).\nonumber
  \end{align}
  Note that in the computation above, we used that some term are vanishing due to the fact that the support of $\hat{\mathfrak{u}}^<_\sigma$ and $\hat{u}^>_\sigma$ are disjoint (see \cite[Proposition 7.1]{Gia1} for more details). We now estimate all the error terms. From the first two lines in the right hand side of \eqref{eq: 1 part comm Q3}, we find two different types of error terms.  The first type of error is for instance 
\[
  \mathrm{I}_a = \int\, dxdydz\, V(x-y)u^>_\uparrow(z;x)\langle \xi_\lambda, a^\ast_\downarrow(u_y)b_\downarrow(\varphi_{z})a_\uparrow( v_z)a^\ast_\uparrow( v_x)a_\downarrow(\mathfrak{u}^<_y)\xi_\lambda\rangle,
 \] 
 where we used the definition of the operator $b(\varphi_\cdot)$ in \eqref{eq: def op b g}. 
 Writing $u^>_\uparrow = \delta - \alpha_\uparrow^<$, we write $\mathrm{I}_a = \mathrm{I}_{a;1} + \mathrm{I}_{a;2}$, respectively. Using the first bound in \eqref{eq: bounds av} and Lemma \ref{lem: bound b phi}, we find
\[
  |\mathrm{I}_{a;1}| \leq \int\, dxdy\, |V(x-y)|\|b_\downarrow(\varphi_{x})\|\|a_\uparrow( v_x)\|\|a^\ast_\uparrow( v_x)\|a_\downarrow(u_y)\xi_\lambda\|\|a_\downarrow(\mathfrak{u}^<_y)\xi_\lambda\| \leq  C\|V\|_1\rho^{\frac{4}{3}}\langle \xi_\lambda, \mathcal{N}\xi_\lambda\rangle \leq C\rho^{\frac{4}{3}}\langle \xi_\lambda,\mathcal{N}\xi_\lambda\rangle.
\]
The estimate for the term $\mathrm{I}_{a;2}$ can be done similarly, using also \eqref{eq: bounds alpha< nu}, i.e., 
\[
  |\mathrm{I}_{a;2}| \leq  \int\, dxdydz\, V(x-y)|\alpha^<_\uparrow(z;x)| \|a^\ast_\downarrow(u_y) \xi_\lambda\|\|b_\downarrow(\varphi_{z})a_\uparrow( v_z)a^\ast_\uparrow( v_x)a_\downarrow(\mathfrak{u}^<_y)\xi_\lambda\|\leq C\|V\|_1\|\alpha^<_\uparrow\|_1\rho^{\frac{4}{3}}\langle \xi_\lambda, \mathcal{N}\xi_\lambda\rangle.
\]
 We therefore get that 
\[
  |\mathrm{I}_a| \leq C\rho^{\frac{4}{3}}\langle \xi_\lambda, \mathcal{N}\xi_\lambda\rangle.
\]
The other type of error is coming from the first two terms in the right hand side of \eqref{eq: 1 part comm Q3}. We consider the first one and we rewrite it in normal order: 
  \begin{eqnarray}
    \mathrm{II}
    &=& \int\, dxdydzdz^\prime\, V(x-y)\varphi(z-z^\prime) u^>_\uparrow(z;x)u^>_\downarrow(z^\prime;y)\langle \xi_\lambda, a^\ast_\uparrow( v_x)a_\uparrow( v_z)a_\downarrow( v_{z^\prime})a_\downarrow(\mathfrak{u}^<_y)\xi_\lambda\rangle\nonumber
    \\
    && -\int\, dxdydzdz^\prime\, V(x-y)\varphi(z-z^\prime) u^>_\uparrow(z;x)u^>_\downarrow(z^\prime;y)\omega_\uparrow(x;z) \langle \xi_\lambda, a_\downarrow( v_{z^\prime}) a_\downarrow(\mathfrak{u}^<_y)\xi_\lambda\rangle\equiv \mathrm{II}_{a} + \mathrm{II}_{b}\nonumber.
  \end{eqnarray}
We first notice that  $\mathrm{II}_b$ is vanishing, this is due to the fact that $a_\downarrow( v_{z^\prime})$ and $a_\downarrow(\mathfrak{u}^<_y)$ are supported for momenta in disjoint sets. For more details, we refer to \cite[Eqs. (7.25)--(7.26)]{Gia1}.
  To estimate $\mathrm{II}_{a}$, it is convenient to write each $u^>_\sigma = \delta - \alpha^<_\sigma$.  The estimate of $\mathrm{II}_a$ can be done as the estimate of the term \eqref{eq: term IIa Q2} in Proposition \ref{pro: Q2} using also the first bound in \eqref{eq: bounds av}. We omit the details and we directly write the final bound
\[
    |\mathrm{II}_{a}| \leq C\|V\|_1\rho \langle\xi_\lambda,\mathcal{N}\xi_\lambda\rangle \leq C\rho \langle\xi_\lambda,\mathcal{N}\xi_\lambda\rangle.
\]
The last two terms in the right hand side of \eqref{eq: 1 part comm Q3} can be estimated in the same way. We consider for instance
\begin{eqnarray*}
  \mathrm{III} =  \int\, dxdydzdz^\prime\, V(x-y)\varphi(z-z^\prime)v_\uparrow(z;x)\langle \xi_\lambda, a^\ast_\uparrow(u_x)a^\ast_\downarrow(u_y) a_\downarrow(\mathfrak{u}^<_y) a_\uparrow(u^>_z)a_\downarrow(u^>_{z^\prime}) a_\downarrow( v_{z^\prime}) \xi_\lambda\rangle.
\end{eqnarray*}
It is convenient to rewrite $\mathrm{III}$ as  
\[
  \mathrm{III} = \frac{1}{L^3}\sum_k\hat{v}_\uparrow(k)\left\langle \left(\int dxdy\, e^{ik\cdot x} V(x-y) a^\ast_\uparrow(\mathfrak{u}^<_y)a_\downarrow(u_y)a_\uparrow(u_x)\xi_\lambda\right), \left( \int dx \, e^{ik\cdot z} b_\downarrow(\varphi_z)a_\uparrow(u^>_z)\xi_\lambda\right)\right\rangle.
\]
By Cauchy-Schwarz's inequality and $|\hat{v}_\uparrow(k)| \leq 1$, we then get 
\begin{equation}\label{eq: CS k III Q3}
  |\mathrm{III}| \leq \sqrt{\frac{1}{L^3}\sum_k \left\| \int dxdy\, e^{ik\cdot x} V(x-y) a^\ast_\uparrow(\mathfrak{u}^<_y)a_\downarrow(u_y)a_\uparrow(u_x)\xi_\lambda \right\|^2}\sqrt{\frac{1}{L^3}\sum_k \left\| \int dz \, e^{ik\cdot z} b_\downarrow(\varphi_z)a_\uparrow(u^>_z)\xi_\lambda\right\|^2}.
\end{equation}
By Cauchy-Schwarz again and using also $\|a^\ast_\sigma(\mathfrak{u}^<_\cdot)\|\leq C\rho^{1/2}$, we get
\[
  \frac{1}{L^3}\sum_k \left\| \int dxdy\, e^{ik\cdot x} V(x-y) a^\ast_\uparrow(\mathfrak{u}^<_y)a_\downarrow(u_y)a_\uparrow(u_x)\xi_\lambda \right\|^2\leq C\|V\|_1\rho\langle \xi_\lambda\mathbb{Q}_4 \xi_\lambda\rangle.
\]
Moreover, for the other term in the right hand side of \eqref{eq: CS k III Q3}, we can use Lemma \ref{lem: bound b phi} to write
\[
  \frac{1}{L^3}\sum_k \left\| \int dz \, e^{ik\cdot z} b_\downarrow(\varphi_z)a_\uparrow(u^>_z)\xi_\lambda\right\|^2 = \int dz, \|b_\downarrow(\varphi_z)a_\uparrow(u^>_z)\xi_\lambda\|^2 \leq C\rho^{\frac{2}{3}} \langle \xi_\lambda, \mathcal{N}_> \xi_\lambda\rangle.
\]
Combining the estimates together, we get
\[
  |\mathrm{III}| \leq C\|V\|_1^{\frac{1}{2}}\rho^{\frac{1}{2} + \frac{1}{3}}\|\mathbb{Q}_4^{\frac{1}{2}}\xi_\lambda\|\|\mathcal{N}_>^{\frac{1}{2}}\xi_\lambda\| \leq  C\rho^{\frac{1}{2} + \frac{1}{3}}\|\mathbb{Q}_4^{\frac{1}{2}}\xi_\lambda\|\|\mathcal{N}_>^{\frac{1}{2}}\xi_\lambda\| \leq C\rho^{\frac{1}{2}}\|\mathbb{Q}_4^{\frac{1}{2}}\xi_\lambda\|\|\mathbb{H}_0^{\frac{1}{2}}\xi_\lambda\|,
\]
where in the last estimate we used \eqref{eq: N>}. 
Putting all the estimates together, we conclude the proof.
\end{proof}
\subsection{Conjugation under the transformation $T_1$: conclusions}\label{sec: lw bd T1}
Before analyzing the correlation energy via conjugation with the unitary transformation $T_1$, we collect some useful estimates. In the following, $\psi$ denotes any approximate ground state in the sense of Definition \ref{def: approx gs}.
Using the propagation estimates proved in Proposition \ref{pro: propagation est} and proceeding as in \cite[Proposition 5.9]{FGHP}, we get that  for any $\lambda\in [0,1]$,
\begin{equation}\label{eq: improv N lw bd}
  \langle T^\ast_{1;\lambda}R^\ast \psi, \mathcal{N}T^\ast_{1;\lambda}R^\ast \psi\rangle \leq CL^{\frac{3}{2}}\rho^{\frac{1}{6}}\|\mathbb{H}_0^{\frac{1}{2}}T^\ast_{1}R^\ast \psi\| + CL^3\rho^{\frac{5}{3}}. 
\end{equation}
As a consequence, for any $0<\delta<1$ and for all $\lambda\in [0,1]$, we can write 
\begin{equation}\label{eq: rhoN}
  \rho\langle T^\ast_{1;\lambda}R^\ast \psi, \mathcal{N}T^\ast_{1;\lambda}R^\ast \psi\rangle\leq \delta\langle T^\ast_{1}R^\ast \psi, \mathbb{H}_0 T^\ast_{1}R^\ast \psi\rangle + CL^3\rho^{\frac{7}{3}}.
\end{equation}
We start by using Proposition \ref{pro: fermionic transf} and \eqref{eq: est Q1} to write
\begin{equation}\label{eq: corr energy R conj}
  E_L(N_\uparrow, N_\downarrow)\geq E_{\mathrm{FFG}} + \langle R^\ast\psi, (\mathbb{H}_0 + \mathbb{Q}_4 + \mathbb{Q}_2)R^\ast \psi\rangle + \langle R^\ast\psi, \mathbb{Q}_3 R^\ast \psi\rangle + \mathcal{E}_1(\psi),
\end{equation}
with 
\begin{equation}
  |\mathcal{E}_1(\psi)| \leq C\rho\langle R^\ast \psi, \mathcal{N} R^\ast \psi\rangle \leq CL^3\rho^{\frac{7}{3}} + \delta\langle T^\ast_1 R^\ast\psi, \mathbb{H}_0 T^\ast_1 R^\ast \psi \rangle \qquad \forall \delta \in (0,1),
\end{equation}
where we also used \eqref{eq: rhoN} for $\lambda = 0$.
Moreover, from Proposition \ref{pro: Q3 final} and using \eqref{eq: rhoN} together with the the estimates in Proposition \ref{pro: propagation est}, we get, for any $0<\delta<1$, 
\[
   |\langle R^\ast\psi, \mathbb{Q}_3 R^\ast\psi\rangle |\leq CL^3\rho^{\frac{7}{3}} + \delta \langle T^\ast_1 R^\ast \psi, (\mathbb{H}_ 0 + \mathbb{Q}_4)T^\ast_1 R^\ast \psi\rangle.
\]
We now extract the constant term from $\langle R^\ast\psi,  (\mathbb{H}_0 + \mathbb{Q}_4 + \mathbb{Q}_2)R^\ast\psi\rangle$. By Duhamel's formula, we can write 
\begin{eqnarray*}
\langle R^\ast \psi, (\mathbb{H}_0 + \mathbb{Q}_4) R^\ast\psi \rangle &=& \langle T^\ast_1 R^\ast\psi, (\mathbb{H}_0 + \mathbb{Q}_4)T^\ast_1 R^\ast \psi\rangle - \int_0^1 d\lambda\, \partial_\lambda\langle T^{\ast}_{1;\lambda}R^\ast \psi, (\mathbb{H}_0 + \mathbb{Q}_4)T^{\ast}_{1;\lambda}R^\ast \psi\rangle 
\\
&=& \langle T^\ast_1R^\ast \psi, (\mathbb{H}_0 + \mathbb{Q}_4)T^\ast_1 R^\ast \psi\rangle - \int_0^1 d\lambda\,\langle T^{\ast}_{1;\lambda}R^\ast \psi, (\mathbb{T}_1 + \mathbb{T}_2)T^\ast_{1;\lambda}R^\ast\psi\rangle + \mathcal{E}_2, \nonumber
\end{eqnarray*}
where $\mathcal{E}_2$ can be bounded using Proposition \ref{pro: H0} and Proposition \ref{pro: Q4} together with the propagation bounds proved in Proposition \ref{pro: propagation est} and the bound in Proposition \eqref{pro: N}:
\begin{equation}
  |\mathcal{E}_2| \leq CL^3\rho^{\frac{7}{3}} + \delta\langle T^\ast_1 R^\ast \psi, \mathbb{H}_0 T^\ast_1 R^\ast \psi\rangle,
\end{equation}
for any $0<\delta<1$. Moreover, using Proposition \ref{pro: reg} together with Proposition \ref{pro: scatt canc}, we can write 
\[
  \langle T^\ast_{1;\lambda}R^\ast \psi, (-\mathbb{T}_1 - \mathbb{T}_2)T^\ast_{1;\lambda}R^\ast \psi\rangle =  -\langle T^\ast_{1;\lambda}R^\ast \psi, \mathbb{Q}_{2;>} T^\ast_{1;\lambda}R^\ast \psi \rangle - \langle T^\ast_{1;\lambda}R^\ast \psi, \mathbb{T}_2^< T^\ast_{1;\lambda}R^\ast \psi\rangle + \mathcal{E}_3,
\]
where $\mathcal{E}_3$ is estimated as in \eqref{eq: error T2} and \eqref{eq: est err scatt eq canc}, i.e., 
\[
  |\mathcal{E}_3| = |\langle T^\ast_{1;\lambda} R^\ast \psi, \mathfrak{E}_{\mathbb{T}_2}T^\ast_{1;\lambda}R^\ast \psi\rangle + \langle T^\ast_{1;\lambda} R^\ast \psi, (-\mathbb{T}_1 - \mathbb{T}_2^> + \mathbb{Q}_{2;>})T^\ast_{1;\lambda}R^\ast \psi\rangle| \leq CL^3\rho^{\frac{7}{3}} + \delta\langle T^\ast_{1}R^\ast \psi, \mathbb{H}_0 T^\ast_{1}R^\ast \psi\rangle,
\]
for any $0<\delta<1$. All together we find that 
\begin{eqnarray*}
  \langle R^\ast \psi, (\mathbb{H}_0 + \mathbb{Q}_4 + \mathbb{Q}_2) R^\ast\psi \rangle\hspace{-0.2cm} &=& \hspace{-0.2cm}\langle T^\ast_1 R^\ast \psi, (\mathbb{H}_0 + \mathbb{Q}_4) T^\ast_1 R^\ast \psi\rangle + \langle R^\ast\psi, \mathbb{Q}_{2} R^\ast\psi \rangle
  \\
  && -  \int_0^1 d\lambda \langle T^{\ast}_{1;\lambda}R^\ast \psi, \mathbb{Q}_{2;>} T^\ast_{1;\lambda}R^\ast \psi \rangle - \int_0^{1}d\lambda\, \langle T^{\ast}_{1;\lambda}R^\ast\psi, \mathbb{T}_2^< T^\ast_{1;\lambda}R^\ast \psi\rangle + \mathcal{E},
  \end{eqnarray*}
  with 
  \[
  |\mathcal{E}| \leq CL^3\rho^{\frac{7}{3}} + \delta\langle T^\ast_1 R^\ast \psi, (\mathbb{H}_0 + \mathbb{Q}_4) T^\ast_1 R^\ast \psi\rangle, 
  \]
for any $0<\delta<1$. We now use Duhamel's formula again to write 
  \begin{multline*}
\langle R^\ast\psi, \mathbb{Q}_{2} R^\ast\psi \rangle -  \int_0^1 d\lambda \langle T^{\ast}_{1;\lambda}R^\ast \psi, \mathbb{Q}_{2;>} T^\ast_{1;\lambda}R^\ast \psi \rangle = \langle T^\ast_1 R^\ast \psi , (\mathbb{Q}_2 - \mathbb{Q}_{2;>}) T^\ast_1 R^\ast \psi\rangle 
\\
- \int_0^1 d\lambda\, \partial_\lambda\langle T^\ast_{1;\lambda}R^\ast \psi, \mathbb{Q}_2 T^\ast_{1;\lambda}R^\ast \psi\rangle - \int_0^1\int_1^\lambda\, d\lambda^\prime \partial_{\lambda^\prime}\langle T^\ast_{1;\lambda^\prime}R^\ast \psi, \mathbb{Q}_{2;>} T^\ast_{1;\lambda^\prime}R^\ast \psi\rangle.
  \end{multline*}
  Therefore, from Proposition \ref{pro: scatt canc}, see \eqref{eq: error Q2}, we get 
  \[
    \langle T^\ast_1 R^\ast \psi , (\mathbb{Q}_2 - \mathbb{Q}_{2;>}) T^\ast_1 R^\ast \psi\rangle = \langle T^\ast_1 R^\ast \psi , \mathbb{Q}_{2;<} T^\ast_1 R^\ast \psi\rangle + \langle T^\ast_1 R^\ast \psi , \mathfrak{E}_{\mathbb{Q}_2} T^\ast_1 R^\ast \psi\rangle,
  \]
  with 
  \[
    |\langle T^\ast_1 R^\ast \psi , \mathfrak{E}_{\mathbb{Q}_2} T^\ast_1 R^\ast \psi\rangle| \leq CL^3 \rho^{\frac{7}{3}} + \delta \langle T^\ast_1 R^\ast \psi ,\mathbb{H}_0 T^\ast_1 R^\ast \psi \rangle,
  \]
  for any $0<\delta<1$. And, using Proposition \ref{pro: Q2} we can write 
  \begin{multline*}
    - \int_0^1 d\lambda\, \partial_\lambda\langle T^\ast_{1;\lambda}R^\ast \psi, \mathbb{Q}_2 T^\ast_{1;\lambda}R^\ast \psi\rangle - \int_0^1 d\lambda\int_1^\lambda\, d\lambda^\prime \partial_{\lambda^\prime}\langle T^\ast_{1;\lambda^\prime}R^\ast \psi, \mathbb{Q}_{2;>} T^\ast_{1;\lambda^\prime}R^\ast \psi\rangle
    \\
     = -\rho_\uparrow\rho_\downarrow\int dxdy\, V(x-y)\varphi(x-y) + \mathcal{E}_4,
  \end{multline*}
  with 
  \[
    |\mathcal{E}_4| \leq CL^3\rho^{\frac{7}{3}} + \delta\langle T^\ast_1 R^\ast\psi,\mathbb{H}_0 T^\ast_1 R^\ast\psi\rangle,
  \]
  for any $0<\delta<1$, using also the propagation of the estimates proved in Proposition \ref{pro: propagation est} and \eqref{eq: rhoN}. 
To conclude, we use Duhamel's formula and Proposition \ref{pro: T2<} to write 
\begin{eqnarray*}
 - \int_0^{1}d\lambda\, \langle T^{\ast}_{1;\lambda}R^\ast\psi, \mathbb{T}_2^< T^\ast_{1;\lambda}R^\ast \psi\rangle &=& -\langle T^{\ast}_{1}R^\ast\psi, \mathbb{T}_2^< T^\ast_{1}R^\ast \psi\rangle + \int_0^{1}d\lambda\int_\lambda^1 d\lambda^\prime\, \partial_{\lambda^\prime}\langle T^{\ast}_{1;\lambda^\prime}R^\ast\psi, \mathbb{T}_2^< T^\ast_{1;\lambda^\prime}R^\ast \psi\rangle
  \\
  &=& -\langle T^{\ast}_{1}R^\ast\psi, \mathbb{T}_2^< T^\ast_{1}R^\ast \psi\rangle + \mathcal{E}_5,
\end{eqnarray*}
with $|\mathcal{E}_5| \leq CL^3\rho^{\frac{7}{3}}$,
where we used the propagation estimates proved in Proposition \ref{pro: propagation est} and the bound in Proposition \ref{pro: N}.
Combining all the estimates together, we get that there exists $0<\widetilde{\delta}<1$ such that 
\begin{eqnarray}\label{eq: final lw bt T1}
  E_{L}(N_\uparrow, N_\downarrow) &\geq& E_{\mathrm{FFG}} + (1-\widetilde{\delta})\langle T^\ast_1 R^\ast \psi, (\mathbb{H}_0 + \mathbb{Q}_4) T^\ast_1 R^\ast \psi \rangle - \rho_\uparrow\rho_\downarrow \int dxdy V(x-y)\varphi(x-y) \nonumber
  \\
  && + \langle T^\ast_1 R^\ast \psi, (\mathbb{Q}_{2;<} - \mathbb{T}_2^<) T^\ast_1 R^\ast \psi\rangle + {\mathcal{E}}_{\mathrm{T_1}},
\end{eqnarray}
with
\begin{equation}
  |{\mathcal{E}}_{\mathrm{T_1}}| \leq CL^3 \rho^{\frac{7}{3}}.
\end{equation}
Using the assumptions on the interaction potential, together with the expression of the free Fermi gas energy in \eqref{eq: FFG energy} and taking $L$ large enough, similarly as in \cite[Section 6]{Gia1} we obtain that
\begin{eqnarray}\label{eq: final lw bd T1 unit volume}
  E_{L}(N_\uparrow, N_\downarrow) &\geq& (1-\widetilde{\delta})\langle T^\ast_1 R^\ast \psi, (\mathbb{H}_0 + \mathbb{Q}_4) T^\ast_1 R^\ast \psi \rangle + \langle T^\ast_1 R^\ast \psi, (\mathbb{Q}_{2;<} - \mathbb{T}_2^<) T^\ast_1 R^\ast \psi\rangle\nonumber
  \\
  && + \frac{3}{5}(6\pi^2)^{\frac{2}{3}} (\rho_\uparrow^{\frac{5}{3}} + \rho_\downarrow^{\frac{5}{3}})L^3 + 8\pi a\rho_\uparrow\rho_\downarrow L^3  + \mathcal{E}_{\mathrm{T}_1},\qquad |\mathcal{E}_{\mathrm{T}_1}|\leq CL^3\rho^{\frac{7}{3}}.
\end{eqnarray}
\subsection{Optimal upper bound}\label{sec: up bd}
Proceeding as in Section \ref{sec: lw bd T1} and taking as trial state $\psi_{\mathrm{trial}} = RT_1\Omega$, one can prove the following optimal upper bound:
\begin{equation}\label{eq: optimal upper bound sect 4}
  e(\rho_\uparrow, \rho_\downarrow) \leq \frac{3}{5}(6\pi^2)^{\frac{2}{3}} (\rho_\uparrow^{\frac{5}{3}} + \rho_\downarrow^{\frac{5}{3}}) + 8\pi a\rho_\uparrow\rho_\downarrow  + C\rho^{\frac{7}{3}}.
\end{equation}
The trial state $\psi_{\mathrm{trial}}$ has exactly $N_\uparrow$ particles with spin $\uparrow$ and $N_\downarrow$ particles with spin $\downarrow$ (see \cite[Section 7]{FGHP},\cite[Section 6]{Gia1}).

The proof follows the same strategy in \cite{Gia1}, see in particular \cite[Section 6]{Gia1}. The difference with respect to the proof of \eqref{eq: optimal upper bound sect 4}  in \cite{Gia1} concerns the definition of the momenta in $T_1$: more precisely, one should compare $\hat{v}^r_\sigma$ and $\hat{u}^r_\sigma$ in \cite[Eq. (4.1)]{Gia1} with $\hat{v}_\sigma$ and $\hat{u}^>_\sigma$ in \eqref{eq: def u,v} and \eqref{eq: def u>}, respectively. Nevertheless, the resulting error terms can be controlled exactly as in the previous propositions. For completeness, we briefly outline the main steps of the argument.

Since we are dealing with an upper bound, we cannot neglect the interaction between particles with equal spin. For this, we refer to \cite[Proposition 3.1]{FGHP}, which provides the adaptation of the conjugation under the particle-hole transformation introduced in Section \ref{sec: particle-hole}. Proceeding as for the lower bound (see \eqref{eq: corr energy R conj}) and using \cite[Proposition 3.1]{FGHP}, see also \cite[Section 6]{Gia1}, we obtain
\begin{equation}\label{eq: starting up bd}
  E_L(N_\uparrow, N_\downarrow) \leq E_{\mathrm{FFG}} + \langle R^\ast\psi_{\mathrm{trial}}, (\mathbb{H}_0 + \mathbb{Q}_2 + \mathbb{Q}_4^{\sigma,\sigma^\prime}) R^\ast \psi_{\mathrm{trial}}\rangle + \mathcal{E}\leq  E_{\mathrm{FFG}} + \langle T_1\Omega, (\mathbb{H}_0 + \mathbb{Q}_2 + \mathbb{Q}_4^{\sigma,\sigma^\prime})T_1\Omega\rangle + \mathcal{E},
\end{equation}
where $|\mathcal{E}| \leq C\rho \langle T_1\Omega \psi, \mathcal{N}T_1\Omega\rangle$, and $\mathbb{Q}_{4}^{\sigma,\sigma^\prime}$ as in Proposition \ref{pro: fermionic transf}.
Although $\mathbb{Q}_4^{\sigma,\sigma^\prime}$ includes the interaction between particles with equal spin, such contribution can be neglected in $\mathbb{Q}_2$, owing to the particular form of the trial state (see \cite[Section 7]{FGHP} and \cite[Section 6]{Gia1}). Moreover, from Proposition \ref{pro: N}, and the definition of the trial state, we directly get that for any $\lambda\in [0,1]$,
\begin{equation}\label{eq: numb op up bd}
  \langle T_{1;\lambda} \Omega, \mathcal{N} T_{1;\lambda}\Omega \rangle \leq CL^3\rho^{\frac{5}{3}}.
\end{equation}
Hence the error term in \eqref{eq: starting up bd} satisfies $|\mathcal{E}|\leq CL^3\rho^{\frac{8}{3}}$. Proceeding as in the previous section (see also \cite[Section 6]{Gia1}) and using \eqref{eq: numb op up bd} together with the propagation estimates from Proposition \ref{pro: propagation est}  and Remark \ref{rem: prop Q4sigmasigma'}, we finally obtain
\begin{equation}
  e(\rho_\uparrow, \rho_\downarrow) \leq \frac{3}{5}(6\pi^2)^{\frac{2}{3}} (\rho_\uparrow^{\frac{5}{3}} + \rho_\downarrow^{\frac{5}{3}}) + 8\pi a\rho_\uparrow\rho_\downarrow  + C\rho^{\frac{7}{3}}.
\end{equation}
This extends the optimal upper bound to any interaction potential $V$ satisfying Assumption \ref{asu: potential V}.

\subsection{Reduction to the new effective correlation energy}\label{sec: reduction new eff corr en}
In this section we isolate the new effective correlation energy that we want to conjugate under the unitary transformation $T_2$ defined in \eqref{eq: def T2}. To start, we define
\[
  \widetilde{\mathbb{Q}} := \mathbb{Q}_{2;<} - \mathbb{T}_2^<,
\]
with $\mathbb{Q}_{2;<}$ and  $\mathbb{T}_2^<$ as in \eqref{eq: Q2<} and \eqref{eq: def T2<}, respectively. We define 
\[
  \hat{u}^{\ll}_{\epsilon,\sigma}(k) := \hat{u}^<_\sigma(k) \chi_{\epsilon,\sigma}^\ll(k), \qquad \hat{u}^{\gg}_{\epsilon,\sigma}(k) := \hat{u}^<_\sigma(k) \chi_{\epsilon,\sigma}^\gg(k),
\]
with $\hat{u}^<_\sigma$ as in \eqref{eq: def u<} and
\[
 \hat{\chi}^{\ll}_{\epsilon,\sigma}(k) := \begin{cases} 1 &\mbox{if}\,\,\, k_F^\sigma \leq |k| \leq k_F^\sigma + (k_F^\sigma)^{1+\epsilon}, 
 \\
 0 &\mbox{if}\,\,\, |k| > k_F^\sigma + (k_F^\sigma)^{1+\epsilon},\end{cases} 
 \qquad 
 \hat{\chi }^{\gg}_{\epsilon,\sigma}(k) := \begin{cases} 0 &\mbox{if}\,\,\, k_F^\sigma \leq |k| \leq k_F^\sigma + (k_F^\sigma)^{1+\epsilon}, 
 \\
 1 &\mbox{if}\,\,\,  k_F^\sigma + (k_F^\sigma)^{1+\epsilon} < |k| \leq 3k_F^\sigma
 \end{cases}
\]
Moreover, we set
\begin{equation}\label{eq: Q1/2>}
  \widetilde{\mathbb{Q}}_{\epsilon,>} := \int dxdy\, V (1- \varphi)(x-y) a_\uparrow(u_{\epsilon,x}^\gg)a_\uparrow( v_x)a_\downarrow(u_{\epsilon,y}^\gg) a_\downarrow( v_y)+ \mathrm{h.c.}, 
\end{equation}
and
\[
  \widetilde{\mathbb{Q}}_< := \widetilde{\mathbb{Q}} - \widetilde{\mathbb{Q}}_{\epsilon,>},
\]
with $u^{\gg}_{\epsilon,x}(\cdot):= u^\gg_{\epsilon}(\cdot - x)$ and $u^{\ll}_{\epsilon,x}(\cdot):= u^\ll_{\epsilon}(\cdot - x)$.
We also write $V_\varphi := V(1-\varphi)$. Note that $\|V_\varphi\|_1 \leq C\|V\|_1$, using that $\|\varphi\|_\infty \leq C$ (see Lemma \ref{lem: bound phi}). Next, we want to show that $\langle T^\ast_1 R^\ast \psi, \widetilde{\mathbb{Q}}_< T^\ast_1 R^\ast \psi\rangle$ is $\mathcal{O}(L^3\rho^{7/3})$. We start by noticing that in each operator in $\widetilde{\mathbb{Q}}_<$, there is at least one $a_\sigma(u^\ll_\cdot)$. Since $\|a_\sigma( v_\cdot)\|\leq C\rho^{1/2}$, $\|a_\sigma(u^\ll_\cdot)\|\leq C\rho^{1/2 + \epsilon/6}$ and 
\begin{equation}\label{eq: N>> H0}
  \int dx\, \|a_\sigma(u^{\gg}_{\epsilon,x})T^\ast_1 R^\ast \psi\|^2 \leq C\rho^{-\frac{2}{3} - \frac{\epsilon}{3}}\langle T^\ast_1 R^\ast \psi, \mathbb{H}_0 T^\ast_1 R^\ast \psi\rangle,
\end{equation}
by Cauchy-Schwarz inequality, we find that for any $0<\delta<1$,
\begin{align*}
  &|\langle T_1^\ast R^\ast \psi,  \widetilde{\mathbb{Q}}_< T_1^\ast R^\ast \psi\rangle| \leq C\rho^{1 + \frac{\epsilon}{6}}\int dxdy \, V_\varphi(x-y)\left(\|a_\sigma( v_x) a_\sigma(u^\gg_{\epsilon,x})T^\ast_1 R^\ast \psi  \| + \|a_\sigma( v_x)a_\sigma(u^\ll_{\epsilon,x})T^\ast_1 R^\ast \psi  \|\right)
  \\
  &\leq C\rho^{\frac{3}{2} + \frac{\epsilon}{6}}\int dxdy \, V_\varphi(x-y)\|a_\sigma(u^\gg_{\epsilon,x})T^\ast_1 R^\ast \psi \rangle \| + C\rho^{\frac{3}{2} + \frac{\epsilon}{3}}\int dxdy\, V_\varphi(x-y)\|a_\sigma( v_x)T^\ast_1 R^\ast \psi  \|
  \\
  &\leq CL^{\frac{3}{2}}\|V_\varphi\|_1\rho^{\frac{3}{2} + \frac{\epsilon}{6} -\frac{1}{3} -\frac{\epsilon}{6}}\|\mathbb{H}_0^{\frac{1}{2}}T^\ast_1 R^\ast \psi\| + CL^{\frac{3}{2}}\rho^{\frac{3}{2} + \frac{\epsilon}{3}}\|V_\varphi\|_1\|\mathcal{N}^{\frac{1}{2}}T^\ast_1 R^\ast \psi\|
  \\
  &\leq CL^3 \rho^{\frac{7}{3}} + CL^3\rho^{2 + \frac{1}{9} + \frac{4\epsilon}{9}} + \delta \langle T^\ast_1 R^\ast \psi, \mathbb{H}_0 T^\ast_1 R^\ast\psi\rangle,
\end{align*}
where in the last inequality, we used Cauchy-Schwarz together with the estimate \eqref{eq: improv N lw bd} and Young's inequality with exponents $p=3/4$, $q=1/4$. Therefore, taking $\epsilon = 1/2$, we can conclude that for any $0<\delta<1$, 
\[
   |\langle T_1^\ast R^\ast \psi,  \widetilde{\mathbb{Q}}_< T_1^\ast R^\ast \psi\rangle|\leq CL^3\rho^{\frac{7}{3}} + \delta \langle T^\ast_1 R^\ast \psi, \mathbb{H}_0 T^\ast_1 R^\ast\psi\rangle.
\]
We now proceed similarly for the operators $a_\sigma( v_\cdot)$. More precisely, we define 
\[
  \hat{v}^\ll_\sigma(k) := \hat{v}_\sigma(k)\eta_\sigma^\ll(k), \qquad \hat{v}^\gg_\sigma(k) := \hat{v}_\sigma(k)\eta_\sigma^\gg(k),
\]
with $\hat{v}_\sigma$ defined as in \eqref{eq: def u,v} and 
\[
  \hat{\eta}^{\ll}_{\sigma}(k) := \begin{cases} 1 &\mbox{if}\,\,\, k_F^\sigma - (k_F^\sigma)^{1+\frac{1}{2}}  \leq |k| \leq k_F^\sigma , 
 \\
 0 &\mbox{if}\,\,\,  |k| \leq k_F^\sigma - (k_F^\sigma)^{1+\frac{1}{2}},
 \end{cases} \qquad \hat{\eta}^{\gg}_{\sigma}(k) := \begin{cases} 0 &\mbox{if}\,\,\, k_F^\sigma - (k_F^\sigma)^{1+\frac{1}{2}}  \leq |k| \leq k_F^\sigma,
 \\
 1 &\mbox{if}\,\,\,  |k| \leq k_F^\sigma - (k_F^\sigma)^{1+\frac{1}{2}}.
 \end{cases}
\]
We first observe that for the terms in $\widetilde{\mathbb{Q}}_{1/2, >}$ containing one operator $a_\sigma( v^\ll_\cdot)$, the expectation value is already of the order $\mathcal{O}(L^3\rho^{7/3})$. More precisely, proceeding as above and using the bounds 
\[
\|a_\sigma( v^\ll_\cdot)\|\leq C\rho^{1/2 + \epsilon/6}, \qquad \|a_\sigma( v^\gg_\cdot)\|\leq C\rho^{1/2}, \qquad \|a_\sigma( u^\gg_{1/2,\cdot})\|\leq C\rho^{1/2},
\] 
we obtain 
\begin{align*}
&\left|\int dxdy\, V (1- \varphi)(x-y) \langle T^\ast_1 R^\ast \psi, a_\uparrow(u_{1/2,x}^\gg)a_\uparrow( v_x^\ll)a_\downarrow(u_{1/2,y}^\gg) a_\downarrow( v_y^\gg)T^\ast_1 R^\ast \psi\rangle\right| 
\\
&\leq C\rho^{\frac{3}{2} +\frac{\epsilon}{6}}\int dxdy \, |V_\varphi(x-y)|\|a_\uparrow(u^\gg_x)T^\ast_1 R^\ast \psi\|\leq C\rho^{\frac{3}{2} +\frac{\epsilon}{6} -\frac{1}{3} -\frac{\epsilon}{6}}\|\mathbb{H}_0^{\frac{1}{2}}T^\ast_1 R^\ast \psi\|\leq CL^3\rho^{\frac{7}{3}} + \delta\langle T^\ast_1 R^\ast \psi,\mathbb{H}_0 T^\ast_1 R^\ast \psi\rangle.
\end{align*}
The term in $\widetilde{\mathbb{Q}}_{1/2, >}$ containing $a_\uparrow( v^\gg_x)a_\downarrow( v^\ll_y)$ can be estimated in the same way, while the contribution with $a_\uparrow( v^\ll_x)a_\uparrow( v^\ll_x)$ yields an even smaller error.
Next, we define $\hat{u}_\sigma^\gg := \hat{u}^<_\sigma(k)\chi_\sigma^{\gg}(k)$ with
\[
  \hat{\chi }^{\gg}_{\sigma}(k) := \begin{cases} 0 &\mbox{if}\,\,\, k_F^\sigma \leq |k| \leq k_F^\sigma + (k_F^\sigma)^{1+\frac{1}{2}}, 
 \\
 1 &\mbox{if}\,\,\,  k_F^\sigma + (k_F^\sigma)^{1+\frac{1}{2}} < |k| \leq 3k_F^\sigma
 \end{cases}
\]
and set
\begin{equation}\label{eq: def tildeQ>}
  \widetilde{\mathbb{Q}}_{>} = \int dxdy\, V (1- \varphi)(x-y) a_\uparrow(u_{x}^\gg)a_\uparrow( v_x^\gg)a_\downarrow(u_{y}^\gg) a_\downarrow( v_y^\gg)+ \mathrm{h.c.}
\end{equation}
From \eqref{eq: final lw bd T1 unit volume}, for any $0<\delta<1$, we then obtain for $L$ large enough
\begin{equation}\label{eq: final T1}
E_{L}(N_\uparrow, N_\downarrow) \geq \frac{1}{2}\langle T^\ast_1 R^\ast \psi, (\mathbb{H}_0  + 2\widetilde{\mathbb{Q}}_>) T^\ast_1 R^\ast \psi \rangle + \frac{3}{5}(6\pi^2)^{\frac{2}{3}} (\rho_\uparrow^{\frac{5}{3}} + \rho_\downarrow^{\frac{5}{3}})L^3 + 8\pi a\rho_\uparrow\rho_\downarrow L^3  + \mathcal{E}_{\mathrm{T}_1},
\end{equation}
with $|\mathcal{E}_{\mathrm{T}_1}| \leq CL^3\rho^{7/3}$. Here we have fixed the prefactor of $\mathbb{H}_0$ to be $1/2$, and used the positivity of  $\mathbb{Q}_4$ in \eqref{eq: final lw bd T1 unit volume}.
Finally, we define the new effective correlation Hamiltonian as 
\begin{equation}\label{eq: new eff corr en}
  \widetilde{\mathcal{H}{}}_{\mathrm{corr}}^{\mathrm{eff}}:= (1/2)\mathbb{H}_0 + \widetilde{\mathbb{Q}}_>,
\end{equation}
where $\mathbb{H}_0$ and $\widetilde{\mathbb{Q}}_>$ are given by \eqref{eq: def H0} and \eqref{eq: def tildeQ>}, respectively. 
\section{Conjugation under the transformation $T_2$}\label{sec: T2 conjugation}
To establish the optimal lower bound, we conjugate the effective correlation energy $\widetilde{\mathcal{H}}_{\mathrm{corr}}^{\mathrm{eff}}$, defined in \eqref{eq: new eff corr en}, under the unitary transformation introduced in Section \ref{sec: T2 unitary} (see also \eqref{eq: def T2 unitary}). Before turning to the rigorous analysis, we first outline the main ideas.\\

\noindent\textbf{Ideas on the conjugation under $T_2$}\label{sec: strategy T2}. Starting from \eqref{eq: final T1}, i.e., 
\begin{equation}\label{eq: starting T2 strategy}
E_{L}(N_\uparrow, N_\downarrow) \geq \frac{1}{2}\langle T^\ast_1 R^\ast \psi, (\mathbb{H}_0  + 2\widetilde{\mathbb{Q}}_>) T^\ast_1 R^\ast \psi \rangle + \frac{3}{5}(6\pi^2)^{\frac{2}{3}} (\rho_\uparrow^{\frac{5}{3}} + \rho_\downarrow^{\frac{5}{3}})L^3 + 8\pi a\rho_\uparrow\rho_\downarrow L^3  + \mathcal{E}_{\mathrm{T}_1},\quad |\mathcal{E}_{T_1} |\leq CL^3\rho^{\frac{7}{3}},
\end{equation}
and conjugating $(1/2)\mathbb{H}_0 + \widetilde{\mathbb{Q}}_>$ under $T_2$, we find that 
\begin{multline*}
   \frac{1}{2}\langle T^\ast_{1}R^\ast \psi, (\mathbb{H}_0   + 2\widetilde{\mathbb{Q}}_>)T^\ast_{1}R^\ast \psi\rangle =  \frac{1}{2}\langle T^\ast_2 T^\ast_{1}R^\ast \psi, (\mathbb{H}_0   + 2\widetilde{\mathbb{Q}}_>)T^\ast_2 T^\ast_{1}R^\ast \psi\rangle
   \\
   -  \frac{1}{2}\int_0^1 d\lambda\, \partial_\lambda\, \langle T^\ast_{2;\lambda} T^\ast_{1}R^\ast \psi, (\mathbb{H}_0   + 2\widetilde{\mathbb{Q}}_>) T^\ast_{2;\lambda}T^\ast_{1}R^\ast \psi\rangle.
\end{multline*}
Using now that $\partial_\lambda T_{2;\lambda}\mathbb{H}_0 T^\ast_{2;\lambda}  = - T_{2;\lambda}[\mathbb{H}_0, B_2 -B_2^\ast] T^\ast_{2;\lambda}$ and that $(1/2)[\mathbb{H}_0, B_2 - B_2^\ast] = - \widetilde{\mathbb{Q}}_>$, we get 
\begin{eqnarray}\label{eq: scatt eq canc T2}
   \frac{1}{2}\langle T^\ast_{1}R^\ast \psi, (\mathbb{H}_0   + 2\widetilde{\mathbb{Q}}_>)T^\ast_{1}R^\ast \psi\rangle &=&  \frac{1}{2}\langle T^\ast_2 T^\ast_{1}R^\ast \psi, (\mathbb{H}_0   + 2\widetilde{\mathbb{Q}}_>)T^\ast_2 T^\ast_{1}R^\ast \psi\rangle\nonumber
   \\
   && +  \int_0^1 d\lambda\,  \langle T^\ast_{2;\lambda} T^\ast_{1}R^\ast \psi, \big(- \widetilde{\mathbb{Q}}_> + [\widetilde{\mathbb{Q}}_>, B_2 - B_2^\ast ]\big) T^\ast_{2;\lambda}T^\ast_{1}R^\ast \psi\rangle.
\end{eqnarray}
Applying Duhamel's formula again (see Section \ref{sec: lw bd T2}), we then get that 
\textcolor{black}{\begin{eqnarray*}
   \frac{1}{2}\langle T^\ast_{1}R^\ast \psi, (\mathbb{H}_0   + 2\widetilde{\mathbb{Q}}_>)T^\ast_{1}R^\ast \psi\rangle &=&  \frac{1}{2}\langle T^\ast_2 T^\ast_{1}R^\ast \psi, \mathbb{H}_0 T^\ast_2 T^\ast_{1}R^\ast \psi\rangle
   +  \int_0^1 d\lambda\,  \langle T^\ast_{2;\lambda} T^\ast_{1}R^\ast \psi, [\widetilde{\mathbb{Q}}_>, B_2 - B_2^\ast ]T^\ast_{2;\lambda} T^\ast_{1}R^\ast \psi\rangle
   \\
   && + \int_0^1 d\lambda\int_1^\lambda d\lambda^{\prime}\,  \langle T^\ast_{2;\lambda^\prime} T^\ast_{1}R^\ast \psi, [\widetilde{\mathbb{Q}}_>, B_2 - B_2^\ast ] T^\ast_{2;\lambda^\prime} T^\ast_{1}R^\ast \psi\rangle,
\end{eqnarray*}}
which implies (see Proposition \ref{pro: tildeQ}) that 
\begin{equation}\label{eq: fin conj T2 strategy}
   \frac{1}{2}\langle T^\ast_{1}R^\ast \psi, (\mathbb{H}_0   + 2\widetilde{\mathbb{Q}}_>)T^\ast_{1}R^\ast \psi\rangle \geq  \frac{1}{2}\langle T^\ast_2 T^\ast_{1}R^\ast \psi, \mathbb{H}_0  T^\ast_2 T^\ast_{1}R^\ast \psi\rangle - CL^3\rho^{\frac{7}{3}}.
\end{equation}
Inserting then \eqref{eq: fin conj T2 strategy} in \eqref{eq: starting T2 strategy}, and taking the thermodynamic limit, we get
\[
  e(\rho_\uparrow, \rho_\downarrow) \geq  \frac{3}{5}(6\pi^2)^{\frac{2}{3}} \left(\rho_\uparrow^{\frac{5}{3}} + \rho_\downarrow^{\frac{5}{3}}\right) + 8\pi a \rho_\uparrow \rho_\downarrow\ - C\rho^{\frac{7}{3}}.
\]
\subsection{Propagation estimates}
In this section we prove that some propagation estimates for $\mathbb{H}_0$.

\begin{proposition}[Propagation estimates for $\mathbb{H}_0$ -- Part II]\label{pro: prop est H0 T2} Let $\lambda\in [0,1]$ and let $\psi$ be an approximate ground state in the sense of Definition \ref{def: approx gs}. Under the same assumptions as in Theorem \ref{thm: optimal lw bd}, it holds that 
\begin{equation}\label{eq: prop H0 T2}
  \langle T^\ast_{2;\lambda} T^\ast_1 R^\ast \psi, \mathbb{H}_0  T^\ast_{2;\lambda} T^\ast_1 R^\ast\rangle \leq CL^3\rho^2.
\end{equation}
\end{proposition}
\begin{remark}
Note that the bound in \eqref{eq: prop H0 T2} is not optimal, but it is enough for our purposes.
\end{remark}
\begin{proof}
To prove \eqref{eq: prop H0 T2}, we use Gr\"onwall's Lemma. We compute 
\begin{equation}\label{eq: der H0 T2}
  \partial_\lambda \langle T^\ast_{2;\lambda} T^\ast_1 R^\ast \psi, \mathbb{H}_0 T^\ast_{2;\lambda} T^\ast_1 R^\ast \psi\rangle = -\langle T^\ast_{2;\lambda} T^\ast_1 R^\ast \psi, [\mathbb{H}_0, B_2 - B_2^\ast]T^\ast_{2;\lambda} T^\ast_1 R^\ast \psi\rangle.
\end{equation}
Computing the commutator above (the calculation is the same as the one in \cite[Proposition 5.1]{Gia1}, we omit the details), we find that 
\begin{multline}\label{eq: com H0 T2}
  [\mathbb{H}_0, B_2 - B_2^\ast] = -\frac{1}{L^3}\sum_{p,r,r^\prime} (|p+r|^2 - |r|^2 + |r^\prime - p|^2 - |r^\prime|^2)\hat{f}_{r,r^\prime}(p) \hat{u}^\gg_\uparrow(r+p)\hat{u}^\gg_\downarrow(r^\prime - p)\hat{v}^\gg_\uparrow(r)\hat{v}^\gg_\downarrow(r^\prime)\times 
\\
\times \hat{a}_{r+p,\uparrow}\hat{a}_{-r,\uparrow}\hat{a}_{r^\prime - p,\downarrow}\hat{a}_{-r^\prime,\downarrow} + \mathrm{h.c.}
\end{multline}
Using now the definition of $\hat{f}_{r,r^\prime}$ (see \eqref{eq: def frr'}), we get that 
\begin{equation}\label{eq: comm H0B2 f}
  [\mathbb{H}_0, B_2 - B_2^\ast] = -\frac{2}{L^3}\sum_{p,r,r^\prime} \widehat{V_\varphi}(p) \hat{u}^\gg_\uparrow(r+p)\hat{u}^\gg_\downarrow(r^\prime - p)\hat{v}^\gg_\uparrow(r)\hat{v}^\gg_\downarrow(r^\prime)\hat{a}_{r+p,\uparrow}\hat{a}_{-r,\uparrow}\hat{a}_{r^\prime - p,\downarrow}\hat{a}_{-r^\prime,\downarrow} + \mathrm{h.c.}
\end{equation}
Rewriting the term we want to estimate in configuration space, and using the notation $\xi_\lambda:= T^\ast_{2;\lambda}T^\ast_1 R^\ast \psi$, by Cauchy-Schwarz, we find that 
\begin{align}\label{eq: comm H0 B2}
  &|\langle \xi_\lambda, [\mathbb{H}_0, B_2 - B_2^\ast]\xi_\lambda\rangle | \leq 2\int dxdy\,| V_\varphi(x-y)||\langle \xi_\lambda, a_\uparrow(u^\gg_x)a_\uparrow( v^\gg_x)a_\downarrow(u^\gg_y)a_\downarrow( v^\gg_y)\xi_\lambda\rangle|\nonumber
  \\
  &\leq C\rho^{\frac{3}{2}} \int dxdy\, |V_\varphi(x-y)|\|a_\uparrow(u^\gg_x)\xi_\lambda\| \leq CL^{\frac{3}{2}}\|V_\varphi\|_1\rho^{\frac{3}{2} -\frac{1}{3}-\frac{1}{12}}\|\mathbb{H}_0^{\frac{1}{2}}\xi_\lambda\|\leq CL^3\rho^{\frac{7}{3} - \frac{1}{6}} + C\langle \xi_\lambda,\mathbb{H}_0 \xi_\lambda\rangle,
\end{align}
where we also used that $\|a_\sigma( v^\gg_\cdot)\|\leq C\rho^{1/2}$, $\|a_\sigma(u^\gg_\cdot)\|\leq C\rho^{1/2}$, $\|V_\varphi\|_1 \leq C$ and that 
\[
  \int dx\, \|a_\uparrow(u^{\gg}_x)\xi_\lambda\|^2 \leq C\rho^{-\frac{2}{3} -\frac{1}{6}}\langle \xi_\lambda,\mathbb{H}_0\xi_\lambda\rangle.
\]
Inserting \eqref{eq: comm H0 B2} in \eqref{eq: der H0 T2}, from Gr\"onwall's Lemma we get
\[
 \langle \xi_\lambda,\mathbb{H}_0 \xi_\lambda\rangle \leq C\langle \xi_0, \mathbb{H}_0\xi_0\rangle + CL^3\rho^{2+\frac{1}{6}}.
\]
From the propagation estimates proved in Proposition \ref{pro: propagation est}, we know that $\langle \xi_0, \mathbb{H}_0\xi_0\rangle \leq CL^3\rho^2$, this concludes the proof.
\end{proof}
\begin{lemma}[Useful bound]\label{lem: integral t} Let $\hat{u}^t_\sigma$, $\hat{v}^t_\sigma$ as in \eqref{eq def ut vt}. It holds that
\begin{equation}\label{eq: est int t uv final}
  \int_0^{\infty} dt \, \|\hat{u}^{t}_\sigma\|_2 \|\hat{v}^t_\sigma\|_2  \leq C\rho^{\frac{1}{6}}.
\end{equation}
\end{lemma}
\begin{proof}
From the definitions of $\hat{u}^t_\sigma$, $\hat{v}^t_\sigma$ in \eqref{eq def ut vt} and by Cauchy-Schwarz's inequality, we get
\[
  \int_0^{\infty} dt \, \|\hat{u}^{t}_\sigma\|_2 \|\hat{v}^t_\sigma\|_2 \leq \sqrt{\int_0^{\infty} dt \frac{1}{L^3}\sum_k |\hat{u}^{\gg}_\sigma(k)|^2 e^{-2t(|k|^2 - (k_F^\sigma)^2)}}\sqrt{ \int_0^{\infty} dt\, \frac{1}{L^3}\sum_k |\hat{v}^{\gg}_\sigma (k)|^2e^{-2t((k_F^\sigma)^2 - |k|^2)}}.
\]
We now estimate the right hand side above. We have 
\begin{equation}\label{eq: int t u}
  \int_0^{\infty}dt\,  \frac{1}{L^3}\sum_k |\hat{u}^{\gg}_\sigma(k)|^2 e^{-2t(|k|^2 - (k_F^\sigma)^2)} = \frac{1}{L^3}\sum_k \frac{|u^\gg_\sigma(k)|^2}{|k|^2 - (k_F^\sigma)^2}\leq \frac{1}{L^3}\sum_{k_F^\sigma + (k_F^\sigma)^{3/2} \leq |k| \leq 3k_F^\sigma}\frac{1}{|k|^2 - (k_F^\sigma)^2} \leq C\rho^{\frac{1}{6}},
\end{equation}
Proceeding similarly as above, we get 
\begin{equation}\label{eq: est int t vt}
  \int_0^{\infty} dt\, \frac{1}{L^3}\sum_k |\hat{v}^\gg_\sigma(k)|^2 e^{-2t( (k_F^\sigma)^2 - |k|^2)} = \frac{1}{L^3}\sum_k\frac{|\hat{v}^\gg_\sigma(k)|^2}{(k_F^\sigma)^2 - |k|^2}\leq C\rho^{\frac{1}{6}}.
\end{equation}
Combining the estimates in \eqref{eq: int t u} and \eqref{eq: est int t vt}, we conclude the proof.
\end{proof}
\begin{proposition}[Propagation estimates for $\widetilde{\mathbb{Q}}_>$]\label{pro: tildeQ} Let $\lambda\in [0,1]$ and let $\psi$ be an approximate ground state in the sense of Definition \ref{def: approx gs}. Under the same assumptions as in Theorem \ref{thm: optimal lw bd}, it holds that
\[
  \left|\partial_\lambda\langle T^\ast_{2;\lambda}T^\ast_1 R^\ast \psi, \widetilde{\mathbb{Q}}_> T^\ast_{2;\lambda}T^\ast_1 R^\ast\psi\rangle \right| \leq CL^3\rho^{\frac{7}{3}}.
\]
\end{proposition}
\begin{proof}
To simplify the notation, we write $\xi_\lambda = T^\ast_{2;\lambda}T^\ast_1 R^\ast\psi$. We then start the proof by computing 
\[
  \partial_\lambda\langle \xi_\lambda, \widetilde{\mathbb{Q}}_> \xi_\lambda\rangle = -\langle \xi_\lambda, [\widetilde{\mathbb{Q}}_>, B_2]\xi_\lambda\rangle + \mathrm{c.c.}
\]
Writing $B_2$ in configuration space (see \eqref{eq: B2 conf space})
\[
  B_2 = 2\int_0^{\infty} dt \int dzdz^\prime\, V_\varphi(z-z^\prime)\,a_\uparrow(u^t_z)a_\uparrow( v^t_z)a_\downarrow(u^t_{z^\prime})a_\downarrow( v^t_{z^\prime}).
\]
and computing the commutator above, we find the error terms to estimate. The structure of those is similar to the ones in Proposition \ref{pro: Q2}. We then omit the calculations and directly estimate them. The first one we consider corresponds to \eqref{eq: term I Q4}, and explicitly is:
\[
  \mathrm{I}_a = 2\int_0^{\infty}dt\int dxdydzdz^\prime\, V(x-y)V_\varphi(z-z^\prime)v_\uparrow^t(x;z) \langle \xi_\lambda, a_\uparrow^\ast(u_x^\gg)a_\downarrow^\ast(u_y^\gg) a^\ast_\downarrow( v_y^\gg)a_\downarrow( v_{z^\prime}^t) a_\downarrow(u_{z^\prime}^t)a_\uparrow(u_{z}^t)\xi_\lambda\rangle,
\]
where $v_\sigma^t(x;y)$  is as in \eqref{eq: def ut vt x space} (see also \eqref{eq: def ut vt x space}). It is convenient to write 
\begin{multline*}
  \mathrm{I}_a = \int_0^{\infty}dt\, \mathrm{I}_a^t = \int_0^{\infty}dt \frac{2}{L^3}\sum_k \hat{v}^\gg_\uparrow(k)\hat{v}^t_\uparrow(k)\times 
  \\
  \times\left\langle \left(\int dxdy\, e^{-ik\cdot x}V(x-y)a_\uparrow(u^\gg_x)a_\downarrow(u^\gg_y)a_\downarrow( v^\gg_y)\xi_\lambda\right), \left(\int dzdz^\prime\, e^{-ik\cdot z}V_\varphi(z-z^\prime)a_\downarrow( v_{z^\prime}^t) a_\downarrow(u_{z^\prime}^t)a_\uparrow(u_{z}^t)\xi_\lambda\right)\right\rangle.
\end{multline*}
Therefore, by Cauchy-Schwarz's inequality and using that $|\hat{v}^t_\uparrow(k)| \leq e^{t(k_F^\uparrow)^2}$ , we find that 
\begin{multline}\label{eq: CS Iat}
  |\mathrm{I}_a^t| \leq e^{t(k_F^\uparrow)^2}\sqrt{\frac{1}{L^3}\sum_k\left\|\int dxdy\, e^{-ik\cdot x}V(x-y)a_\uparrow(u^\gg_x)a_\downarrow(u^\gg_y)a_\downarrow( v^\gg_y)\xi_\lambda\right\|^2}\times 
  \\
  \times \sqrt{\frac{1}{L^3}\sum_k\left\|\int dzdz^\prime\, e^{-ik\cdot z}V_\varphi(z-z^\prime)a_\downarrow( v_{z^\prime}^t) a_\downarrow(u_{z^\prime}^t)a_\uparrow(u_{z}^t)\xi_\lambda\right\|^2}.
\end{multline}
We now estimate the two terms above. For the first one, using that $\|a_\sigma(u^\gg_\cdot)\|, \|a_\sigma( v^\gg_\cdot)\|\leq C\rho^{1/2}$ and \eqref{eq: N>> H0} (with $\epsilon =1/2$), we obtain
\begin{multline*}
\frac{1}{L^3}\sum_k\left\|\int dxdy\, e^{-ik\cdot x}V(x-y)a_\uparrow(u^\gg_x)a_\downarrow(u^\gg_y)a_\downarrow( v^\gg_y)\xi_\lambda\right\|^2
\leq C\rho^2 \int dxdydy^\prime |V(x-y)||V(x-y^\prime)| \|a_\uparrow(u_x^\gg)\xi_\lambda\|^2 
\\
\leq C\|V\|_1^2 \rho^{2} \int dx\, \|a_\uparrow(u_x^\gg)\xi_\lambda\|^2\leq C\rho^{2-\frac{5}{6}}\langle \xi_\lambda,\mathbb{H}_0 \xi_\lambda\rangle.
\end{multline*}
We can proceed similarly for the other term in \eqref{eq: CS Iat}. Therefore, we get
\begin{multline*}
\frac{1}{L^3}\sum_k\left\|\int dzdz^\prime\, e^{-ik\cdot z}V_\varphi(z-z^\prime)a_\downarrow( v_{z^\prime}^t) a_\downarrow(u_{z^\prime}^t)a_\uparrow(u_{z}^t)\xi_\lambda\right\|^2
\\
\leq C\|V_\varphi\|_1^2 \|u^t_\downarrow\|_2^2 \|v^t_\downarrow\|_2^2\int dz\, \|a_\uparrow(u^t_z)\xi_\lambda\|^2 \leq C\|u^t_\downarrow\|_2^2 \|v^t_\downarrow\|_2^2 e^{-2t(k_F^\uparrow)^2}\rho^{-\frac{5}{6}} \langle \xi_\lambda,\mathbb{H}_0 \xi_\lambda\rangle,
\end{multline*}
where we also used that 
\begin{multline}\label{eq: Nt H0}
  \int dz\,\|a_\uparrow(u^t_z)\xi_\lambda\|^2 =\sum_k |\hat{u}^t_\uparrow(k)|^2\|\hat{a}_{k,\uparrow} \xi_\lambda\|^2 \leq e^{-2t(k_F^\uparrow)^2}\sum_{|k| \geq k_F^\uparrow + (k_F^\uparrow)^{\frac{3}{2}}} \|\hat{a}_{k,\uparrow}\xi_\lambda\|^2 \leq Ce^{-2t(k_F^\downarrow)^2}\rho^{-\frac{5}{6}}\langle \xi_\lambda, \mathbb{H}_0 \xi_\lambda\rangle.
\end{multline}
Combining the estimates together, we get
\[
  |\mathrm{I}_a| \leq \int_0^{\infty}dt\, |\mathrm{I}_a^t| \leq C \|V\|_1 \|V_\varphi\|_1 \rho^{1-\frac{5}{6}}\left(\int dt\, e^{t(k_F^\uparrow)^2}e^{-t(k_F^\uparrow)^2} \|u^t_\uparrow\|_2 \|v^t_\uparrow\|_2\right)\langle\xi_\lambda,\mathbb{H}_0 \xi_\lambda\rangle \leq C\rho^{\frac{1}{3}}\langle\xi_\lambda,\mathbb{H}_0 \xi_\lambda\rangle,
\]
where we also used Lemma \ref{lem: integral t}.
The next error term that we consider corresponds to the one in \eqref{eq: term Ib Q4} in Proposition \ref{pro: Q2}. Explicitly it reads
\[
  \mathrm{I}_b= 2\int_0^{\infty}\hspace{-0.25cm}dt\int dxdydzdz^\prime\, V(x-y)V_\varphi(z-z^\prime)v_\uparrow^t(x;z)v_\downarrow^t(y;z^\prime)\langle \xi_\lambda, a^\ast_\uparrow(u^\gg_x)a^\ast_{\downarrow}(u^\gg_y) a_\downarrow(u^t_{z^\prime})a_\uparrow(u^t_{z})\xi_\lambda\rangle = \int_0^{\infty}\hspace{-0.25cm}dt\, \mathrm{I}_b^t.
\]
We then use that $\|a_\downarrow^\ast(u^\gg_y) a_\downarrow(u^t_{z^\prime})\|\leq \|u^\gg_\downarrow\|_2 \|u^t_\downarrow\|_2\leq C\rho^{1/2}\|u^t_\downarrow\|_2$ and Cauchy-Schwarz to write
\begin{eqnarray*}
  |\mathrm{I}_b^t| 
  &\leq& C\rho^{\frac{1}{2}}\|u_\downarrow^t\|_2\left(\int dxdydzdz^\prime\, |V(x-y)||V_\varphi(z-z^\prime)||v_\uparrow^t(x;z)||v_\downarrow^t(y;z^\prime)|\|a_\uparrow(u^\gg_x)\xi_\lambda\|^2\right)^{\frac{1}{2}}\times 
  \\
  &&\times \left(\int dxdydzdz^\prime\, |V(x-y)||V_\varphi(z-z^\prime)||v_\uparrow^t(x;z)||v_\downarrow^t(y;z^\prime)|\|a_\uparrow(u^t_{z})\xi_\lambda\|^2\right)^{\frac{1}{2}}
  \\
  &\leq& C\|V\|_1\|V_\varphi\|_1\rho^{\frac{1}{2}-\frac{5}{6}}\|u^t_\downarrow\|_2 \|v^t_\downarrow\|_2 \|v^t_\uparrow\|_2 e^{-t(k_F^\uparrow)^2}\langle \xi_\lambda,\mathbb{H}_0 \xi_\lambda\rangle,
\end{eqnarray*}
where we used again \eqref{eq: N>> H0} (with $\varepsilon = 1/2$), \eqref{eq: Nt H0} together with 
\begin{equation}\label{eq: int V ww CS Q2 T2}
  \int dxdy\, |V(x-y)||v_\uparrow^t(x;z)||v_\downarrow^t(y;z^\prime)|\leq \|V\|_1\|v_\uparrow^t\|_2 \|v_\downarrow^t\|_2,\quad 
\end{equation}
and the analogous estimate with $V$ replaced by $V_\varphi$.
Using now that  $\|v^t_\uparrow\|_2 \leq C\rho^{1/2} e^{t(k_F^\uparrow)^2}$, we find 
\[
  |\mathrm{I}_b| \leq C\rho^{1-\frac{5}{6}}\left(\int_0^{\infty}dt\, e^{t(k_F^\uparrow)^2}e^{-t(k_F^\uparrow)^2} \|u^t_\downarrow\|\|v^t_\downarrow\|\right)\langle \xi_\lambda,\mathbb{H}_0\xi_\lambda\rangle \leq C\rho^{\frac{1}{3}}\langle \xi_\lambda,\mathbb{H}_0\xi_\lambda\rangle,
\]
using again \ref{lem: integral t}.
We now consider the error term corresponding to the one in \eqref{eq: IIb Q2}. After a change of variables, we  have 
\[
  \mathrm{II}_a = 2\int_0^{\infty}\hspace{-0.35cm} dt\int dxdydzdz^\prime\,{V}(x-y)V_\varphi(z-z^\prime)u^{t}_\uparrow(z;x) \langle \xi_\lambda, a^\ast_\downarrow(u^\gg_y)a_\downarrow(u^t_{z^\prime}) a_\uparrow( v_z^t)a_\downarrow( v_{z^\prime}^t) a^\ast_{\downarrow}( v_y^\gg)a^\ast_\uparrow( v_x^\gg)  \xi_\lambda\rangle=: \int_0^{\infty}\hspace{-0.25cm}\mathrm{II}_b^t,
\]
where we used the notation 
\[
  u^t_\sigma(x;y) = \frac{1}{L^3}\sum_k \hat{u}^t_\sigma(k)\hat{u}^\gg_\sigma(k)e^{ik\cdot (x-y)}.
\]
Note that in the estimate for $\mathrm{II}_a$, the dependence on the spin is important. We have in total four possible errors of the form of $\mathrm{II}_a$ with different spin combinations. All of these terms can be estimated similarly as $\mathrm{II}_a$, we therefore estimate $\mathrm{II}_a$ and omit the details for the others. It is convenient to rewrite $\mathrm{II}_a$ in normal order. We have $\mathrm{II}_a = \mathrm{II}_{a;1}+\mathrm{II}_{a;2} + \mathrm{II}_{a;3} + \mathrm{II}_{a;4}$ with
\begin{equation*}
  \mathrm{II}_{a;1} = -2\int_0^{\infty}dt\int dxdydzdz^\prime\, V(x-y)V_\varphi(z-z^\prime)u^t_\uparrow(z;x) \langle \xi_\lambda, a^\ast_{\downarrow}( v_y^\gg)a^\ast_\uparrow( v_x^\gg)a^\ast_\downarrow(u^\gg_y)a_\downarrow(u^t_{z^\prime})a_\uparrow( v_z^t)a_\downarrow( v_{z^\prime}^t)\xi_\lambda\rangle,
\end{equation*}
\[
  \mathrm{II}_{a;2} = 2\int_0^{\infty}dt\, \int dxdydzdz^\prime\, V(x-y)V_\varphi(z-z^\prime) v_\uparrow^t(x;z)v_\downarrow^t(y;z^\prime) u^t_\uparrow(z;x) \langle\xi_\lambda, a^\ast_\downarrow(u^\gg_y)a_\downarrow(u^t_{z^\prime})\xi_\lambda\rangle,
\]
\[
 \mathrm{II}_{a;3} = -2\int_0^{\infty}dt \int dxdydzdz^\prime V(x-y)V_\varphi(z-z^\prime)u^t_\uparrow(z;x)v_\downarrow^t(y;z^\prime) \langle \xi_\lambda, a^\ast_\downarrow(u^\gg_y) a^\ast_\uparrow( v_x^\gg) a_\uparrow( v_{z}^t)a_\downarrow (u_{z^\prime}^t) \xi_\lambda\rangle,
\]
\[
 \mathrm{II}_{a;4} = 2\int_0^{\infty}dt \int dxdydzdz^\prime V(x-y)V_\varphi(z-z^\prime)u^t_\uparrow(z;x)v_\uparrow^t(x;z) \langle \xi_\lambda, a^\ast_\downarrow(u^\gg_y) a^\ast_\downarrow( v_y^\gg) a_\downarrow( v_{z^\prime}^t)a_\downarrow (u_{z^\prime}^t) \xi_\lambda\rangle.
\]
We write $\mathrm{II}_{a;1}=\int_0^{\infty}dt\,\mathrm{II}_{a;1}^t$, and begin by estimating $\mathrm{II}_{a;1}^t$. The argument proceeds similarly to the case $\mathrm{I}_a$. Expressing $u^t_\uparrow(z;x)$ in momentum space and using that $|\hat{u}^t_\uparrow(k)|\leq e^{-t(k_F^\uparrow)^2}$, we obtain by the Cauchy-Schwarz inequality that
\begin{align*}
  &|\mathrm{II}_{a;1}^t| \leq Ce^{-t(k_F^\uparrow)^2}\left(\frac{1}{L^3}\sum_k \left\| \int dxdy\, e^{ik\cdot x}\, V(x-y) a_\downarrow(u^\gg_y)a_\uparrow( v^\gg_x)a_\downarrow( v^\gg_y)\xi_\lambda\right\|^2\right)^{\frac{1}{2}}\times
  \\
  &\hspace{5cm}\times \left(\frac{1}{L^3}\sum_k \left\| \int dzdz^\prime e^{ik\cdot z}\, V_\varphi(z-z^\prime) a_\downarrow(u^t_{z^\prime})a_\uparrow( v_z^t)a_\downarrow( v_{z^\prime}^t)\xi_\lambda\right\|^2\right)^{\frac{1}{2}}
  \\
  &\leq C\rho e^{-t(k_F^\uparrow)^2}\|V\|_1\|V_\varphi\|_1\|u^t_\downarrow\|_2\|v^t_\downarrow\|_2\sqrt{\int dy\, \|a_\downarrow(u^\gg_y)\xi_\lambda\|^2}\sqrt{\int dz\, \|a_\uparrow( v^t_z)\xi_\lambda\|^2}.
\end{align*}
Using now \eqref{eq: N>> H0} and 
\begin{equation}\label{eq: vt N}
  \int dz\, \|a_\downarrow( v^t_{z})\xi_\lambda\|^2 = \sum_k |\hat{v}^t_\downarrow(k)|^2 \hat{a}_{k,\downarrow}^\ast \hat{a}_{k,\downarrow} \leq e^{t(k_F^\uparrow)^2}\sum_k |\hat{v}^\gg_\downarrow(k)|^2 \hat{a}_{k,\downarrow}^\ast \hat{a}_{k,\downarrow} \leq Ce^{t(k_F^\uparrow)^2}\rho^{-\frac{5}{6}}\langle \xi_\lambda,\mathbb{H}_0\xi_\lambda\rangle,
\end{equation}
we get
\[
  |\mathrm{II}_{a;1}| \leq C\rho^{\frac{1}{6}} \left(\int_0^{\infty} dt \, \|u^t_\downarrow\|_2 \|v^t_\downarrow\|_2\right)\langle \xi_\lambda, \mathbb{H}_0 \xi_\lambda\rangle \leq C\rho^{\frac{1}{3}}\langle \xi_\lambda,\mathbb{H}_0\xi_\lambda\rangle,
\]
where Lemma \ref{lem: integral t} was used in the last step. We now proceed to estimate $\mathrm{II}_{a;2}$. Using the bound $\|v^t_\downarrow\|_\infty \leq C\rho e^{t(k_F^\downarrow)^2}$ and applying the Cauchy-Schwarz inequality, we obtain
\begin{eqnarray*}
  |\mathrm{II}_{a;2}| 
  &\leq& C\rho\int_0^{\infty}dt\, e^{t(k_F^\downarrow)^2}\left(\int dxdydzdz\, |V(x-y)||V_\varphi(z-z^\prime)| |v_\uparrow^t(x;z)|| u^t_\uparrow(z;x)|  \|a_\downarrow(u^\gg_y)\xi_\lambda\|^2\right)^{\frac{1}{2}}\times 
  \\
  && \times \left(\int dxdydzdz\, |V(x-y)||V_\varphi(z-z^\prime)| |v_\uparrow^t(x;z)|| u^t_\uparrow(z;x)|  \|a_\downarrow(u^t_{z^\prime})\xi_\lambda\rangle\|^2\right)^{\frac{1}{2}}
  \\
  &\leq& C\rho\|V\|_1\|V_\varphi\|_1\int_0^{+\infty}dt\, e^{t(k_F^\downarrow)^2} e^{-t(k_F^\downarrow)^2}\|u^t_\uparrow\|_2\|v^t_\uparrow\|_2 \sqrt{\int dy\,  \|a_\downarrow(u^\gg_y)\xi_\lambda\|^2}\sqrt{\int dz\, \|a_\downarrow(u^t_{z^\prime})\xi_\lambda\rangle\|^2},
\end{eqnarray*}
where we use analogous estimates to \eqref{eq: int V ww CS Q2 T2}. Using also \eqref{eq: N>> H0} (with $\varepsilon = 1/2$) and \eqref{eq: Nt H0}, we find that 
\[
  |\mathrm{II}_{a;2}| \leq C\rho^{1-\frac{5}{6}}\left(\int_0^{\infty} dt\, \|u^t_\uparrow\|_2 \|v^t_\uparrow\|_2\right)\langle \xi_\lambda,\mathbb{H}_0 \xi_\lambda\rangle \leq C\rho^{\frac{1}{3}}\langle \xi_\lambda,\mathbb{H}_0 \xi_\lambda\rangle,
\]
where we have also used Lemma \ref{lem: integral t}. We now consider the next error term, $\mathrm{II}_{a;3}$. Using the bound $\|a^\ast_\uparrow( v_x^\gg) a_\uparrow( v_{z}^t)\| = \|v^\gg_\uparrow\|_2 \|v^t_\uparrow\|_2 \leq C\rho^{1/2}\|v^t_\uparrow\|_2$ and applying  Cauchy-Schwarz, we proceed in the same way as above to obtain
\begin{eqnarray*}
  |\mathrm{II}_{a;3}| 
  &\leq& C\rho^{\frac{1}{2}}\int_0^{\infty}dt\,\|v^t_\uparrow\|_2\left(\int dxdydzdz^\prime |V(x-y)||V_\varphi(z-z^\prime)||u^t_\uparrow(z;x)||v_\downarrow^t(y;z^\prime)| \|a_\downarrow(u^\gg_y)\xi_\lambda\|^2 \right)^{\frac{1}{2}}\times 
  \\
  &&\times \left(\int dxdydzdz^\prime |V(x-y)||V_\varphi(z-z^\prime)||u^t_\uparrow(z;x)||v_\downarrow^t(y;z^\prime)|  \|a_\downarrow (u_{z^\prime}^t) \xi_\lambda\|^2\right)^{\frac{1}{2}}
  \\
  &\leq& C\rho^{\frac{1}{2}}\|V\|_1\|V_\varphi\|_1\int_0^{\infty}dt\,\|v^t_\uparrow\|_2\|u^t_\uparrow\|_2 \|v_\downarrow^t\|_2\left(\int dy \, \|a_\downarrow(u^\gg_y)\xi_\lambda\|^2\right)^{\frac{1}{2}}\left(\int dz^\prime \, \|a_\downarrow(u^t_{z^\prime})\xi_\lambda\|^2 \right)^{\frac{1}{2}}.
\end{eqnarray*}
By $\|v_\downarrow^t\|_2 \leq C\rho^{1/2} e^{t(k_F^\downarrow)^2}$ together with \eqref{eq: N>> H0} (for $\varepsilon =1/2$), \eqref{eq: Nt H0} and Lemma \ref{lem: integral t}, we obtain
\[
  |\mathrm{II}_{b;3}| \leq  C\rho^{\frac{1}{3} - \frac{1}{6}} \left(\int_0^{\infty} dt \, \|u^t_\downarrow\|_2 \|v^t_\downarrow\|_2\right)\langle \xi_\lambda, \mathbb{H}_0 \xi_\lambda\rangle \leq C\rho^{\frac{1}{3}}\langle \xi_\lambda,\mathbb{H}_0\xi_\lambda\rangle.
\]
The error term $\mathrm{II}_{a;4}$ can be estimated similarly, we omit the details. Combining all the estimates, we find 
\[
  |\mathrm{II}_{a}| \leq C\rho^{\frac{1}{3}}\langle \xi_\lambda,\mathbb{H}_0 \xi_\lambda\rangle.
\]
We now consider $\mathrm{II}_b$, which corresponds to \eqref{eq: term IIa Q2}: 
\[
  \mathrm{II}_b = -4\int_0^{\infty}\hspace{-0.35cm}dt \int dxdydzdz^\prime\, V(x-y)V_\varphi(z-z^\prime)u^t_\downarrow(z^\prime;y)u^t_\uparrow(z;x) \langle \xi_\lambda, a_\uparrow( v_z^t)a_\downarrow( v_{z^\prime}^t) a^\ast_{\downarrow}( v_y^\gg)a^\ast_\uparrow( v_x^\gg)\xi_\lambda\rangle .
\]
It is convenient to put it in normal order, we have $\mathrm{II}_b = \mathrm{II}_{b;1} + \mathrm{II}_{b;2} + \mathrm{II}_{b;3} + \mathrm{II}_{b;4}$ with
\[
  \mathrm{II}_{b;1} =-2 \int_0^{\infty}dt\int dxdydzdz^\prime\, V(x-y)V_\varphi(z-z^\prime)u^t_\downarrow(z^\prime;y)u^t_\uparrow(z;x) \langle \xi_\lambda, a^\ast_{\downarrow}( v_y^\gg)a^\ast_\uparrow( v_x^\gg)a_\uparrow( v_z^t)a_\downarrow( v_{z^\prime}^t) \xi_\lambda\rangle,
\]
\[
  \mathrm{II}_{b;2} = 2\int_0^{\infty}dt\int dxdydzdz^\prime V(x-y)V_\varphi(z-z^\prime)u^t_\downarrow(z^\prime;y)u^t_\uparrow(z;x)v_\uparrow^t(x;z) \langle \xi_\lambda,a^\ast_\downarrow( v_y^\gg) a_\downarrow ( v_{z^\prime}^t)\xi_\lambda\rangle,
\]
\[
  \mathrm{II}_{b;3} = 2\int_0^{\infty}dt\int dxdydzdz^\prime V(x-y)V_\varphi(z-z^\prime)u^t_\downarrow(z^\prime;y)u^t_\uparrow(z;x)v_\uparrow^t(y;z^\prime) \langle \xi_\lambda,a^\ast_\downarrow( v_x^\gg) a_\downarrow (v_{z}^t)\xi_\lambda\rangle,
\]
\[
  \mathrm{II}_{b;4} = -2\frac{1}{L^6}\sum_{p,r,r^\prime}\hat{V}(p)\hat{f}_{r,r^\prime}(p)(\hat{u}^\gg_\uparrow(r+p))^2(\hat{u}^\gg_\downarrow(r^\prime - p))^2\hat{v}_\uparrow^\gg(r)\hat{v}_\downarrow^\gg(r^\prime).
\]
Note that in $\mathrm{II}_{a;4}$ we used that $\hat{v}^\gg = (\hat{v}^\gg)^2$, which does not hold for $\hat{u}^\gg$.
Using the bound  $\|a^\ast_\uparrow( v_x^\gg)a_\uparrow( v_z^t)\|\leq C\rho^{1/2}\|v^t_\uparrow\|_2$ and applying Cauchy-Schwarz inequality, we find
\begin{eqnarray*}
  |\mathrm{II}_{b;1}^t| \leq  C\rho^{\frac{1}{2}-\frac{5}{6}}\|V\|_1\|V_\varphi\|_1\|u_\uparrow^t\|_2\|u^t_\downarrow\|_2 \|v^t_\uparrow\|_2 e^{t(k_F^\downarrow)^2}\langle \xi_\lambda,\mathbb{H}_0\xi_\lambda\rangle,
\end{eqnarray*}
where we used analogous estimates to \eqref{eq: int V ww CS Q2 T2}, \eqref{eq: vt N} and 
\begin{equation}\label{eq: v>> N}
  \int dy, \|a_\downarrow( v^\gg_y)\xi_\lambda\|^2 = \sum_k |\hat{v}^\gg_\downarrow(k)|^2 \hat{a}_{k,\downarrow}^\ast \hat{a}_{k,\downarrow} \leq \sum_{|k|<k_F^\downarrow - (k_F^\downarrow)^{\frac{3}{2}}}  \hat{a}_{k,\downarrow}^\ast \hat{a}_{k,\downarrow}\leq C\rho^{-\frac{5}{6}}\langle \xi_\lambda,\mathbb{H}_0 \xi_\lambda\rangle.
\end{equation}
Therefore, using also that $\|u^t_\downarrow\|_2\leq C\rho^{1/2}e^{-t(k_F^\downarrow)^2}$, we obtain
\[
  |\mathrm{II}_{b;1}| \leq C\rho^{\frac{1}{6}} \left(\int_0^{+\infty}dt\, \|u^t_\uparrow\|_2\|v^t_\uparrow\|\right)\langle \xi_\lambda,\mathbb{H}_0\xi_\lambda\rangle\leq C\rho^{\frac{1}{3}}\langle \xi_\lambda,\mathbb{H}_0\xi_\lambda\rangle, 
\]
where we have again used Lemma \ref{lem: integral t}. Proceeding analogously, and additionally using the bound $\|u^t_\downarrow\|_\infty \leq C\rho e^{-t(k_F^\downarrow)^2}$, we obtain the following estimate for $\mathrm{II}_{b;2}$:
\begin{eqnarray*}
  |\mathrm{II}_{b;2}| \leq C\rho^{\frac{1}{6}} \|V\|_1 \|V_\varphi\|_1 \left(\int dt\, \|u^t_\uparrow\|_2 \|v^t_\uparrow\|_2\right) \langle \xi_\lambda,\mathbb{H}_0\xi_\lambda\rangle \leq C\rho^{\frac{1}{3}}\langle \xi_\lambda, \mathbb{H}_0\xi_\lambda\rangle,
\end{eqnarray*}
where we have also used \eqref{eq: v>> N} and \eqref{eq: vt N}. The term $\mathrm{II}_{b;3}$
can be estimated in the same way as $\mathrm{II}_{b;2}$,  we omit the details. Finally, we consider $\mathrm{II}_{b;4}$. Using the definition of $\hat{f}_{r,r^\prime}$ (see \eqref{eq: def frr'}), we  write
\[
  \mathrm{II}_{\mathrm{b;4}} = -\frac{4}{L^6}\sum_{p,r,r^\prime}\hat{V}(p)\frac{\widehat{V_\varphi}(p)}{|r+p|^2 - |r|^2 + |r^\prime - p|^2 - |r^\prime|^2}(\hat{u}^\gg_\uparrow(r+p))^2(\hat{u}^\gg_\downarrow(r^\prime - p))^2\hat{v}_\uparrow^\gg(r)\hat{v}_\downarrow^\gg(r^\prime).
\]
Using that $\|\hat{V}(p)\widehat{V_\varphi}(p)\|_\infty \leq \|V\|_1\|V_\varphi\|_1 \leq C$ and taking $L$ large enough, we find that
\[
  L^{-3}|\mathrm{II}_{\mathrm{b;4}}| \leq C\int_{\mathbb{R}^3}dp\int_{ |r| < k_F^\uparrow <|r+p|> k_F^\uparrow}dr\int_{ |r^\prime| < k_F^\downarrow < |r^\prime - p|> k_F^\downarrow }dr^\prime \frac{1}{|r+p|^2 - |r|^2 + |r^\prime - p|^2 - |r^\prime|^2}.
\]
By re-scaling via (for instance) $k_F^\uparrow$, we get that 
\[
  L^{-3}|\mathrm{II}_{b;4}| \leq C(k_F^\uparrow)^{\frac{7}{3}}\int_{\mathbb{R}^3}dp\int_{ |r| < 1 < |r+p|}dr\int_{ |r^\prime| < \frac{k_F^\downarrow}{k_F^\uparrow} < |r^\prime -p|}dr^\prime \frac{1}{|r+p|^2 - |r|^2 + |r^\prime - p|^2 - |r^\prime|^2}.
\]
The integral above is well-know to be uniformly bounded (see e.g. \cite{CW,Kanno} and \cite[Appendix C]{GHNS24}). Therefore, we can conclude that 
\[
  |\mathrm{II}_{4;b}| \leq CL^3 \rho^{\frac{7}{3}}.
\]
Putting all the estimates together and using the propagation estimate in \eqref{eq: prop H0 T2}, we conclude the proof. 
\end{proof}
\subsection{Conjugation under the first quadratic quasi-bosonic transformation $T_2$: conclusions}\label{sec: lw bd T2}
In this section we conclude the proof of Theorem \ref{thm: optimal lw bd}. The starting point is \eqref{eq: final T1}, which holds for any $\psi$ approximate ground state (see Definition \ref{def: approx gs}) and for $L$ large enough:
\begin{equation}\label{eq: fin T1 sec T2}
E_{L}(N_\uparrow, N_\downarrow) \geq \frac{1}{2}\langle T^\ast_1 R^\ast \psi, (\mathbb{H}_0  + 2\widetilde{\mathbb{Q}}_>) T^\ast_1 R^\ast \psi \rangle + \frac{3}{5}(6\pi^2)^{\frac{2}{3}} (\rho_\uparrow^{\frac{5}{3}} + \rho_\downarrow^{\frac{5}{3}})L^3 + 8\pi a\rho_\uparrow\rho_\downarrow L^3  + \mathcal{E}_{\mathrm{T}_1},
\end{equation}
with $|\mathcal{E}_{\mathrm{T}_1}|\leq CL^3\rho^{7/3}$. We now use Duhamel's formula and get: 
\begin{multline*}
  \frac{1}{2}\langle T^\ast_1 R^\ast \psi, (\mathbb{H}_0  + 2\widetilde{\mathbb{Q}}_>) T^\ast_1 R^\ast \psi \rangle 
  \\
  = \frac{1}{2}\langle T^\ast_2 T^\ast_1 R^\ast \psi, (\mathbb{H}_0  +2\widetilde{\mathbb{Q}}_>) T^\ast_2T^\ast_1 R^\ast \psi \rangle - \frac{1}{2}\int_0^1 d\lambda\, \partial_\lambda \langle T^\ast_{2;\lambda}T^\ast_1 R^\ast \psi, (\mathbb{H}_0 +  2\widetilde{\mathbb{Q}}_>) T^\ast_{2;\lambda}T^\ast_1 R^\ast \psi\rangle.
\end{multline*}
We now use Proposition \ref{pro: prop est H0 T2}, see \eqref{eq: der H0 T2} and \eqref{eq: comm H0B2 f}, to write 
\[
  \partial_\lambda\langle T^\ast_{2;\lambda}T^\ast_1 R^\ast \psi, \mathbb{H}_0  T^\ast_{2;\lambda}T^\ast_1 R^\ast \psi\rangle
  =  2\langle T^\ast_{2;\lambda}T^\ast_1 R^\ast \psi,\widetilde{\mathbb{Q}}_> T^\ast_{2;\lambda}T^\ast_1 R^\ast \psi \rangle.
\]
Using Duhamel's formula once more, we thus get 
\begin{multline*}
  \frac{1}{2}\langle T^\ast_1 R^\ast \psi, (\mathbb{H}_0  +2\widetilde{\mathbb{Q}}_>) T^\ast_1 R^\ast \psi \rangle  = \frac{1}{2}\langle T^\ast_2 T^\ast_1 R^\ast \psi, \mathbb{H}_0 T^\ast_2T^\ast_1 R^\ast \psi \rangle 
  \\
  - \int_0^1d\lambda\int_{1}^\lambda d\lambda^\prime\, \partial_{\lambda^\prime}\langle T^\ast_{2;\lambda^\prime}T^\ast_1 R^\ast \psi, \widetilde{\mathbb{Q}}_>  T^\ast_{2;\lambda^\prime}T^\ast_1 R^\ast \psi\rangle - \int_0^1 d\lambda\, \partial_\lambda\langle T^\ast_{2;\lambda}T^\ast_1 R^\ast \psi, \widetilde{\mathbb{Q}}_> T^\ast_{2;\lambda}T^\ast_1 R^\ast \psi\rangle .
\end{multline*}
From Proposition \ref{pro: tildeQ}, we deduce that 
\[
  \frac{1}{2}\langle T^\ast_1 R^\ast \psi, (\mathbb{H}_0  +2\widetilde{\mathbb{Q}}_>) T^\ast_1 R^\ast \psi \rangle  \geq  \frac{1}{2}\langle T^\ast_2 T^\ast_1 R^\ast \psi, \mathbb{H}_0 T^\ast_2T^\ast_1 R^\ast \psi \rangle - CL^3\rho^{\frac{7}{3}}.
\]
Inserting the bound above in \eqref{eq: fin T1 sec T2} and taking the thermodynamic limit, we find that the energy density satisfies
\[
e(\rho_\uparrow, \rho_\downarrow) \geq \frac{1}{2}\langle T^\ast_2 T^\ast_1 R^\ast \psi, \mathbb{H}_0   T^\ast_2T^\ast_1 R^\ast \psi \rangle + \frac{3}{5}(6\pi^2)^{\frac{2}{3}} (\rho_\uparrow^{\frac{5}{3}} + \rho_\downarrow^{\frac{5}{3}}) + 8\pi a\rho_\uparrow\rho_\downarrow   -C\rho^{\frac{7}{3}} \geq \frac{3}{5}(6\pi^2)^{\frac{2}{3}} (\rho_\uparrow^{\frac{5}{3}} + \rho_\downarrow^{\frac{5}{3}}) + 8\pi a\rho_\uparrow\rho_\downarrow   -C\rho^{\frac{7}{3}},
\]
where we also used the positivity of $\mathbb{H}_0$.
\section{Proof of Theorem \ref{thm: optimal number operator} }\label{sec: number excit}
The proof of Theorem \ref{thm: optimal number operator} is a consequence of Theorem \ref{thm: optimal lw bd}. Indeed, from \eqref{eq: opt lw bd}, we get that for $\psi\in \mathcal{F}_{\mathrm{f}}$ being a normalized $N$--particle state with $N = N_\uparrow + N_\downarrow$ and satisfying,
\[
  \left|\frac{\langle \psi,\mathcal{H}\psi\rangle}{L^3} -\frac{3}{5}(6\pi^2)^{\frac{2}{3}}(\rho_\uparrow^{\frac{5}{3}} + \rho_\downarrow^{\frac{5}{3}}) - 8\pi a \rho_\uparrow\rho_\downarrow\right| \leq C\rho^{\frac{7}{3}},
\]
it holds that
\begin{equation}\label{eq: est H0 T2 7/3}
  \langle T^\ast_2 T^\ast_1 R^\ast\psi, \mathbb{H}_0 T^\ast_2 T^\ast_1 R^\ast \psi\rangle \leq CL^3\rho^{\frac{7}{3}}.
\end{equation}
To prove Theorem \ref{thm: optimal number operator}, we first prove some propagation bounds for the number operator. Similarly as for the proof for Proposition \ref{pro: N}, we use Gr\"onwall's Lemma. We then write $\xi_\lambda$ in place of $T^\ast_{2;\lambda}T^\ast_1 R^\ast \psi$ and compute 
  \begin{eqnarray*}
    |\partial_\lambda \langle \xi_\lambda, \mathcal{N} \xi_\lambda\rangle | &=& |\langle \xi_\lambda, [\mathcal{N}, B_2 - B_2^\ast] \xi_\lambda\rangle |
    \\
    &=&  8\int_0^{\infty} dt\, e^{-tk_F^2}\int dxdy\, V_\varphi(x-y) \langle\xi_\lambda, a_\uparrow(u^{t}_x)a_\uparrow({v}^{t}_x)a_\downarrow(u^{t}_y)a_\downarrow({v}^{t}_y)\xi_\lambda\rangle + \mathrm{c.c.},
  \end{eqnarray*}
  where we used the notations introduced in \eqref{eq def ut vt}.      
By Cauchy-Schwarz inequality and using that $\|a_\sigma(u^t_\cdot)\| = \|u^t_\sigma\|_2$, $\|a_\sigma(v^t_\cdot)\| = \|v^t_\sigma\|_2\leq C\rho^{1/2}e^{t(k_F^\sigma)^2}$, we get
  \begin{align}\label{eq: est der N T2}
    |\partial_\lambda \langle \xi_\lambda, \mathcal{N}\xi_\lambda\rangle| 
    &\leq C\rho^{\frac{1}{2}}\int_0^{\infty} dt\, e^{t(k_F^\downarrow)^2}\|{u}^{t}_\uparrow\|_2\|{v}^{t}_\uparrow\|_2\int dxdy\, V_\varphi(x-y)\|a_\downarrow(u^{t}_y)\xi_\lambda\|\nonumber
    \\
    &\leq CL^{\frac{3}{2}}\rho^{\frac{1}{2}}\|V_\varphi\|_1\left(\int_0^{\infty}dt\,\|{u}^{t}_\uparrow\|_2\|{v}^{,t}_\uparrow\|_2\right)\|\mathcal{N}^{\frac{1}{2}}\xi_\lambda\| \leq CL^{\frac{3}{2}} \rho^{\frac{2}{3}} \|\mathcal{N}^{\frac{1}{2}}\xi_\lambda\|\leq \langle\xi_\lambda, \mathcal{N}\xi_\lambda\rangle + CL^3\rho^{\frac{4}{3}},\nonumber
  \end{align}
  where we also used Lemma \ref{lem: integral t}, together with $\|V_\varphi\|_1 \leq C$ and 
\[
  \int dy\, \|a_\downarrow(u^t_y)\xi_\lambda\|^2  = \sum_k |\hat{u}^{t}_\downarrow(k)|^2\langle \xi_\lambda, \hat{a}_{k,\downarrow}\hat{a}_{k,\downarrow}\xi_\lambda\rangle \leq \sum_{|k| \geq k_F^\downarrow + (k_F^\downarrow)^{3/2}} e^{-t|k|^2}\langle \xi_\lambda, \hat{a}_{k,\downarrow}\hat{a}_{k,\downarrow}\xi_\lambda\rangle\leq Ce^{-t(k_F^\downarrow)^2} \langle \xi_\lambda, \mathcal{N}\xi_\lambda\rangle.
\]
By Gr\"onwall's Lemma, we can then conclude that for any $\lambda\in [0,1]$,
\[
  \langle T^\ast_{2;\lambda}T^\ast_1 R^\ast \psi, \mathcal{N} T^\ast_{2,\lambda}T^\ast_1 R^\ast \psi \rangle \leq C L^3 \rho^{\frac{4}{3}} + \langle  T^\ast_{2}T^\ast_1 R^\ast \psi, \mathcal{N} T^\ast_{2}T^\ast_1 R^\ast \psi \rangle.
\] 
From the propagation estimate \eqref{eq: prop H0 T2} and proceeding as in \cite[Proposition 5.9]{FGHP} (see also \eqref{eq: improv N lw bd}), one can prove that for any $\lambda\in [0,1]$, it holds that
\[
   \langle T^\ast_{2;\lambda}T^\ast_1 R^\ast \psi,\mathcal{N} T^\ast_{2;\lambda}T^\ast_1 R^\ast \psi\rangle \leq CL^{\frac{3}{2}}\rho^{\frac{1}{6}}\|\mathbb{H}_0^{\frac{1}{2}}T^\ast_2T^\ast_1 R^\ast\psi\| + CL^3\rho^{\frac{4}{3}} \leq CL^3\rho^{\frac{4}{3}},
\]
where we also used \eqref{eq: est H0 T2 7/3}. It then follows that  $\langle T^\ast_1 R^\ast \psi,\mathcal{N} T^\ast_1 R^\ast \psi\rangle \leq CL^3\rho^{4/3}$, which together with Proposition \ref{pro: N}, implies that 
\begin{equation}\label{eq: est R*NR}
  \langle R^\ast \psi, \mathcal{N} R^\ast \psi\rangle \leq CL^3\rho^{\frac{4}{3}}.
\end{equation}
From the definition of the particle-hole transformation (see in particular \eqref{eq: def R momentum space}), we then can conclude that 
\[
  \sum_{\sigma = \{\uparrow, \downarrow\}}\sum_{\substack{k\in\frac{2\pi}{L}\mathbb{Z}^3 \\ |k|> k_F^\sigma}}\langle \psi, \hat{a}_{k,\sigma}^\ast \hat{a}_{k,\sigma}\psi\rangle \leq CL^3\rho^{\frac{4}{3}}, \qquad  \sum_{\sigma = \{\uparrow, \downarrow\}}\sum_{\substack{k\in\frac{2\pi}{L}\mathbb{Z}^3 \\ |k| \leq k_F^\sigma}}\langle \psi, \hat{a}_{k,\sigma} \hat{a}_{k,\sigma}^\ast \psi\rangle \leq CL^3\rho^{\frac{4}{3}}.
\]
We now prove the second part of Theorem \ref{thm: optimal number operator}.
We start by recalling that
\[
  \mathcal{N}^{(\epsilon)}_> = \sum_\sigma \sum_{|k| > k_F^\sigma + (k_F^\sigma)^{1+\epsilon}}\hat{a}_{k,\sigma}^\ast\hat{a}_{k,\sigma} =: \sum_\sigma \mathcal{N}^{(\epsilon)}_{>,\sigma},
\]
with $0\leq \epsilon < 1$.
This immediately implies that 
\begin{equation}\label{eq: est Neps T2 lambda 1}
  \langle T^\ast_2 T^\ast_1 R^\ast \psi,  \mathcal{N}^{(\epsilon)}_>T^\ast_2 T^\ast_1 R^\ast \psi\rangle \leq C\rho^{-\frac{2}{3}- \frac{\epsilon}{3}}\langle T^\ast_2 T^\ast_1 R^\ast \psi, \mathbb{H}_0 T^\ast_2 T^\ast_1 R^\ast \psi\rangle \leq CL^3\rho^{\frac{5}{3} -\frac{\epsilon}{3}}.
\end{equation}
We now want to use Gr\"onwall's Lemma to propagate this estimate. A before, to have a shorter notation, we write $\xi_\lambda: = T^\ast_{2;\lambda}T^\ast_1 R^\ast \psi$. In particular, computing $\partial_\lambda \langle \xi_\lambda, \mathcal{N}^{(\epsilon)}_{>,\sigma}\xi_\lambda\rangle$ , see \cite[Proposition 4.15, Corollary 4.16]{Gia1} for more details,  we can prove that for $\sigma \neq \sigma^\prime$:
\begin{equation}\label{eq: gronwall Neps}
  \partial_\lambda\langle \xi_\lambda, \mathcal{N}^{(\epsilon)}_{>,\sigma} \xi_\lambda\rangle \leq CL^{\frac{3}{2}}\rho^{\frac{1}{2}}\|V_\varphi\|_1 \left(\int_0^{+\infty}dt\, \|\hat{u}^t_{\sigma^\prime}\|_2\|\hat{v}_{\sigma^\prime}^t\|_2 \right) \|(\mathcal{N}^{(\epsilon)}_{>,\sigma})^{\frac{1}{2}}\xi_\lambda\|.
\end{equation}
To propagate the estimate for $\mathcal{N}^{(\epsilon)}_{>,\sigma}$  we now estimate the integral in $t$ in a more refined way than in Lemma \ref{lem: integral t}. From the definitions of $\hat{u}^t_\sigma$, $\hat{v}^t_\sigma$ in \eqref{eq: def ut vt x space} and by Cauchy-Schwarz, we get
\begin{align*}
  \int_0^{\infty} dt \, \|\hat{u}^{t}_{\sigma}\|_2 \|\hat{v}^t_{\sigma}\|_2 &\leq \sqrt{\int_0^{\infty} dt \frac{1}{L^3}\sum_k |\hat{u}^{\gg}_{\sigma}(k)|^2 e^{-2t(|k|^2 - (k_F^{\sigma})^2)}}\sqrt{\int_0^{\infty} dt\, \frac{1}{L^3}\sum_k |\hat{v}^{\gg}_{\sigma} (k)|^2e^{-2t((k_F^{\sigma})^2 - |k|^2)}}
  \\
  &= \sqrt{\frac{1}{L^3}\sum_k \frac{|u^\gg_{\sigma}(k)|^2}{|k|^2 - (k_F^\sigma)^2}} \sqrt{\frac{1}{L^3}\sum_k\frac{|\hat{v}^\gg_{\sigma}(k)|^2}{(k_F^\sigma)^2 - |k|^2}}.
\end{align*}
We now estimate the two factors above. We proceed differently than in Lemma \ref{lem: integral t}. In particular, we define $\eta = (k_F^{\sigma})^2 \rho^\delta$ for some $\delta > 0$ and we write 
\begin{eqnarray*}
   \frac{1}{L^3}\sum_k \frac{|u^\gg_\sigma(k)|^2}{|k|^2 - (k_F^\sigma)^2} &=& \frac{1}{L^3}\sum_k \frac{|u^\gg_\sigma(k)|^2}{|k|^2 - (k_F^\sigma)^2 + \eta} + \frac{\eta}{L^3}\sum_k \frac{|u^\gg_\sigma(k)|^2}{(|k|^2 - (k_F^\sigma)^2)(|k|^2 - (k_F^\sigma)^2 + \eta)}
   \\
   &\leq&  \frac{1}{L^3}\sum_k \frac{|u^\gg_\sigma(k)|^2}{|k|^2 - (k_F^\sigma)^2 + \eta} + \frac{\eta}{L^3}\sum_k \frac{|u^\gg_\sigma(k)|^2}{(|k|^2 - (k_F^\sigma)^2)^2}.
\end{eqnarray*}
We then estimate the two terms above. We consider the first sum. Taking $L$ large enough, and rescaling with respect to $k_F^{\sigma}$, we get  
\[
   \frac{1}{L^3}\sum_k \frac{|u^\gg_\sigma(k)|^2}{|k|^2 - (k_F^\sigma)^2 + \eta} \leq C\rho^{\frac{1}{3}}\int_{1 + (k_F^{\sigma^\prime})^{1/2} <|r| <2 } dr\frac{1}{(|r|-1)(|r| + 1) + \rho^{\delta}}  \leq   C\rho^{\frac{1}{3}}\log (1 + \rho^{-\delta}).
\]
Similarly, by rescaling, we also get that 
\[
  \frac{2\eta}{L^3}\sum_k \frac{|u^\gg_\sigma(k)|^2}{(|k|^2 - (k_F^\sigma)^2)^2} \leq C\rho^{\frac{1}{3} + \delta}\int_{1 + (k_F^{\sigma^\prime})^{1/2} <|r| <2 } dr\frac{1}{(|r|^2-1)^2}  \leq C\rho^{\frac{1}{6} + \delta}.
\]
Therefore, we can conclude that for $\delta >0$ small enough, we have
\[
   \frac{1}{L^3}\sum_k \frac{|u^\gg_\sigma(k)|^2}{|k|^2 - (k_F^\sigma)^2}  \leq  C\rho^{\frac{1}{6} + \delta}
\]
Proceeding in a similar way, we can also prove that 
\[
  \frac{1}{L^3}\sum_k\frac{|\hat{v}^\gg_\sigma(k)|^2}{(k_F^\sigma)^2 - |k|^2} \leq  C\rho^{\frac{1}{6} + \delta}.
\]
Inserting the estimates above in \eqref{eq: gronwall Neps} and using that $\|V_\varphi\|_1 \leq C$, we get 
\[
  |\partial_\lambda\langle \xi_\lambda, \mathcal{N}^{(\epsilon)}_{>,\sigma} \xi_\lambda\rangle| \leq  CL^{\frac{3}{2}}\rho^{\frac{1}{2} + \frac{1}{3} + 2\delta} \|(\mathcal{N}^{(\varepsilon)}_{>,\sigma})^{\frac{1}{2}}\xi_\lambda\|\leq \langle \xi_\lambda, \mathcal{N}^{(\epsilon)}_{>,\sigma} \xi_\lambda\rangle + CL^3\rho^{\frac{5}{3} + 4\delta}.
\]
From Gr\"onwall's Lemma together with \eqref{eq: est Neps T2 lambda 1}, we then find that for any $\lambda\in[0,1]$, we have
\begin{equation}\label{eq: N eps T2}
  \langle T^\ast_{2;\lambda}T^\ast_1 R^\ast \psi, \mathcal{N}^{(\epsilon)}_{>,\sigma} T^\ast_{2;\lambda}T^\ast_1 R^\ast \psi\rangle \leq C\langle T^\ast_{2}T^\ast_1 R^\ast \psi, \mathcal{N}^{(\epsilon)}_{>,\sigma} T^\ast_{2;\lambda}T^\ast_1 R^\ast \psi\rangle +  CL^3\rho^{\frac{5}{3} +4\delta} \leq  CL^3\rho^{\frac{5}{3} - \frac{\epsilon}{3}}, 
\end{equation}
which implies that for any $\lambda \in [0,1]$,
\[
   \langle T^\ast_{2;\lambda}T^\ast_1 R^\ast \psi, \mathcal{N}^{(\epsilon)}_{>} T^\ast_{2;\lambda}T^\ast_1 R^\ast \psi\rangle = \sum_{\sigma} \langle T^\ast_{2;\lambda}T^\ast_1 R^\ast \psi, \mathcal{N}^{(\epsilon)}_{>,\sigma} T^\ast_{2;\lambda}T^\ast_1 R^\ast \psi\rangle \leq  CL^3\rho^{\frac{5}{3} - \frac{\epsilon}{3}}.
\]
Using Gr\"onwall's Lemma again (see \cite[Proposition 4.15, Corollary 4.16]{Gia1} for more details), we deduce that 
\begin{equation}\label{eq: N eps T1}
  \langle T^\ast_{1;\lambda} R^\ast \psi,\mathcal{N}^{(\epsilon)}_> T^\ast_{1;\lambda} R^\ast \psi\rangle \leq  \langle T^\ast_{1} R^\ast \psi,\mathcal{N}^{(\epsilon)}_> T^\ast_{1} R^\ast \psi\rangle + CL^3\rho^{\frac{5}{3}}.
\end{equation}
Combining \eqref{eq: N eps T2} with \eqref{eq: N eps T1}, we then get that 
\[
  \langle R^\ast \psi, \mathcal{N}^{(\epsilon)}_>R^\ast \psi\rangle = \langle \psi, \mathcal{N}^{(\epsilon)}_> \psi\rangle \leq CL^3\rho^{\frac{5}{3} -\frac{\epsilon}{3}},
\]
where in the first equality above we used the definition of the particle-hole transformation. This concludes the proof of Theorem \ref{thm: optimal number operator}.\\

\noindent\textbf{Acknowledgments.} The author thanks Christian Hainzl, Phan Thành Nam and Robert Seiringer for helpful discussions. Financial support from the Deutsche Forschungsgemeinschaft (DFG) -- Project-ID 470903074 -- TRR 352 is gratefully acknowledged. \\
 
\noindent\textbf{Data Availibility.} The author declares that all data supporting this article are available within the article.

\end{document}